\documentclass[a4paper,onecolumn,accepted=2022-08-24]{quantumarticle}
\pdfoutput=1
\usepackage[utf8]{inputenc}
\usepackage[affil-it]{authblk}
\usepackage{amsmath}
\usepackage{amssymb}
\usepackage{pifont}
\usepackage{graphicx}
\usepackage[colorlinks=true, allcolors=blue]{hyperref}
\usepackage{amsthm}
\usepackage[nottoc]{tocbibind}
\usepackage{braket}

\usepackage{tikz} 
\usetikzlibrary{arrows, shapes.gates.logic.US, calc}
\usetikzlibrary{quantikz}

\usepackage{verbatim}

\usepackage{hyperref}

\theoremstyle{definition}
\newtheorem{definition}{Definition}[section]
\newtheorem{theorem}{Theorem}[section]
\newtheorem{lemma}{Lemma}[section]

\newcommand{\Es}[1]{\textsc{E}_{#1}}
\newcommand{\supp}{\textsc{Supp}}
\newcommand{\Gen}{\textsc{Gen}}

\newcommand{\Dec}{\textsc{Dec}}
\newcommand{\Rrc}{\textsc{Rrc}}
\newcommand{\sX}{\mathcal{X}}
\newcommand{\sY}{\mathcal{Y}}

\newcommand{\Inv}{\ensuremath{\textsc{Inv}}}
\newcommand{\LWRInv}{\ensuremath{\textsc{LWRInv}}}
\newcommand{\Chk}{\ensuremath{\textsc{Chk}}}

\newcommand{\mZ}{\mathbb{Z}}

\newcommand{\dset}{G}
\newcommand{\TVD}{\operatorname{TVD}}

\newcommand{\biglfloor}{\left\lfloor}
\newcommand{\bigrfloor}{\right\rfloor}
\newcommand{\PrM}{\mathrm{Pr}_\mathrm{M}}

\newcommand{\EA}{\mathbb{E}_\mathbf{A}}

\newcommand{\Transform}{\ensuremath{\textsc{Transform}}}

\newcommand{\pinc}{p_\mathrm{inc}}
\newcommand{\pcor}{p_\mathrm{cor}}
\newcommand{\bpinc}{\bar{p}_\mathrm{inc}}
\newcommand{\bpcor}{\bar{p}_\mathrm{cor}}
\newcommand{\maj}{\mathrm{Maj}}

\newlength{\protowidth}
\newcommand{\pprotocol}[5]{
{\begin{figure*}[#4]
\begin{center}
\setlength{\protowidth}{\textwidth}
        {
        \hrulefill \vspace{5pt}
        \small
        {\quad
        \begin{minipage}{\protowidth}
        \begin{center}
        {\bf #1}
        \end{center}
        #5

        \hrulefill

        \end{minipage}
        \quad}
        }

        \caption{\label{#3} #2}
\end{center}
\vspace{-4ex}
\end{figure*}
} }

\newcommand{\protocol}[4]{
\pprotocol{#1}{#2}{#3}{tbh!}{#4} }

\title{Depth-efficient proofs of quantumness}
\author{Zhenning Liu}
\email{zhenliu@ethz.ch}
\affiliation{Department of Physics, ETH Z{\"u}rich, Switzerland}

\author{Alexandru Gheorghiu}
\email{agheorghiu@ethz.ch}
\affiliation{Institute for Theoretical Studies, ETH Z{\"u}rich, Switzerland}

\begin{document}

\sloppy

\maketitle

\begin{abstract}
A proof of quantumness is a type of challenge-response protocol in which a classical verifier can efficiently certify the \emph{quantum advantage} of an untrusted prover. That is, a quantum prover can correctly answer the verifier's challenges and be accepted, while any polynomial-time classical prover will be rejected with high probability, based on plausible computational assumptions.
To answer the verifier's challenges, existing proofs of quantumness typically require the quantum prover to perform a combination of polynomial-size quantum circuits and measurements.

In this paper, we give two proof of quantumness constructions in which the prover need only perform \emph{constant-depth quantum circuits} (and measurements) together with log-depth classical computation. 
Our first construction is a generic compiler that allows us to translate existing proofs of quantumness into constant quantum depth versions. Our second construction is based around the \emph{learning with rounding} problem, and yields circuits with shorter depth and requiring fewer qubits than the generic construction. In addition, the second construction also has some robustness against noise.
\end{abstract}

\newpage
\tableofcontents
\newpage

\section{Introduction}
Quantum computation is currently in the era of noisy intermediate-scale (NISQ) devices~\cite{preskill2018quantum}. This means that existing devices have a relatively small number of qubits (on the order of 100), perform operations that are subject to noise and are not able to operate fault-tolerantly. As a result, they are limited to running quantum circuits of small depth in order to obtain high fidelity outputs.
Despite these limitations, there have been a number of demonstrations of \emph{quantum computational advantage}~\cite{arute2019quantum, zhong2020quantum, wu2021strong, zhu2021quantum}, i.e. performing a task on a quantum device that cannot be \emph{efficiently} reproduced by classical computers, based on plausible complexity-theoretic assumptions~\cite{aaronson2011computational, harrow2017quantum, Bouland2019}. Indeed, with the best known classical algorithms it takes several days of supercomputing power to match the results of the quantum devices, which required only a few minutes to produce \cite{huang2020classical}.

These milestone results illustrate the impressive capabilities of existing quantum devices and highlight the potential of quantum computation in the near future. Yet, one major challenge still remains: how do we know whether the results from the quantum devices are indeed correct?    
For the existing demonstrations of quantum advantage, verification is achieved using various statistical tests on the output samples from the quantum devices~\cite{arute2019quantum, boixo2018characterizing, zhong2020quantum}. However, performing these tests either involves an exponential-time classical computation or there is no formal guarantee that an efficient classical adversary cannot \emph{spoof} the results of the test~\cite{aaronson2016complexity, aaronson2019classical, popova2021cracking}.

One conceptually simple way to demonstrate quantum advantage, that's also efficiently verifiable, is to ask the quantum computer to factor large composite integers using Shor's algorithm~\cite{shor1994algorithms}. Assuming factoring is classically intractable, this task yields a quantum advantage and is tractable to verify (simply multiply the output factors and check if they produce the number to be factored). However, Shor's algorithm requires fault-tolerant quantum computation to perform and so is not suitable for near-term devices~\cite{gidney2021factor}.

An alternative way of performing efficient tests of quantum advantage was initiated by the work of Brakerski et al.~in~\cite{brakerski2018cryptographic}. There, the authors proposed an interactive protocol between a polynomial-time classical \emph{verifier} and a self-claimed polynomial-time quantum \emph{prover}. The verifier issues a number of challenges to the prover and checks the prover's responses, accepting only when the prover answers the challenges correctly. The defining property of such a protocol is that no polynomial-time classical prover can make the verifier accept with high probability, but there exists a quantum polynomial-time strategy that makes the verifier always accept. This is referred to as a \emph{proof of quantumness} protocol. The protocol of Brakerski et al.~is based around a family of collision-resistant hash functions known as \emph{trapdoor claw-free functions} (TCFs)\footnote{Concurrently, Mahadev showed how TCFs can be used to perform classical verification of polynomial-time quantum computations~\cite{mahadev2018classical}.}. In essence, for the quantum prover to correctly answer the verifier's challenges, one of the things it is required to do is \emph{evaluate these functions in superposition}. With the trapdoor, the verifier is able to check whether the prover performed this evaluation correctly. It can also be shown that for any classical prover to succeed in the protocol, it would effectively have to find collisions for the TCFs. Brakerski et al.~showed that TCFs can be constructed assuming the intractability of the \emph{learning with errors} (LWE) problem~\cite{regev2009lattices}. In effect, this shows that efficient classical provers cannot succeed in the proof of quantumness, unless LWE is classically tractable. 
Subsequent works have also shown that TCFs can be based on other problems assumed to be classically intractable, such as factoring, the discrete logarithm problem or ring learning with errors~\cite{kahanamokumeyer2021classicallyverifiable}. Additionally, TCF-based proofs of quantumness can also be made non-interactive in the random oracle model~\cite{brakerski2020simpler}.
In all of these cases, however, to succeed in the protocol the ideal quantum prover must evaluate the TCFs coherently and this requires, at best, logarithmic quantum depth~\cite{gheorghiu2020estimating}.

It is thus the case that, on the one hand, we have statistical tests of quantum advantage that are suitable for NISQ computations but which either require exponential runtime or do not provide formal guarantees of verifiability. On the other hand, we have proofs of quantumness based on plausible computational assumptions, but that are not suitable for NISQ devices, as they require running deep quantum circuits. Is it possible to bridge the gap between the two approaches?
One step towards that goal would be to construct proofs of quantumness where the prover is only required to perform \emph{constant-depth} quantum circuits (together with short-depth classical circuits). This would also answer an important theoretical question: can one achieve quantum advantage with constant-depth quantum circuits while also being able to classically verify the results in polynomial time? 
This is the main result of our work: we give two proof of quantumness constructions in which the prover's evaluation can be performed in constant quantum depth and logarithmic classical depth. For the purposes of certifying quantum advantage, this leads to highly depth-efficient proofs of quantumness. Both constructions also yield depth-efficient protocols for \emph{certifiable randomness generation}, based on the scheme from~\cite{brakerski2018cryptographic}. The first construction is a generic compiler that can take existing proof of quantumness protocols, based on TCFs, and convert them into constant-depth versions. The second construction uses a specific TCF based on the \emph{learning with rounding} (LWR) problem~\cite{lwroriginal} and achieves circuits of smaller width and with some amount of noise robustness compared to the generic construction.

\subsection{Proofs of quantumness} \label{subsect:intropoq}
To explain our approach, we first need to give a more detailed overview of TCF-based proof of quantumness protocols. As the name suggests, the starting point is trapdoor claw-free functions. A TCF, denoted as $f$, is a type of 2-to-1 one-way function---a function that can be evaluated efficiently (in polynomial time) but which is intractable to invert.
The fact that the function is 2-to-1 means that there are exactly two preimages for each image of the function.
The function also has an associated trapdoor which, when known, allows for efficiently inverting $f(x)$, for any $x$. Finally, ``claw-free'' means that, without knowledge of the trapdoor, it should be intractable to find a pair of preimages, $x_0$, $x_1$, such that $f(x_0) = f(x_1)$. Such a pair is known as a \emph{claw}. 

For many of the protocols developed so far, an additional property is required known as the \emph{adaptive hardcore bit property}, first introduced in~\cite{brakerski2018cryptographic}. Intuitively, this says that for any $x_0$ it should be computationally intractable to find even a single bit of $x_1$, whenever $f(x_0) = f(x_1)$.
As was shown in~\cite{kahanamokumeyer2021classicallyverifiable}, this property is not required in order to construct proof of quantumness protocols, provided one adds an additional round of interaction in the protocol, as will become clear later. We will refer to TCFs having the adaptive hardcore bit property as \emph{strong TCFs}.
More formally, there exists $\lambda_0>0$, such that for any $\lambda>\lambda_0$, known as the \emph{security parameter}, a strong TCF, $f$, is a $2$-to-$1$ function which satisfies the following properties:
\begin{enumerate}
\item \textbf{Efficient generation.} There is a $poly(\lambda)$-time algorithm that can generate a description of $f$ as well as a trapdoor, $t \in \{0, 1\}^{poly(\lambda)}$.
\item \textbf{Efficient evaluation.} There is a $poly(\lambda)$-time algorithm for computing $f(x)$, for any $x \in \{0, 1\}^{\lambda}$.
\item \textbf{Hard to invert.} Any $poly(\lambda)$-time algorithm has \emph{negligible}\footnote{We say that a function $\mu(\lambda)$ is negligible if for any polynomial $p(\lambda)$, it is the case that $\lim_{\lambda \to \infty} p(\lambda) \mu(\lambda) = 0$.} probability to invert $y = f(x)$, for $x$ chosen uniformly at random from $\{0, 1\}^{\lambda}$.
\item \textbf{Trapdoor.} There is a $poly(\lambda)$-time algorithm that, given the trapdoor $t$, can invert $y = f(x)$, for any $x \in \{0, 1\}^{\lambda}$.
\item \textbf{Claw-free.} Any $poly(\lambda)$-time algorithm has negligible probability to find $(y, x_0, x_1)$, such that $y = f(x_0) = f(x_1)$, $x_0 \neq x_1$.
\item \textbf{Adaptive hardcore bit.} Any $poly(\lambda)$-time algorithm succeeds with probability negligibly close to $1/2$ in producing a tuple $(y, x_b, d)$, with $b \in \{0, 1\}$, such that
\begin{equation*}
y = f(x_0) = f(x_1), \quad\quad\quad d \cdot (x_0 \oplus x_1) = 0.
\end{equation*} 
\end{enumerate}
It should be noted that the properties, as stated here, are not independent of each other. For instance, property 6 implies properties 3 and 5 (and 5 also implies 3). We chose to present the properties this way for the sake of clarity.
Without the requirement of an adaptive hardcore bit, we recover the definition of an ordinary or regular TCF. Note that all $poly(\lambda)$-time algorithms mentioned above can be assumed to be classical algorithms.

We now outline the proof of quantumness protocol introduced in~\cite{brakerski2018cryptographic}. The classical verifier fixes a security parameter $\lambda > 0$ and generates a strong TCF, $f$, together with a trapdoor $t$. It then sends $f$ to the prover. The prover is instructed to create the state
\begin{equation} \label{eqn:coherentf}
\frac{1}{2^{\lambda/2}}\sum_{(b,x) \in \{0, 1\} \times \{0, 1\}^{\lambda - 1}} \ket{b, x}\ket{f(b, x)}
\end{equation}
and measure the second register, obtaining the result $y$. 
Note here that the input to the function was partitioned into the bit $b$ and the string $x$, of length $\lambda - 1$.
The string $y$ is sent to the verifier, while the prover keeps the state in the first register,
\begin{equation*} 
\frac{1}{\sqrt{2}}\left( \ket{0, x_0} + \ket{1, x_1} \right)
\end{equation*}
with $f(0, x_0) = f(1, x_1) = y$.
The string $y$ essentially \emph{commits} the prover to its leftover quantum state.

The verifier will now instruct the prover to measure this state in either the computational basis, referred to as the \emph{preimage test} or the Hadamard basis, referred to as the \emph{equation test}, and report the result. 
For the preimage test, the verifier simply checks whether the reported $(b, x_b)$ of the prover satisfies $f(b, x_b) = y$. For the equation test, the prover will report $(b', d) \in \{0, 1\} \times \{0, 1\}^{\lambda - 1}$ and the verifier checks whether
\begin{equation} \label{eq:eqtest}
d \cdot (x_0 \oplus x_1) = b'.
\end{equation}
In this case, the verifier has to use the trapdoor to recover both $x_0$ and $x_1$ from $y$ in order to compute Equation~\ref{eq:eqtest}.

It is clear that a quantum device can always succeed in this protocol by following the steps outlined above. However, the properties of the strong TCF make it so that no polynomial-time classical algorithm can succeed with high probability. At a high level, the reason for this is the following. Suppose a classical polynomial-time algorithm, $\mathcal{A}$, always succeeds in both the preimage test and the equation test. First, run $\mathcal{A}$ in order to produce the string $y$. Then, perform the preimage test with $\mathcal{A}$, resulting in $(b, x_b)$, such that $f(b, x_b) = y$. Since $\mathcal{A}$ is a classical algorithm, it can be \emph{rewound} to the point immediately after reporting $y$ and now instructed to perform the equation test. This will result in the tuple $(b', d)$ such that $d \cdot (x_0 \oplus x_1) = b'$. Importantly, $f(0, x_0) = f(1, x_1) = y$. We therefore have an efficient classical algorithm that yields both a hardcore bit for a claw as well as one of the preimages in the claw. As this contradicts the adaptive hardcore bit property, no such algorithm can exist.

As explained in~\cite{kahanamokumeyer2021classicallyverifiable, brakerski2020simpler, zhu2021demonstration}, the above argument can be made robust so that the success probabilities of any polynomial-time classical strategy in the two tests satisfy the relation
\begin{equation} \label{eq:winrel}
p_{pre} + 2 p_{eq} - 2 \leq \text{negl}(\lambda) 
\end{equation}
where $p_{pre}$ denotes the success probability in the preimage test, $p_{eq}$ is the success probability in the equation test and $\text{negl}(\lambda)$ is a negligible function in the security parameter $\lambda$.

The protocol described above crucially relies on the adaptive hardcore bit property to achieve \emph{soundness} against classical polynomial-time algorithms. Thus far, this property has only been shown for TCFs constructed from LWE~\cite{brakerski2018cryptographic}. It should also be noted that the above protocol is also a scheme for certifiable randomness generation: the bit $b$ obtained in the preimage test can be used as statistical randomness.

Is it possible to construct proof of quantumness protocols based on other computational assumptions than the classical intractability of LWE? Yes, in fact it is not difficult to see that simple proofs of quantumness can be based on the classical intractability of \emph{factoring} or the \emph{discrete logarithm problem} (DLP): ask the prover to solve multiple instances of these problems using Shor's algorithm~\cite{shor1994algorithms}. Since their solutions can be classically verified efficiently and since the problems are assumed to be classically intractable, this immediately yields a proof of quantumness. The issue with doing this is that the prover has to run large instances of Shor's algorithm, which would require a fault-tolerant quantum computer~\cite{gidney2021factor}.
Instead, as was shown in~\cite{kahanamokumeyer2021classicallyverifiable}, one can construct proofs of quantumness based on factoring or DLP, in which the prover can implement smaller circuits than those required for Shor's algorithm. Such protocols would then be more amenable to experimental implementation on near-term devices.

Let us briefly outline the approach in~\cite{kahanamokumeyer2021classicallyverifiable}. The idea is to consider TCFs that need not satisfy the adaptive hardcore bit property. Such TCFs can be constructed from more varied computational assumptions than LWE, including factoring, DLP or the ring-LWE problem~\cite{lyubashevsky2010ideal}. All of these are generally considered to be \emph{standard computational assumptions}. Having such a TCF, the protocol then proceeds in the same way as the one outlined above: the verifier requests that the prover prepare the state in~\ref{eqn:coherentf}, measure the function register obtaining the string $y$ and then send it to the verifier. The prover will be left with the state from~\ref{eqn:coherentf}. As before, the verifier will then instruct the prover to perform either a preimage test or an equation test. The preimage test is unchanged: the prover is asked to measure the state from~\ref{eqn:coherentf} in the computational basis and report back the result.

For the equation test, however, the verifier will first sample a random string $v \in \{0, 1\}^{\lambda}$ and send it to the prover. The prover must then prepare the state
\begin{equation*} 
\frac{1}{\sqrt{2}}\left( \ket{v \cdot x_0} \ket{x_0} + \ket{v \cdot x_1}\ket{x_1} \right)
\end{equation*}
The $x$ register is measured in the Hadamard basis, resulting in the string $d \in \{0, 1\}^{\lambda - 1}$ which is sent to the verifier. Upon receiving $d$, the verifier chooses a random $\phi \in \{ \pi/4, -\pi/4 \}$ and asks the prover to measure its remaining qubit in the rotated basis
\begin{equation*}
\left\{ 
\begin{array}{ll}
\cos \left( \frac{\phi}{2} \right) \ket{0} + & \sin \left( \frac{\phi}{2} \right) \ket{1} \\
\cos \left( \frac{\phi}{2} \right) \ket{1} - & \sin \left( \frac{\phi}{2} \right) \ket{0} \\
\end{array}
\right\}
\end{equation*}
Denoting as $b \in \{0, 1\}$ the prover's response, the verifier uses $d$ and the trapdoor to determine which $b$ is the likely outcome of the measurement and accepts if that matches the prover's response.

The last step in the protocol is reminiscent of the honest quantum strategy in the CHSH game for violating Bell's inequality~\cite{clauser1969proposed}. In fact, much like in the CHSH game, the success probability of any classical prover in this protocol is upper bounded by $0.75 + negl(\lambda)$, whereas a quantum prover can succeed with probability $\cos^2(\pi/8) \approx 0.85$. For this reason, the authors of~\cite{kahanamokumeyer2021classicallyverifiable} refer to the protocol as a \emph{computational Bell test}.

The soundness against classical polynomial-time algorithms follows from a similar rewinding argument to the one outlined for the previous protocol, which used a strong TCF. The main difference is that in this case the verifier introduces an additional challenge for the prover, in the form of the string $v$ and the bit $m$, from the modified equation test. This equation test is still checking for a hardcore bit of a claw, but unlike the previous protocol, the hardcore bit is no longer adaptive. Intuitively, this is because the verifier chooses which hardcore bit to request; a choice encapsulated by $v$ and $m$. For more details, we refer the reader to~\cite{kahanamokumeyer2021classicallyverifiable}.

\subsection{Our results}
In the proofs of quantumness outlined above, the honest quantum prover needs to coherently evaluate a TCF in order to pass the verifier's tests. A first step towards making the protocol depth-efficient would be to make it so that the prover can evaluate the TCF in constant quantum depth.
In fact, all that is required is for the prover to prepare the state from~\ref{eqn:coherentf} in constant depth, since the remaining operations can also be performed in constant depth. 
To that end, we first give a generic construction allowing the prover to prepare the state in~\ref{eqn:coherentf}, in constant depth, for all existing TCFs.
We then consider a second construction with a TCF based on the \emph{learning with rounding} (LWR) problem~\cite{lwroriginal} (a problem that is, for all intents, equivalent to LWE in terms of computational intractability) in which the prover will prepare a state that is essentially equivalent to that in~\ref{eqn:coherentf}. The advantage of this second construction is that the resulting circuits have smaller depth, smaller width (requiring fewer qubits) and have a certain degree of noise robustness, compared to the generic construction.
The first construction is presented in detail in Section~\ref{sect:poq}, while the second is in Section~\ref{sect:phase_enc}.

\subsubsection{First construction - A generic compiler}
We start with the observation from~\cite{gheorghiu2020estimating} that the strong TCFs based on LWE can be evaluated in classical logarithmic depth. In fact this also holds for the TCFs based on factoring, DLP and ring-LWE from~\cite{kahanamokumeyer2021classicallyverifiable}. 
As in~\cite{gheorghiu2020estimating}, one can then construct \emph{randomized encodings} for these TCFs, which can be evaluated by constant depth classical circuits.
A randomized encoding of some function, $f$, is another function, denoted $\hat{f}$, which is \emph{information-theoretically equivalent} to $f$. In other words, $f(x)$ can be uniquely and efficiently decoded from $\hat{f}(x, r)$, for any $x$ and for a uniformly random $r$. In addition, there is an efficient procedure for outputting $\hat{f}(x, r)$, given only $f(x)$. That is to say that $\hat{f}(x, r)$ contains no more information about $f(x)$ than $f(x)$ itself. The formal definition of randomized encodings is given in Subsection~\ref{subsect:re}.
It was shown in~\cite{AIK06} that all functions computable by log-depth circuits admit randomized encodings that can be evaluated in constant depth.
However, this doesn't immediately imply that a quantum prover can coherently evaluate these encodings in constant depth. The reason is that these circuits will typically use gates of \emph{unbounded fan-out}. These are gates that can create arbitrarily-many copies of their output. But the gate set one typically considers for quantum computation has only gates of bounded fan-out (single-qubit rotations and the two-qubit $CNOT$, for instance). How then can the prover evaluate the randomized encoding in constant depth with gates of bounded fan-out?

The key observation is that we do not require the prover to be able to evaluate $\hat{f}$ coherently on an arbitrary input, merely on \emph{a uniform superposition over classical inputs}. One of our main results is then the following:
\begin{theorem} [informal] \label{thm:constdepth}
There is a strategy consisting of alternating constant depth quantum circuits and logarithmic-depth classical circuits for preparing the state:
\begin{equation}
\sum_{x} \ket{x} \ket{\hat{f}(x)},
\end{equation}
up to an isometry, for any $\hat{f}$ that can be evaluated by a constant-depth classical circuit, potentially including unbounded fan-out gates.
\end{theorem}
To prove this result, we use an idea from the theory of \emph{quantum error-correction}. It is known that cat states (also known as GHZ states) cannot be prepared by a fixed constant-depth quantum circuit~\cite{watts2019exponential}. However, if we can interleave short-depth quantum circuits (and measurements) with classical computation, it is possible to prepare cat states in constant quantum-depth. This is akin to performing corrections in quantum error correction, based on the results of syndrome measurements.

In our case, this works as follows.
First, prepare a \emph{poor man's cat state} in constant depth, as described in~\cite{watts2019exponential}. This is a state of the form
\begin{equation*}
X(w) \; \frac{\ket{0}^{\otimes n} + \ket{1}^{\otimes n}}{\sqrt{2}}
\end{equation*}
where $w$ is a string in $\{0, 1\}^n$ and
\begin{equation*}
X(w) = X^{w_1} \otimes X^{w_2} \otimes ... \otimes X^{w_n},
\end{equation*}
with $X$ denoting the Pauli-$X$ qubit flip operation.
As explained in~\cite{watts2019exponential}, the constant-depth preparation of the poor man's cat state involves a measurement of the parities of neighboring qubits. In other words, the measurement yields the string $z \in \{0, 1\}^{n-1}$, with $z_i = w_{i} \oplus w_{i+1}$, for $i \in [n-1]$. Using a log-depth classical circuit, this parity information can be used to determine either $w$ or its binary complement. One then applies the correction operation $X(w)$ to the poor man's cat state, thus yielding the desired cat state
\begin{equation*}
\frac{\ket{0}^{\otimes n} + \ket{1}^{\otimes n}}{\sqrt{2}}.
\end{equation*}
Having multiple copies of cat states, it is possible to replicate the effect of unbounded fan-out classical gates on a uniform input\footnote{We attribute this idea, of replicating unbounded fan-out with constant-depth quantum circuits and classical measurements, to folklore.}. To see why, consider the following example. Suppose we have a classical AND gate, having fan-out $n$. On inputs $a, b \in \{0, 1\}$, it produces the output $c \in \{0, 1\}^n$, with $c_i = a \land b$, for all $i \in [n]$. To perform the same operation with bounded fan-out gates, it suffices to have $n$ copies of $a$ and $b$. That is, given $\bar{a}, \bar{b} \in \{0, 1\}^n$, with $\bar{a}_i = a$, $\bar{b}_i = b$, for all $i \in [n]$, one can compute $c_i = \bar{a}_i \land \bar{b}_i$ using $n$ parallel AND gates. This is illustrated in Figure~\ref{fig:fanoutb}.

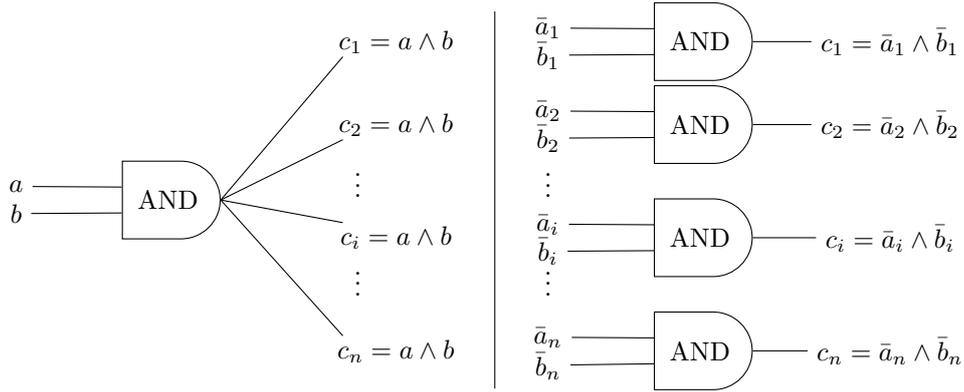
\begin{figure}
    \centering
    
\begin{tikzpicture}
    \node (c1) at (5, 4.1) {$c_1=a\land b$};
    \node (c2) at (5, 3) {$c_2=a\land b$};
    \node (c3d) at (4.5, 2.3) {$\vdots$};
    \node (ci) at (5,1.5) {$c_i=a\land b$};
    \node (c4d) at (4.5, 1) {$\vdots$};
    \node (cm) at (5, 0) {$c_n=a\land b$};
    
    \node[and gate US, draw, rotate=0, logic gate inputs=nn] at (2,2) (and0) {AND};
    
    \node (a) at (0,2.18) {$a$};
    \node (b) at (0,1.82) {$b$};
    
    \draw (a) --  (and0.input 1);
    \draw (b) --  (and0.input 2);
    \draw (and0.output) --  (4.3,3.9);
    \draw (and0.output) --  (4.3,2.8);
    
    \draw (and0.output) --  (4.3,0.2);
     \draw (and0.output) --  (4.3,1.7);
    
    \draw (6.3,-0.5)--(6.3,4.5);

    \node (a1) at (7,4.28) {$\bar{a}_1$};
    \node (b1) at (7,3.92) {$\bar{b}_1$};
    \node[and gate US, draw, rotate=0, logic gate inputs=nn] at (9,4.1) (and1) {AND};
    \draw (a1) -- (and1.input 1);
    \draw (b1) -- (and1.input 2);
    \node (d1) at (11.5, 4.1) {$c_1=\bar{a}_1\land \bar{b}_1$};
    \draw (and1.output) -- (d1);

    \node (a2) at (7,3.18) {$\bar{a}_2$};
    \node (b2) at (7,2.82) {$\bar{b}_2$};
     \node[and gate US, draw, rotate=0, logic gate inputs=nn] at (9,3) (and2) {AND};
    \draw (a2) -- (and2.input 1);
    \draw (b2) -- (and2.input 2);
    \node (d2) at (11.5, 3) {$c_2=\bar{a}_2\land \bar{b}_2$};
    \draw (and2.output) -- (d2);
    
    \node (a3d) at (7,2.3) {$\vdots$};
    \node (ai) at (7,1.68) {$\bar{a}_i$};
    
    \node (bi) at (7,1.32) {$\bar{b}_i$};
    \node (b4d) at (7,1) {$\vdots$};

    \node[and gate US, draw, rotate=0, logic gate inputs=nn] at (9,1.5) (andi) {AND};
    \draw (ai) -- (andi.input 1);
    \draw (bi) -- (andi.input 2);
    \node (di) at (11.5, 1.5) {$c_i=\bar{a}_i\land \bar{b}_i$};
    \draw (andi.output) -- (di);
    
    \node (am) at (7,0.18) {$\bar{a}_n$};
    \node (bm) at (7,-0.19) {$\bar{b}_n$};
     \node[and gate US, draw, rotate=0, logic gate inputs=nn] at (9,0) (andm) {AND};
    \draw (am) -- (andm.input 1);
    \draw (bm) -- (andm.input 2);
    \node (dm) at (11.5, 0) {$c_n=\bar{a}_n\land \bar{b}_n$};
    \draw (andm.output) -- (dm);
\end{tikzpicture}
    \caption{The left-hand side shows an AND gate with fan-out $n$. The right-hand side is its bounded fan-out equivalent. Here $\bar{a}_i=a$ and $\bar{b}_i=b$. Gates of unbounded fan-out can be implemented with bounded fan-out as long as sufficient copies of the inputs are provided.}
    \label{fig:fanoutb}
\end{figure}

In our case, each input qubit to the classical function is of the form $\frac{1}{\sqrt{2}}(\ket{0} + \ket{1})$. Replacing it with $n$ copies is equivalent to using a cat state $\frac{1}{\sqrt{2}}(\ket{\bar{0}} + \ket{\bar{1}})$, where $\ket{\bar{0}} = \ket{0}^{\otimes n}$, $\ket{\bar{1}} = \ket{1}^{\otimes n}$. As mentioned, the prover can prepare cat states in constant depth using the ``measure-and-correct'' trick. It then follows that the prover can also prepare the state
\begin{equation} \label{eqn:ghzeval}
\sum_{x} \ket{\bar{x}} \ket{f(x)}
\end{equation}
where each bit of $x$ is encoded as a cat state having the same number of qubits as the number of input copies required to evaluate $f$ with bounded fan-out gates.

With the ability to prepare the state from~\ref{eqn:coherentf} (or one equivalent to it, such as the one from~\ref{eqn:ghzeval}) in constant quantum depth, the honest prover can then proceed to perform the rest of the steps in the proof of quantumness protocols outlined above. It will measure the image register and report the result to the verifier. The remaining operations can also be performed in constant depth. For the preimage test, the prover simply measures the $x$ register in the computational basis and reports the result. For the equation test, the prover needs to first apply a layer of Hadamard gates to the $x$ register before measuring it in the computational basis. Lastly, for the Bell-type measurement required in the protocol of~\cite{kahanamokumeyer2021classicallyverifiable}, a slightly more involved procedure is used to perform the measurement in constant depth. All of these steps are described in detail in Subsection~\ref{subsect:completeness}. 

While we have outlined a procedure for the prover to perform its operations in constant quantum depth, using a randomized encoding of a TCF, it is not immediately clear if we need to also modify the verifier's operations. Indeed, one question that is raised by this approach is whether a randomized encoding of a TCF preserves all the properties of a TCF. If, for instance, the trapdoor property is not preserved, the verifier would be unable to check the prover's responses in the equation test. Our second result resolves this issue:
\begin{theorem}[informal] \label{thm:tcfre}
A randomized encoding of a (strong) TCF is a (strong) TCF.
\end{theorem}
This theorem implies that substituting the TCFs used in proofs of quantumness with randomized encodings will not affect the soundness of those protocols. 
The proof can be found in Subsection~\ref{subsect:soundness}.
A similar result was derived in~\cite{AIK06}, where the authors show that randomized encodings of cryptographic hash functions are also cryptographic hash functions. A (strong) TCF is different, however\footnote{A TCF has exactly two collisions for each image, it has a trapdoor and strong TCFs additionally have the adaptive hardcore bit property. None of these properties are satisfied by generic cryptographic hash functions.}. 
To prove this result, first note that most of the TCF properties follow almost immediately from the definition of a randomized encoding. The more challenging parts concern the existence of a trapdoor and the adaptive hardcore bit property. To show these, we require that the randomized encoding satisfies a property known as \emph{randomness reconstruction}~\cite{AIK06}. This states that whenever there is an efficient procedure to invert the original function, $f$, there should also be an efficient procedure for inverting $\hat{f}$. In particular, this means that given $\hat{f}(x, r)$ it is possible to recover both the input $x$ and the randomness $r$. 
In~\cite{AIK06}, it's mentioned that the randomized encodings used to ``compress'' functions to constant depth do satisfy the randomness reconstruction property, but no proof is given. We provide a proof in Appendix~\ref{app:rrc}.

With the two results of Theorems~\ref{thm:constdepth} and~\ref{thm:tcfre}, we have that any proof of quantumness using a log-depth computable TCF can be compiled to constant quantum depth for the prover. All of the results for this construction are presented in Section~\ref{sect:poq}, and in Subection~\ref{subsect:resources} we give a detailed account of the resources required for the prover to perform this evaluation.

\subsubsection{Second construction - Phase encoding and learning with rounding} \label{subsect:introsecondpoq}
The second solution to the problem comes from an attempt to directly parallelize the coherent evaluation of the TCF based on LWE, hence to implement the protocol in \cite{brakerski2018cryptographic} in constant quantum depth. We start with the observation that the TCF based on LWE contains only mod-$q$ matrix multiplication and mod-$q$ vector addition operations, where $q\in \mathbb{N}$ is the field size. Since the phases of quantum states have the same periodicity property as the ``mod-$q$'' operation, it is natural to consider implementing the mod-$q$ arithmetic with phase $Z$-rotations ($R_z$ and Controlled-$R_z$ gates). In the standard basis, the $R_z$ operation is expressed as
\begin{equation*}
    R_z(\theta) = 
    \left(\begin{matrix}
    e^{-i\frac{\theta}{2}}&0\\
    0&e^{i\frac{\theta}{2}}
    \end{matrix}\right)
\end{equation*}
Note that, for a given cat state, $\ket{\psi} = \frac{1}{\sqrt{2}}( \ket{\bar{0}} + \ket{\bar{1}})$, applying two $R_z$ phase rotations on \emph{distinct qubits} results in the phases being added into the relative phase of the state. Specifically, if we were to rotate qubit $i$ by $\theta_i$ and qubit $j$ by $\theta_j$ we would obtain
\begin{equation*}
    R_z(\theta_i) R_z(\theta_j) \ket{\psi} = \frac{1}{\sqrt{2}}( \ket{\bar{0}} + e^{i(\theta_i + \theta_j)} \ket{\bar{1}})
\end{equation*}
By taking $\theta_i = \frac{2 \pi a}{q}$ and $\theta_j = \frac{2 \pi b}{q}$, with $a, b \in \mathbb{Z}_q$, we can see that the net effect is a state with a relative phase proportional to $(a + b) \mod q$,
\begin{equation}\label{eqn:phaseencodstate}
    \frac{1}{\sqrt{2}}( \ket{\bar{0}} + e^{\frac{2 \pi i (a + b)}{q}} \ket{\bar{1}})
\end{equation}
The key idea is that because these operations commute, \emph{they can be implemented in parallel by acting on distinct qubits}, yielding a constant depth circuit for performing mod-$q$ arithmetic in phase.
We denote the state in Equation~\ref{eqn:phaseencodstate} as $\ket{\phi(a+b)}$\footnote{Strictly speaking the notation will refer to states with a relative phase of $\frac{2 \pi i (a + b)}{q} - \frac{\pi}{2}$, for reasons that will become clear later. Additionally, when using this notation we will always assume the phases are multiples of the $q$'th roots of unity as in the example outlined above.} and refer to it as a \emph{phase encoding of $a+b$}. 
Encoding the values of the LWE-based TCF in phase seems to introduce a problem for the protocol. Recall that in the standard proof of quantumness protocol (outlined in Subsection~\ref{subsect:intropoq}) the prover encodes evaluations of the function $f$ in the computational basis. If these values were instead encoded in phase, how would the prover be able to obtain an evaluation, $y$, of the function?

To overcome this obstacle, we consider a different TCF based on a problem known as learning with rounding (LWR) \cite{lwroriginal, lwr_revisited}. This problem is equivalent to LWE (for most parameter choices) and was already suggested as a candidate for building TCFs in~\cite{brakerski2018cryptographic}. Specifically, denoting now as $f$ an LWR-based TCF, we take
\begin{equation}
\label{eq:lwrfirst}
f(b,x):\{0,1\}\times \mathbb{Z}_q^n \rightarrow \mathbb{Z}_p^m = \lfloor \mathbf{A}x+b \cdot (\mathbf{A}s+e) \rfloor_p
\end{equation}
where $\mathbf{A}\in\mathbb{Z}_q^{m\times n}$, $b\in\{0,1\}$, $x, s \in\mathbb{Z}_q^n$ and $e\in\mathbb{Z}_p^m$ are vectors and $\lfloor \cdot \rfloor_p$ denotes rounding over $p$. By rounding we mean taking the most significant $\log_2 p$ bits of the result\footnote{In fact this is only true when $q=2^n$. In our case $q$ will be prime and so the rounding operation, for some value $\alpha \in \mathbb{Z}_q$, is defined as $\lfloor \frac{p}{q} \cdot \alpha \rfloor$.}. In this case, the result is a vector and the rounding is performed component-wise, so that the output is a vector with entries in $\mathbb{Z}_p$.
Note that all matrix multiplications and additions are performed modulo $q$ with $q\gg p$. Intuitively, for small values of $e$, a typical \emph{claw} of the function should be $(0,x)$ and $(1,x-s)$. This is due to the fact that the rounding operation takes the most significant bits of the output, which are unlikely to be changed when adding a vector $e$ with small entries, component-wise.
We refer the reader to the preliminaries in Section~\ref{sect:prelim} for a more detailed explanation of the function and its parameters.

Returning to the idea of the phase encoding, we can now begin to see the reason for choosing this LWR-based function. Consider for the moment the function before rounding,
\begin{equation*}
    g(b, x) = \mathbf{A}x+b \cdot (\mathbf{A}s+e).
\end{equation*}
Suppose we were to perform a phase encoding of the entries of this function, which we denote as $\ket{\phi(b, x)}$. Now take the $i$'th entry of that encoding, $\ket{\phi_i(b, x)}$ which encodes the $i$'th component of $g(b, x)$, denoted $g_i(b, x)$. It is not difficult to see that if we were to measure $\ket{\phi_i(b, x)}$ in the Hadamard basis (or in this case, measure the operator $XX...X$, as we have a rotated cat state), the outcome is most likely to be the \emph{most significant bit} of $g_i(b, x)$. Similarly, if in the phase encoding we used the $q/2$ roots of unity, instead of the $q$ roots of unity, a Hadamard measurement of the encoding would likely yield the second most significant bit. Repeating this $\log_2 p$ times we have a way of probabilistically recovering the output $f_i(b, x) = \lfloor g_i(b, x) \rfloor_p$. Of course, due to the probabilistic nature of the measurement, the chance that all bits are recovered correctly will be small. To remedy this issue, we use a classical repetition code. In other words, we view each component of $g(b, x)$ as being repeated several times. When the prover eventually performs its measurements to recover $f(b, x)$ it will take a majority vote for each component.
We find that by choosing a suitably large number of repetitions we can make it so that the prover succeeds in evaluating $f(b, x)$ in this way with overwhelming probability.

Our main result is then the following:
\begin{theorem}[informal]
A proof of quantumness protocol, with constant quantum depth and logarithmic depth classical computation, can be constructed based on LWR.
\end{theorem}

To prove this result, we first need to show that the function $f$ indeed satisfies the properties of a strong TCF. The formal proof of this fact can be found in Subsection~\ref{sec:lwr-ntcf}, which is mainly about showing the adaptive hardcore bit property, as all other properties are fairly straightforward. 

We next discuss the protocol itself, which is essentially unchanged from that of~\cite{brakerski2018cryptographic}, except that it uses the LWR-based TCF. Additionally, what changes will be the prover's honest strategy for coherently evaluating this TCF. As mentioned, for this rounding-based function it is possible to coherently evaluate the function in phase, leading to a state that is equivalent (up to an isometry) to
\begin{equation*}
    \sum_{b,x}\ket{b}_\mathsf{B}\ket{x}_\mathsf{X}\ket{\phi(b,x)}_\mathsf{Z}.
\end{equation*}
To ensure that all mod-$q$ operations, required to prepare this state, can be performed in parallel, the cat states that serve as the basis for the phase encoding must have $\Omega(n\log q)$ qubits. Here, $n$ represents the $n$ rows of the matrix $\mathbf{A}$ and since each component is modulo $q$, this also contributes a multiplicative $\log q$ factor. As mentioned, we also need to repeat each component in order to guarantee that measurements of the phase-encoded $\mathsf{Z}$ register yield a valid image with high probability. We find that the number of repetitions must be $\Omega(n^4 \log^2 n)$ to have a small probability of incorrectly decoding from measurement.

Lastly, we show that the state in the preimage registers, $\mathsf{BX}$, has high overlap with a superposition of preimages, as in the standard version of the protocol.
The proof of this fact is based on the observation that while the states $\ket{\phi(b,x)}$ and $\ket{\phi(b',x')}$ are not exactly orthogonal whenever $((b, x), (b', x'))$ does not constitute a claw, they are sufficiently close to orthogonal for \emph{most choices} of the matrix $\mathbf{A}$. More specifically, we can show that if $\mathbf{A}$ is uniformly sampled\footnote{Strictly speaking, $\mathbf{A}$ will not be uniform as one needs to sample a matrix $\mathbf{A}$ for which a trapdoor is known, in order to construct an STCF. However, as explained in~\cite{brakerski2018cryptographic}, the matrix is sampled from a distribution that is statistically close to uniform.} from $\mathbb{Z}_q^{m \times n}$, the overlap between distinct $\ket{\phi(b, x)}$ states decays exponentially in $m$. On the other hand, if $((b, x), (b', x'))$ does form a claw, we can show that the overlap of $\ket{\phi(b,x)}$ and $\ket{\phi(b',x')}$ is negligibly close to $1$. From these facts and the trace-preserving nature of the operations involved, it follows that the state in the preimage register will have high overlap with a superposition of preimages, upon the prover measuring the image register, $\mathsf{Z}$.

An important observation about this construction is that it requires one to perform phase rotations in increments of $\frac{2\pi}{q}$. While such rotation operations are already native to most existing quantum computing architectures, it is also possible to use a constant-size gate set at the expense of making the circuit polynomially wider. This is achieved by approximating the rotation gates to within inverse-polynomial error through the repetition of a fixed set of rotations (see Remark 3.5 in~\cite{quantumfanout}). 

Our second construction is thus an instantiation of the protocol in~\cite{brakerski2018cryptographic} with an LWR-based TCF and having the prover perform a phase-encoded evaluation of that function. The main appeal of this construction is that it is much simpler than the generic construction from the previous section and achieves circuits with fewer qubits. Specifically, as computed in Subsections~\ref{subsect:resources} and~\ref{subsect:resources2}, for a security parameter $\lambda > 0$, the generic construction uses $O(\lambda^{33})$ qubits, whereas the LWR-based one uses $O(\lambda^8 \log^3 \lambda)$.
Additionally, the use of the repetition code and the error-correcting properties of LWR offer the scheme some level of robustness against noise.
For the full details and proofs related to this construction, see Section~\ref{sect:phase_enc}.

\subsection{Related work}
One of the first efficient computational tests of quantum advantage was proposed in~\cite{shepherd2009temporally}, for certifying that a quantum prover can perform \emph{instantaneous quantum polynomial-time computations} (\textsf{IQP}). However, that test was based on a non-standard hardness assumption and it was later shown that there is an efficient classical algorithm which passes the test~\cite{kahanamoku2019forging}.

The first proof of quantumness based on LWE originated with the work of Brakerski et al.~\cite{brakerski2018cryptographic}. This is the proof of quantumness based on a strong TCF outlined in the introduction. As explained there, the protocol also serves as a certifiable random number generator. A subsequent work achieved a non-interactive version of this protocol in the \emph{quantum random-oracle model}~\cite{brakerski2020simpler}. Notably, in that protocol the adaptive hardcore bit property is not required, however the protocol does make use of a hash function (in addition to the TCF) modeled as a random oracle.

The second proof of quantumness we outlined, based on regular TCFs, was introduced in~\cite{kahanamokumeyer2021classicallyverifiable}. There the authors achieve more efficient proofs of quantumness by removing the requirement of the adaptive hardcore bit and using TCFs having a lower circuit complexity compared to the ones based on LWE. However, as mentioned, the cost of doing this is introducing additional rounds of interaction between the verifier and the prover (in the form of the Bell-like measurement of the equation test).

In terms of constant quantum depth constructions, it is interesting to contrast our work to that of~\cite{coudron2021trading}. There, the authors proposed a protocol for certifiable random-number generation with constant depth quantum circuits. The first difference with respect to our work is that~\cite{coudron2021trading} do not base the soundness of their protocol on the classical intractability of some computational problem, such as LWE. Instead, the protocol assumes that the ``prover'' generating the randomness is a circuit of sub-logarithmic depth (showing that sub-logarithmic classical circuits would not succeed in this task). The second difference is that our protocols require interleaving constant depth quantum circuits with logarithmic depth classical computation, whereas the protocol in~\cite{coudron2021trading} only requires the application of a constant depth quantum circuit. Finally, our protocols are interactive, whereas~\cite{coudron2021trading} is not.  

We also mention the independent work of Hirahara and Le Gall that appeared before ours and which also gives a constant-depth proof of quantumness \cite{hirahara2021test}. Similar to our work, they also considered one of the existing proofs of quantumness and made it so that the prover could perform its operations in constant quantum depth and using log-depth classical computations. In their case, they use a technique inspired from measurement-based quantum computing to have the prover perform the coherent evaluation of the strong TCF based on LWE.
Notably, their prover evaluates that function in the computational basis, unlike our LWR-based scheme which performs the evaluation in phase.

Lastly, we also point out the work of H{\o}yer and {\v S}palek showing that a large class of quantum algorithms can be implemented in constant depth with quantum gates of unbounded fan-out~\cite{quantumfanout}. In particular, the quantum subroutine of Shor's algorithm can be performed this way. It should then be possible to use the same trick of reproducing unbounded fan-out with bounded fan-out gates, through measurements and classical corrections, as we did for both our constructions. This would then yield a factoring algorithm that uses only constant depth quantum circuits. There are however two downsides to doing this, compared to our approach. First, the resulting algorithm would use classical circuits of supra-logarithmic depth (see also~\cite{cleve2000fast} for a discussion of this point), in contrast to the logarithmic depth circuits that we obtain~\cite{francois}. Second, the resulting circuits for factoring would be significantly larger compared to the circuits obtained in our constructions.

\subsection{Discussion and open problems}
We've shown how existing proof of quantumness protocols can be made to work with a prover that performs constant-depth quantum computations and log-depth classical computations. Thus, all protocols based on TCFs can be compiled to constant-depth versions using randomized encodings and preparations of cat states.

One potential objection to our result is the practicality of this construction. The prover must not only run constant-depth quantum circuits, but it must do so based on the outcomes of previous measurements or based on instructions from the verifier. This is similar to syndrome measurements and corrections in quantum error-correcting codes and so it might seem as if the prover must have the capability of doing fault-tolerant quantum computations.
In fact this is not the case. For the protocols based on strong TCFs the number of quantum-classical \emph{interleavings} --- that is, the number of alternations between performing a constant depth quantum circuit followed by a log-depth classical circuit --- is exactly three. The first is required for the preparation of cat states. In this case, the prover simply needs to apply $X$ corrections conditioned on the outcomes of certain parity measurements. The prover then evaluates the randomized-encoded TCF and measures one of its registers, sending that result to the verifier. Conditioned on its response it either measures the remaining state in the computational basis or in the Hadamard basis. Similar operations are performed for the LWR-based construction.
The prover, therefore, needs to do only a very restricted type of conditional operations and is only required to do this three times. Furthermore, the protocol is robust and some degree of noise is acceptable, provided Inequality~\ref{eq:winrel} is violated.
When using regular TCFs, in the generic compilation scheme, the protocol requires two additional quantum-classical interleavings, for a total of five. This is due to the Bell-like measurement of that protocol.
In both cases, only a small number of quantum-classical interleavings are required, unlike in a fully fault-tolerant computation where many such interleavings would be required~\cite{fowler2012surface}.  

It would, of course, be desirable to have a single-round proof of quantumness with a constant-depth prover and no quantum-classical interleavings. In other words, a protocol in which the prover has to run a single constant-depth quantum circuit and the verifier is able to efficiently certify that the prover is indeed quantum. Such a result would yield a weak separation between polynomial-time classical computation and constant-depth quantum computation. Basing such a separation on just the classical hardness of LWE seems unlikely\footnote{See the first paragraph of the ``Our results'' subsection in~\cite{brakerski2020simpler}.}. Basing it on the classical intractability of factoring or DLP seems more realistic, as those assumptions already yield a separation between polynomial-time classical computation and logarithmic-depth quantum computation~\cite{cleve2000fast}.
However, it is unclear how to adapt the existing protocols which rely on this commit-and-test approach that requires at least two rounds of interaction.
We leave answering this question as an interesting open problem.

Finally, the computational resources required to implement our constant-depth proofs of quantumness are still too high for existing quantum devices.
In particular, the resulting quantum circuits can be prohibitively wide to be implemented on existing NISQ devices. However, as we've seen, different implementations can lead to very different qubit requirements. Rough estimates show that our generic construction requires $O(\lambda^{33})$ qubits, while the LWR-based one requires $O(\lambda^8 \log^3 \lambda)$. These substantially different estimates give us some hope that further reducing the qubit requirements is possible. Additional optimizations are likely also possible when considering specific values for the security parameter and the choice of TCF. We therefore also leave as an open problem to reduce the width of these constructions so as to make the protocols better suited for use on near-term devices.

\section*{Acknowledgements}
AG is supported by Dr. Max R\"ossler, the Walter Haefner Foundation and the ETH Z\"urich Foundation.

\section{Preliminaries} \label{sect:prelim}

\subsection{Notation and basic concepts}
We let $\mathbb{N}$ denote the set of natural numbers, $\mathbb{Z}$ the set of integers, $\mathbb{Z}_q$ the set of integers modulo $q$, and $\mathbb{R}$ the set of real numbers.
The set $\{0, 1\}^n$ denotes all binary strings of length $n$.
For some binary string $v \in \{0, 1\}^n$, the $i$'th bit of $v$ is denoted $v_i$ (with $1 \leq i \leq n$). 
We denote as $|v|$ the \emph{Hamming weight} of $v$, which is defined as the number of 1's in $v$, or
\begin{equation*}
|v| = \sum_{i=1}^{n} v_i.
\end{equation*} 
The $xor$ of two bits $a$, $b$ is $a \oplus b = a + b \; mod \; 2$. This extends to strings so that for $v, w \in \{0, 1\}^n$, $v \oplus w$ is their bitwise xor. 
The \emph{Hamming distance} of the strings $v$ and $w$ is then defined as:
\begin{equation*}
d_H(v, w) = |v \oplus w|
\end{equation*}
We will also make use of the bitwise inner product of two strings, defined as:
\[
v \cdot w = \sum_{i = 1}^{n} v_i \cdot w_i \; mod \; 2.
\]
For a bit $b \in \{0, 1\}$, we will use $\bar{b}$ to denote a binary string consisting of \emph{copies of} $b$. That is, $\bar{b} = bbb...b$. The number of copies will generally be clear from the context and will otherwise be specified. We also extend this notation to binary strings. For some string $v \in \{0, 1\}^n$, $\bar{v}$ will denote a string in which each bit of $v$ has been repeated. That is, $\bar{v} = v_1 v_1 ... v_1 v_2 ... v_2 v_3 ... v_{n-1} v_n ... v_n$.

For any finite set $X$, we let $x\leftarrow_r X$ denote an element drawn uniformly at random from $X$.
The \emph{total variation distance} between two density functions $f_1, f_2 : X \to [0, 1]$ is $$\TVD(f_1,f_2)=\frac{1}{2}\sum_{x\in X}|f_1(x)-f_2(x)|.$$

For an element $r \in \mathbb{Z}_q$, its unique representative will be $[r]_q \in ( -q/2, q/2 ) \cap \mathbb{Z}$. Following~\cite{brakerski2018cryptographic}, we use the notation $|r| = |[r]_q|$. For any vector $v$ of $n$ components, its $l^2$-norm is defined as
$$
||v||_2 = \sqrt{\sum_{i=1}^{n} |v_i|^2},
$$
and its $l^{\infty}$ norm is
$$
||v||_\infty = \max_i(|v_i|).
$$

The \emph{Hellinger distance} between $f_1$ and $f_2$ is
$$H^2(f_1,f_2)=1-\sum_{x\in X}\sqrt{f_1(x)f_2(x)}.$$

For any discrete probability distribution $p(x)$, its \emph{support} is defined as the set of points where the distribution is positive, $\supp(p(x)) = \{x: \; p(x) > 0 \}$.

For a positive $B\in \mathbb{R}$ and positive integer $q$, the truncated discrete Gaussian distribution over $\mathbb{Z}_q$ with parameter $B$ is supported on $\{x\in\mathbb{Z}_q:\,\|x\|\leq B\}$ and has density 
\begin{equation}\label{eq:d-bounded-def}
 D_{\mZ_q,B}(x) \,=\, \frac{e^{\frac{-\pi\lVert x\rVert^2}{B^2}}}{\sum\limits_{x\in\mZ_q,\, \|x\|\leq B}e^{\frac{-\pi\lVert x\rVert^2}{B^2}}} \;.
\end{equation}

 We let $negl(x)$ denote a \emph{negligible function}. A function $\mu:\mathbb{N}\rightarrow\mathbb{R}$ is negligible if for any positive polynomial $p(x)$ there exists an integer $N > 0$ such that for all $x>N$ it's the case that
$$|\mu(x)|<\frac{1}{p(x)}.$$

We sometimes abbreviate polynomial functions as $poly$.
Throughout the paper, $\lambda$ will denote the \emph{security parameter}.
This will be polynomially-related to the input size of all functions we consider. Consequently, all polynomial and negligible functions will scale in $\lambda$.

Let $\{D_\lambda \}_{\lambda\in\mathbb{N}}$ and $\{E_\lambda \}_{\lambda\in\mathbb{N}}$ be two families of probability distributions defined on $\{0,1\}^\lambda$. They are \emph{computationally indistinguishable} if for every polynomial-time algorithm $\mathcal{A}:\{0,1\}^\lambda \rightarrow \{0,1\}$, it is the case that
\begin{equation*}
    |\Pr_{x\leftarrow D_\lambda}(\mathcal{A}(x)=0) - \Pr_{x\leftarrow E_\lambda}(\mathcal{A}(x)=0)| = negl(\lambda).
\end{equation*}

Letting $g_i\in\mathbb{Z}_q$ with $q \geq 2$, the (mod-$q$) \emph{phase encoding} of $g_i$ is defined as 
\begin{equation}
    \ket{\phi_i} = \frac{1}{\sqrt{2}}(\ket{0}+e^{i\phi_i}\ket{1})
\end{equation}\
where 
\begin{equation}
\phi_i = \frac{2\pi g_i}{q} - \frac{\pi}{2}.
\end{equation}

In terms of quantum information, we follow the usual formalism as outlined, for instance, in~\cite{nielsen2002quantum}.
All Hilbert spaces are finite dimensional. We use sans-serif font to label spaces that correspond to certain quantum registers. For instance, $\mathsf{X}$ will correspond to an $n$-qubit Hilbert space of inputs to a function.
We also extend the bar notation from strings to quantum states. So, for instance $\ket{\bar{0}} = \ket{ 00 ... 0}$. The multi-qubit \emph{cat state} can then be written as $\ket{\psi}=\frac{1}{\sqrt{2}}(\ket{\bar{0}}+\ket{\bar{1}})$.

We now recall some standard notions of classical and quantum computation. For more details, we refer the reader to~\cite{arora2009computational, nielsen2002quantum}. 
\begin{itemize}
\item The notion of computational efficiency will refer to algorithms or circuits that run in polynomial time.
\item We say that an algorithm (or Turing machine) is PPT if it uses randomness and runs in polynomial time. We say it is QPT if it is a quantum algorithm running in polynomial time.
\item All Boolean circuits we consider are comprised of AND, OR, XOR and NOT gates.
\item We say that a classical gate has \emph{bounded fan-out} if the number of output wires is constant (independent of the length of the input to the circuit). Otherwise, we say it has \emph{unbounded fan-out}.
\item For quantum computation we assume the standard circuit formalism with the gate set $\{ R_X, R_Y, R_Z, H, CZ, CNOT, CCNOT \}$ and computational basis measurements. Here, $R_X$, $R_Y$, $R_Z$ denote rotations along the $X$, $Y$ and $Z$ axes of the Bloch sphere. More precisely, $R_W(\theta) = exp(-i \theta W / 2)$, with $W \in \{X, Y, Z\}$, the set of Pauli matrices. The allowed rotation angles can be assumed to be multiples of $\pi/4$.
In addition, $H$ is the Hadamard operation, $CZ$ is a controlled application of a Pauli-$Z$ gate, $CNOT$ is a controlled application of a Pauli-$X$ gate and $CCNOT$ is a doubly-controlled Pauli-$X$ operation, also known as a Toffoli gate. It should be noted that, apart from $CCNOT$, a number of the existing quantum devices can indeed perform all of these gates natively \cite{arute2019quantum, Qiskit, wright2019benchmarking}.
\end{itemize}

We say that a computational problem is \emph{intractable} if there is no polynomial-time algorithm solving that problem. Throughout this paper we are only concerned with computational intractability for PPT algorithms. 
We give a simplified description of some candidate intractable problems of interest:
\begin{itemize}
\item \textbf{Factoring.} Given a composite integer $N$, find its prime-factor decomposition. For the specific case of semiprime $N = p \cdot q$, the task is to find primes $p$ and $q$.
\item \textbf{Discrete logarithm problem (DLP).} For some abelian group $\mathbb{G}$, given $g \in \mathbb{G}$ and $g^k$, with $k > 0$, find $k$.
\item \textbf{Learning with errors (LWE).} Letting $\mathbb{Z}_q$ be the ring of integers modulo $q \geq 2$, given the matrix $\mathbf{A} \in \mathbb{Z}^{m \times n}_q$ and the vector $y = \mathbf{A}s + e$, with $s \in \mathbb{Z}_q^n$ and $e$ sampled from a discrete Gaussian distribution over $\mathbb{Z}_q^m$, find $s$.
\item \textbf{Ring learning with errors (Ring-LWE).} Letting $R_q$ be a quotient ring $R_q = R/qR$, for some (cyclotomic) ring $R$ over the integers, given $m > 0$ pairs $(a_i, y_i)$ with $a_i \in R_q$ and $y_i = a_i \cdot s + e_i$, $i \leq m$, $s \in R_q$ and each $e_i$ sampled independently from a discrete Gaussian distribution over $R_q$, find $s$.
\end{itemize}
LWE and Ring-LWE are also conjectured to be QPT-intractable~\cite{regev2009lattices, lyubashevsky2010ideal}.

\subsection{Learning with rounding (LWR)}

As learning with rounding is the basis for our second proof of quantumness construction, in this subsection we define the problem and state some of its essential properties, taken from~\cite{lwr_revisited}.

\label{sec:def_lwr}
\begin{definition}[Rounding function]
For integers $q\geq p\geq 2$, the $p$-rounding function of an integer $\alpha$ satisfying $0\leq \alpha < q$ is defined as
\begin{equation}
    \lfloor\alpha\rfloor_p : \mathbb{Z}_q \rightarrow \mathbb{Z}_p = \biglfloor \frac{p}{q} \cdot \alpha \bigrfloor.
\end{equation}
\end{definition}
As mentioned in Subsection~\ref{subsect:introsecondpoq}, this rounding operation is equivalent to taking the most significant $\log_2 p$ bits of $\alpha$.

\begin{definition}[The learning with rounding (LWR) assumption \cite{lwr_revisited}]
Suppose $\mathbf{A}\in \mathbb{Z}_q^{m\times n}$, $x\leftarrow_r \mathbb{Z}_q^n$ and $u\leftarrow_r \mathbb{Z}_q^m$, then $(\mathbf{A},\lfloor\mathbf{A}x\rfloor_p)$ and $(\mathbf{A},\lfloor u\rfloor_p)$ are computationally indistinguishable.
\end{definition}
Note that this is the \emph{decision} version of LWR. There is also a \emph{search} version, in analogy to LWE. The search version is: given $(\mathbf{A}, \lfloor \mathbf{A}x \rfloor_p)$, as above, to find $x$. Whenever we refer to the ``learning with rounding problem'' we can use the decision version or the search version interchangeably, as they are equivalent for the parameter choices we use here.

\begin{definition}(Trapdoor one-way functions from LWR \cite{lwr_revisited})
\label{def:usefultools}
\begin{enumerate}
    \item $\Gen(n,m,q)$: an efficient algorithm that receives positive integers $n,m,q$ and samples a matrix $\mathbf{A}\in\mathbb{Z}^{m\times n}_q$ and \emph{trapdoor} $T$ with $\mathbf{A}$ being statistically close to uniform.
    \item $\Inv(T,\mathbf{A},c)$: an efficient algorithm that receives $T,\mathbf{A}$ in the support of $\Gen(n,m,q)$ and $c=\mathbf{A}x+e \in \mathbb{Z}^m_q$ for some $x\in\mathbb{Z}^n_q$ and some error $\| e \|_{\infty}\leq O\left(\frac{q}{\sqrt{n\log_2 q}}\right)$ and outputs $x$. 
    \item $\LWRInv(T,\mathbf{A},c)$: for $(\mathbf{A},T)$ in the support of $\Gen(n,m,q)$ and some $c\in\mathbb{Z}^m_p$ such that $c = \lfloor\mathbf{A}x\rfloor_p$, the function outputs $x$ efficiently.
\end{enumerate}
\end{definition}
\begin{lemma}[Trapdoors for LWR \cite{lwr_revisited}]
\label{lemma:trapdoorLWR}
There exist efficient $\Gen$ and $\LWRInv$ functions for any $n\geq 1$, $q\geq 2$, $m\geq O(n\log q)$ and $p\geq O(\sqrt{mn\log q})$. In particular, $\LWRInv$ is defined as
\begin{equation}
\LWRInv(T,\mathbf{A},c):=\Inv(T,\mathbf{A},\Transform_{q,p}(c))
\end{equation}
where
\begin{equation}
\Transform_{q,p}(c) := \left\lceil \frac{q}{p}\cdot c \right\rceil.
\end{equation}
\end{lemma}

We also note that for the parameter choices we consider throughout this paper, which are essentially the same as the ones in~\cite{brakerski2018cryptographic} (that is, $m$, $n$, $q$, $\|e\|_{\infty}$ as functions of the security parameter), LWE and LWR are computationally equivalent. In other words, there exists a polynomial-time reduction from LWE to LWR and vice-versa. We refer the reader to~\cite{lwroriginal, lwr_revisited} for the details.

\subsection{Proof of quantumness protocols}
\label{sec:poqps}
\subsubsection{Trapdoor claw-free functions}
Most proof of quantumness protocols are based on trapdoor claw-free (TCF) functions or noisy trapdoor claw-free functions (NTCF). We start with definition of a TCF, taken from~\cite{kahanamokumeyer2021classicallyverifiable}. 

\begin{definition}[TCF family \cite{kahanamokumeyer2021classicallyverifiable}]
	\label{def:tcf}

	Let $\lambda$ be a security parameter, $K$ a set of keys, and $X_k$ and $Y_k$ finite sets for each $k \in K$.
	A family of functions
	\[ \mathcal{F} = \{ f_k : X_k \to Y_k \}_{k \in K} \]
	is called a trapdoor claw free (TCF) family if the following conditions hold:

	\begin{enumerate}
		\item \textbf{Efficient Function Generation.} There exists a PPT algorithm $\Gen$ which generates a key $k \in K$ and the associated trapdoor data $t_k$:
		\[(k, t_k) \leftarrow \Gen(1^\lambda)\]
		\item \textbf{Trapdoor Injective Pair.} For all keys $k \in K$, the following conditions hold:
		\begin{enumerate}
			\item Injective pair: Consider the set $R_k$ of all tuples $(x_0, x_1)$ such that $f_k(x_0) = f_k(x_1)$.
			Let $X_k' \subseteq X_k$ be the set of values $x$ which appear in the elements of $R_k$.
			For all $x \in X_k'$, $x$ appears in exactly one element of $R_k$; furthermore, $\lim_{\lambda \to \infty} |X_k '|/|X_k| = 1$.
			\item Trapdoor: There exists a polynomial-time deterministic algorithm INV$_{\mathcal{F}}$ such that for all $y \in Y_k$ and $(x_0, x_1)$ such that $f_k(x_0) = f_k(x_1) = y$, INV$_{\mathcal{F}}(t_k, b, y) = x_b$, with $b \in \{0, 1\}$.
		\end{enumerate}
		\item \textbf{Claw-free.} For any non-uniform probabilistic polynomial time (nu-PPT) classical algorithm $\mathcal{A}$, there exists a negligible function $\mu(\cdot)$ such that
		\[ \Pr \left[ f_k(x_0) = f_k(x_1) \land x_0 \neq x_1 | (x_0, x_1) \leftarrow \mathcal{A}(k) \right] < \mu(\lambda) \]
		where the probability is over both the choice of $k$ and the random coins of $\mathcal{A}$.
		\item \textbf{Efficient Superposition.} There exists a polynomial-size quantum circuit that on input a key $k$ prepares the state
		\[ \frac{1}{\sqrt{|X_k|}} \sum_{x \in X_k} \ket{x} \ket{f_k(x)} \]
	\end{enumerate}
\end{definition}

Next, we define the notion of a noisy TCF, first introduced in~\cite{mahadev2018classical, brakerski2018cryptographic}. These are TCFs for which the efficient superposition is allowed to be approximate, rather than exact. The outputs of these functions are additionally assumed to be distributions over binary strings, rather than just binary strings. NTCFs, as defined in~\cite{brakerski2018cryptographic}, also satisfy a property known as the \emph{adaptive hardcore bit} which is independent of the ``noisy'' aspect of the TCF. As we want to distinguish between TCFs which satisfy this property and those that do not satisfy it, we shall refer to the former as \emph{strong} TCFs and the latter as ordinary TCFs, as per Definition~\ref{def:tcf}.
Thus, the NTCFs we consider will be referred to as strong NTCFs: 

\begin{definition}[Strong NTCF Family \cite{brakerski2018cryptographic}]\label{def:trapdoorclawfree}
Let $\lambda$ be a security parameter. Let $\mathcal{X}$ and $\mathcal{Y}$ be finite sets and $\mathcal{D}_{\mathcal{Y}}$ a collection of distributions over $\mathcal{Y}$.
 Let $\mathcal{K}_{\mathcal{F}}$ be a finite set of keys. A family of functions 
$$\mathcal{F} \,=\, \big\{f_{k,b} : \mathcal{X}\rightarrow \mathcal{D}_{\mathcal{Y}} \big\}_{k\in \mathcal{K}_{\mathcal{F}},b\in\{0,1\}}$$
is called a \textbf{strong noisy trapdoor claw-free (strong NTCF) family} if the following conditions hold:

\begin{enumerate}

\item{\textbf{Efficient Function Generation.}} Same as in Definition~\ref{def:tcf}.

\item{\textbf{Trapdoor Injective Pair.}} Same as in Definition~\ref{def:tcf}.

\item{\textbf{Efficient Range Superposition.}}
For all keys $k\in \mathcal{K}_{\mathcal{F}}$ and $b\in \{0,1\}$ there exists a function $f'_{k,b}:\sX\mapsto \mathcal{D}_{\sY}$ such that
\begin{enumerate} 
\item For all $(x_0,x_1)\in \mathcal{R}_k$ and $y\in \supp(f'_{k,b}(x_b))$, INV$_{\mathcal{F}}(t_k,b,y) = x_b$ and INV$_{\mathcal{F}}(t_k,b\oplus 1,y) = x_{b\oplus 1}$. 
\item There exists an efficient deterministic procedure CHK$_{\mathcal{F}}$ that, on input $k$, $b\in \{0,1\}$, $x\in \sX$ and $y\in \sY$, returns $1$ if  $y\in \supp(f'_{k,b}(x))$ and $0$ otherwise. Note that CHK$_{\mathcal{F}}$ is not provided the trapdoor $t_k$. 
\item For every $k$ and $b\in\{0,1\}$, 
$$ \Es{x\leftarrow_U \sX} \big[\,H^2(f_{k,b}(x),\,f'_{k,b}(x))\,\big] \,\leq\, \mu(\lambda)\;,$$
 for some negligible function $\mu(\cdot)$. Here $H^2$ is the Hellinger distance. Moreover, there exists an efficient procedure  SAMP$_{\mathcal{F}}$ that on input $k$ and $b\in\{0,1\}$ prepares the state
\begin{equation*}
    \frac{1}{\sqrt{|\sX|}}\sum_{x\in \sX,y\in \sY}\sqrt{(f'_{k,b}(x))(y)}\ket{x}\ket{y}\;.
\end{equation*}

\end{enumerate}

\item{\textbf{Adaptive Hardcore Bit.}}
For all keys $k\in \mathcal{K}_{\mathcal{F}}$ the following conditions hold, for some integer $w$ that is a polynomially bounded function of $\lambda$. 
\begin{enumerate}
\item For all $b\in \{0,1\}$ and $x\in \sX$, there exists a set $\dset_{k,b,x}\subseteq \{0,1\}^{w}$ such that $\Pr_{d\leftarrow_U \{0,1\}^w}[d\notin \dset_{k,b,x}]$ is negligible, and moreover there exists an efficient algorithm that checks for membership in $\dset_{k,b,x}$ given $k,b,x$ and the trapdoor $t_k$.
\item If
\begin{eqnarray}\label{eq:defsetsH}
H_k &=& \big\{(b,x_b,d,d\cdot(x_0\oplus x_1))\,|\; b\in \{0,1\},\; (x_0,x_1)\in \mathcal{R}_k,\; d\in \dset_{k,0,x_0}\cap \dset_{k,1,x_1}\big\}\;,\\
\overline{H}_k &=& \{(b,x_b,d,c)\,|\; (b,x,d,c\oplus 1) \in H_k\big\}\;,
\end{eqnarray}
then for any quantum polynomial-time procedure $\mathcal{A}$ there exists a negligible function $\mu(\cdot)$ such that 
\begin{equation}\label{eq:adaptive-hardcore}
\Big|\Pr_{(k,t_k)\leftarrow \textrm{GEN}_{\mathcal{F}}(1^{\lambda})}[\mathcal{A}(k) \in H_k] - \Pr_{(k,t_k)\leftarrow \textrm{GEN}_{\mathcal{F}}(1^{\lambda})}[\mathcal{A}(k) \in\overline{H}_k]\Big| \,\leq\, \mu(\lambda)\;.
\end{equation}
\end{enumerate}
\end{enumerate}
\end{definition}

As a point of clarification, note that a noisy TCF (NTCF) is a TCF with a modified efficient range superposition property. A strong NTCF is a NTCF with the adaptive hardcore bit property. As mentioned, in~\cite{mahadev2018classical, brakerski2018cryptographic} NTCFs are not distinguished from strong NTCFs. 
As an abuse of notation, we will use NTCF and strong NTCF interchangeably.

\subsubsection{The BCMVV protocol}
The first protocol we mention is the one from \cite{brakerski2018cryptographic}, which relies on the adaptive hardcore bit property and so the function family used is NTCF. We outlined the protocol in the introduction, while here we give a step-by-step description of its workings, in Figure \ref{fig:bcmvvprotocol}.

\protocol{BCMVV protocol}{The BCMVV proof of quantumness protocol based on NTCFs \cite{brakerski2018cryptographic}.}{fig:bcmvvprotocol}{
Let $\mathcal{F}$ be an NTCF family of functions. Let $\lambda$ be a security parameter and $N \geq 1$ a number of rounds. The parties taking part in the protocol are a PPT machine, known as the verifier and a QPT machine, known as the prover. They will repeat the following steps $N$ times:
\begin{enumerate}
\item The verifier generates $(k,t_k)\leftarrow \Gen(1^\lambda)$. It sends $k$ to the prover.
\item The prover uses $k$ to run SAMP$_{\mathcal{F}}$ and prepare the state:
\begin{equation*}
    \frac{1}{\sqrt{|\sX|}}\sum_{b \in \{0, 1\}, x\in \sX,y\in \sY}\sqrt{(f'_{k,b}(x))(y)}\ket{b}_{\mathsf{B}}\ket{x}_{\mathsf{X}}\ket{y}_{\mathsf{Y}}\;.
\end{equation*}
It then measures the $\mathsf{Y}$ register, resulting in the string $y\in \{0,1\}^{poly(\lambda)}$ which it sends to the verifier.
\item The verifier selects a uniformly random challenge $c\leftarrow_R \{0,1\}$ and sends $c$ to the prover.
\item \begin{enumerate}
    \item (Preimage test:) When $c=0$, the prover is expected to measure in the standard basis the $\mathsf{BX}$ registers of the state leftover in step 2. It obtains the outcomes $b \in \{0, 1\}$ and $x\in\{0,1\}^n$, with $n(\lambda) = poly(\lambda)$, which it sends to the verifier. If $\Chk_{\mathcal{F}}(k,b,x,y)=0$ the verifier aborts, otherwise it continues.
    \item (Equation test:) When $c=1$, the prover is expected to apply Hadamard gates to each qubit in the $\mathsf{BX}$ registers and measure them in the standard basis (equivalently, measure all qubits in the Hadamard basis). It obtains the outcomes $b' \in \{0, 1\}$ and $d\in\{0,1\}^n$ which it sends to the verifier. The verifier computes $(x_0,x_1) = \Inv_{\mathcal{F}}(t_k,y)$ and rejects if $d \cdot (x_0\oplus x_1) \neq b'$.
\end{enumerate}
\end{enumerate}
At the end of the $N$ rounds, if the verifier has not aborted it accepts.\\
}

\noindent The protocol is complete, in the following sense:
\begin{theorem}[\cite{brakerski2018cryptographic}]
A QPT prover, $\mathcal{P}$, following the honest strategy in the BCMVV protocol is accepted with probability $1 - negl(\lambda)$.
\end{theorem}
\noindent The soundness of the protocol against classical provers follows from the following theorem:
\begin{theorem}[\cite{brakerski2018cryptographic, zhu2021demonstration}]
For any PPT prover, $\mathcal{P}$, in the BCMVV protocol, it is the case that
\begin{equation}\label{ineq:winntcf}
  p_{\mathrm{pre}}+2p_{\mathrm{eq}}-2\leq negl(\lambda)  
\end{equation}
where $p_{\mathrm{pre}}$ is $\mathcal{P}$'s success probability in the preimage test and $p_{\mathrm{eq}}$ is $\mathcal{P}$'s success probability in the equation test.
\end{theorem}
Thus, in any run of the protocol, as long as Inequality~\ref{ineq:winntcf} is violated, we conclude that the prover is quantum.

One known instantiation of the BCMVV protocol, as is described by \cite{brakerski2018cryptographic}, is based on the LWE problem. The LWE-based construction is currently the only known instance of a strong NTCF family of functions.

\subsubsection{The KMCVY protocol}
The BCMVV protocol relies on the adaptive hardcore bit property of NTCFs in order to be sound.
However, this property is only known to be true for NTCFs based on LWE. The authors of~\cite{kahanamokumeyer2021classicallyverifiable} addressed this fact by introducing a proof of quantumness protocol that can use any TCF.
As mentioned in the introduction, their protocol is a sort of computational Bell test. We outline it in Figure \ref{fig:compBellprotocol}.

\protocol{KMCVY protocol}{The KMCVY proof of quantumness protocol based on TCFs \cite{kahanamokumeyer2021classicallyverifiable}.}{fig:compBellprotocol}{
Let $\mathcal{F}$ be a TCF family of functions. Let $\lambda$ be a security parameter, $N \geq 1$ a number of rounds and $T = 1/poly(\lambda)$ a threshold parameter. The parties taking part in the protocol are a PPT machine, known as the verifier and a QPT machine, known as the prover. 
Before interacting with the prover, the verifier initializes two counters $N_s = 0$, $N_t = 0$.
The two will then repeat the following steps $N$ times:
\begin{enumerate}
\item The verifier generates $(k,t_k)\leftarrow \Gen(1^\lambda)$. It sends $k$ to the prover.
\item The prover uses $k$ to prepare the state:
\begin{equation*}
\frac{1}{\sqrt{|X_k|}} \sum_{x \in X_k} \ket{x}_{\mathsf{X}} \ket{f_k(x)}_{\mathsf{Y}}
\end{equation*}
It then measures the $\mathsf{Y}$ register, resulting in the string $y\in \{0,1\}^{poly(\lambda)}$ which it sends to the verifier.
\item The verifier selects a uniformly random challenge $c\leftarrow_R \{0,1\}$ and sends $c$ to the prover.
\item \begin{enumerate}
    \item (Preimage test:) When $c=0$, the prover is expected to measure in the standard basis the $\mathsf{X}$ register of the state leftover in step 2. It obtains the outcome $x\in\{0,1\}^n$, with $n(\lambda) = poly(\lambda)$, which it sends to the verifier. If $f_k(x) \neq y$ the verifier aborts, otherwise it continues.
    \item (Computational Bell test:) When $c=1$,
    \begin{enumerate}
    \item The verifier sends a random bitstring $v \leftarrow_R \{0,1\}^n$ to the prover.
    \item The prover creates the state
\begin{equation*}
\frac{1}{\sqrt{2}}\left( \ket{v \cdot x_0}_{\mathsf{A}} \ket{x_0}_{\mathsf{X}} + \ket{v \cdot x_1}_{\mathsf{A}} \ket{x_1}_{\mathsf{X}} \right)
\end{equation*}
  with $f_k(x_0) = f_k(x_1) = y$.
    \item The prover applies Hadamard gates to all qubits in the $\mathsf{X}$ register and measures them in the standard basis. The measurement outcome is denoted $d\in\{0,1\}^{n}$ and is sent to the verifier. 
	\item The verifier computes $(x_0,x_1) = \Inv_{\mathcal{F}}(t_k,y)$. Together with $d$, the verifier can determine the current state $\ket{\gamma}_{\mathsf{A}} \in \{\ket{0},\ket{1},\ket{+},\ket{-}\}$ in the prover's $\mathsf{A}$ register. It then chooses a random $\phi \in \{\pi/4,-\pi/4\}$ and sends it to the prover.
    \item The prover is expected to measure the qubit in the $\mathsf{A}$ register in the basis:
    \begin{equation*}
\left\{ 
\begin{array}{ll}
\cos \left( \frac{\phi}{2} \right) \ket{0} + & \sin \left( \frac{\phi}{2} \right) \ket{1} \\
\cos \left( \frac{\phi}{2} \right) \ket{1} - & \sin \left( \frac{\phi}{2} \right) \ket{0} \\
\end{array}
\right\}\;.
\end{equation*}
    \item The verifier sets $N_s \leftarrow N_s +1$ if the measurement outcome was the likely one.
    \end{enumerate}
\end{enumerate}
\end{enumerate}
If the verifier has not aborted, it will accept if $\frac{N_s}{N_t} - 0.75 \geq T$.\\
}

\noindent The protocol is complete, in the following sense:
\begin{theorem}[\cite{kahanamokumeyer2021classicallyverifiable}]
A QPT prover, $\mathcal{P}$, following the honest strategy in the KMCVY protocol is accepted with probability $1 - negl(\lambda)$.
\end{theorem}
\noindent The soundness of the protocol against classical provers follows from the following theorem:
\begin{theorem}[\cite{kahanamokumeyer2021classicallyverifiable}]
For any PPT prover, $\mathcal{P}$, in the KMCVY protocol, it is the case that
\begin{equation}\label{ineq:wintcf}
  p_{\mathrm{pre}}+4p_{\mathrm{Bell}}-2\leq negl(\lambda)  
\end{equation}
where $p_{\mathrm{pre}}$ is $\mathcal{P}$'s success probability in the preimage test and $p_{\mathrm{Bell}}$ is $\mathcal{P}$'s success probability in the computational Bell test.
\end{theorem}
Thus, in any run of the protocol, as long as Inequality~\ref{ineq:wintcf} is violated, we conclude that the prover is quantum.

In~\cite{kahanamokumeyer2021classicallyverifiable}, the authors provide the following candidate TCFs: 
\begin{itemize}
\item Rabin's function, or $x^2 \; mod \; n$. The TCF properties are based on the computational intractability of factoring.
\item A Diffie-Hellman-based function. The TCF properties are based on the computational intractability of DLP.
\item A ring-LWE-based function. The TCF properties are based on the computational intractability of ring-LWE.
\end{itemize}
Of course, the NTCF family based on LWE can also be used.

\subsection{Randomized encodings} \label{subsect:re}
Randomized encodings (also known as \emph{garbled circuits} \cite{yao1986generate}) are probabilistic encodings of functions that are information-theoretically equivalent to the functions they encode.
The idea of constructing randomized encodings which can be evaluated in constant depth originated with~\cite{AIK06}. We restate here the essential definitions and results from that paper.

\begin{definition}[Randomized encoding \cite{AIK06}]
Let $f:\{0,1\}^n\rightarrow\{0,1\}^l$ be a function and $r\leftarrow_{R}\{0,1\}^m$ be $m$ bits sampled uniformly at random from $\{0,1\}^m$. We say that a function $\hat{f}:\{0,1\}^n\times\{0,1\}^m\rightarrow\{0,1\}^s$ is a $\delta$-correct, $\epsilon$-private randomized encoding of $f$ if it satisfies the following properties.
\begin{itemize}
\item Efficient generation. There exists a deterministic polynomial-time algorithm that, given a description of the circuit implementing $f$, outputs a description of a circuit for implementing $\hat{f}$.
  \item $\delta$-correctness. There exists a deterministic polynomial-time algorithm $\Dec$, called a decoder such that for every input $x\in\{0,1\}^n$, $\underset{{r\leftarrow_{R}}\{0,1\}^m}{\Pr}[\Dec(\hat{f}(x,r))\neq f(x)] \leq \delta$.
  \item $\epsilon$-privacy. There exists a PPT algorithm $S$, called a simulator, such that for every $x\in\{0,1\}^n$, $\TVD(S(f(x)), \hat{f}(x,r))\leq\epsilon$.
\end{itemize}
\end{definition}

A \emph{perfect randomized encoding} is one for which $\delta = 0$ (perfect correctness) and $\epsilon = 0$ (perfect privacy).
Note that for perfect encodings $f(x)$ can always be reconstructed from $\hat{f}(x, r)$. Additionally, perfect privacy means that $\hat{f}(x, r)$ encodes as much information about $x$ as $f(x)$.
An important property of perfect encodings that we will use is that of \emph{unique randomness}:

\begin{theorem}[Unique randomness \cite{AIK06}] \label{thm:uniquerand}
Suppose $\hat{f}$ is a perfect randomized encoding of $f$. Then for any input $x$, the function $\hat{f}(x,\cdot)$ is injective; namely, there are no distinct $r$,$r^\prime$ such that $\hat{f}(x,r)=\hat{f}(x,r^\prime)$. Moreover, if $f$ is a permutation, then so is $\hat{f}$.
\end{theorem}

The main result in~\cite{AIK06} is the following:
\begin{theorem}[\cite{AIK06}] \label{thm:constre}
Any Boolean function that can be computed by a log-depth circuit, admits a perfect randomized encoding that can be computed in constant depth.
\end{theorem}
In fact a more general result is shown in~\cite{AIK06}, however the result of the above theorem is sufficient for our purposes. We also require the following result:
\begin{lemma}[Randomness reconstruction] \label{lem:rrc}
Given $x$ and $\hat{f}(x,r)$, where $\hat{f}$ is a randomized encoding following the construction from~\cite{AIK06}, there is a deterministic polynomial-time algorithm, denoted $\Rrc$, for computing the randomness $r$.
\end{lemma}
Note that this property is not universal to randomized encodings, in that it cannot be derived from the definition of randomized encodings. However, the property is satisfied by the specific encodings defined in~\cite{AIK06}. This fact is mentioned in~\cite{AIK06}, however no formal proof is provided. We outline their construction in Appendix~\ref{app:bre} and prove the randomness reconstruction property in Appendix~\ref{app:rrc}.

Finally, we show the following fact concerning randomized encodings of functions that may have collisions:
\begin{lemma}[Collision preservation]
\label{col-pre}
For every $x_1, x_2$ with $x_1 \neq x_2$ for which $f(x_1) = f(x_2)$ there exist unique $r_1$ and $r_2$ such that $\hat{f}(x_1,r_1)=\hat{f}(x_2,r_2)$. In addition, for every $(x_1, r_1), (x_2, r_2)$, $x_1 \neq x_2$, such that $\hat{f}(x_1,r_1)=\hat{f}(x_2,r_2)$ it is the case that $f(x_1) = f(x_2)$.
\end{lemma}
\begin{proof}
Perfect privacy says that there exists a polynomial-time simulator $S$, such that for all $x$, it should be that $\TVD(S(f(x)), \hat{f}(x, r)) = 0$, where $\TVD$ is the total variation distance and $r$ is sampled uniformly at random. Essentially, $S$ should always be able to sample from the set of randomized encoding values that can be decoded to $f(x)$ (i.e. all $\hat{f}(x, r)$, for all $r$).

But now suppose we have $x_1$ and $x_2$ such that $f(x_1)$ = $f(x_2)$. By perfect privacy it must be that $\TVD(S(f(x_1)), \hat{f}(x_1, r_1)) = 0$ and $\TVD(S(f(x_2)), \hat{f}(x_2, r_2)) = 0$, for uniform $r_1$ and $r_2$. Since $f(x_1) = f(x_2)$, it must be that $\TVD(\hat{f}(x_1, r_1), \hat{f}(x_2, r_2)) = 0$. In other words, $\hat{f}(x_1, r_1)$ and $\hat{f}(x_2, r_2)$ are the same distribution (for random choices of $r_1$ and $r_2$) and so the randomized encodings that can be decoded to $f(x_1) = f(x_2)$ are the same for both $x_1$ and $x_2$.

Moreover, unique randomness (Theorem~\ref{thm:uniquerand}) ensures that there are no distinct $r_1$ and $r_1'$ such that $\hat{f}(x_1, r_1) = \hat{f}(x_1, r_1')$ (with the analogous statement holding for the $x_2$ case). Thus, for uniform $r_1$, $\hat{f}(x_1, r_1)$ is the uniform distribution over all randomized encodings which decode to $f(x_1) = f(x_2)$. As $\hat{f}(x_2, r_2)$ is the same distribution (for uniform $r_2$), it is the case that there are unique $r_1$ and $r_2$ such that $\hat{f}(x_1, r_1) = \hat{f}(x_2, r_2)$.
This shows the first part of the lemma, that for every $x_1, x_2$ with $x_1 \neq x_2$ for which $f(x_1) = f(x_2)$ there exist unique $r_1$ and $r_2$ such that $\hat{f}(x_1,r_1)=\hat{f}(x_2,r_2)$.

Next, consider $(x_1, r_1), (x_2, r_2)$, $x_1 \neq x_2$, such that $\hat{f}(x_1,r_1)=\hat{f}(x_2,r_2)$. Since $\Dec(\hat{f}(x_1,r_1)) = f(x_1)$ and $\Dec(\hat{f}(x_2,r_2)) = f(x_2)$, because $\hat{f}(x_1,r_1)=\hat{f}(x_2,r_2)$ it follows that $\Dec(\hat{f}(x_1,r_1)) = \Dec(\hat{f}(x_2,r_2))$ and so $f(x_1) = f(x_2)$.

Hence, the collisions of the original function are exactly preserved by the encoding.
\end{proof}

\section{Generic proofs of quantumness in constant quantum depth} \label{sect:poq}
We now have all the tools for presenting our generic compiler which can take the two proof of quantumness protocols from Subsection~\ref{sec:poqps} and map them to equivalent protocols in which the prover's operations require only constant quantum depth and logarithmic classical depth.
The idea is the following:
provided the (N)TCF of the original protocol can be evaluated in log depth, simply \emph{replace it with a constant-depth randomized encoding}, as follows from Theorem~\ref{thm:constre}. In other words, $y=f(x)$ should be replaced by $\hat{y}=\hat{f}(\hat{x})$ where $\hat{x}=(x,r)$ and $r$ denotes the randomness of the encoding. As mentioned, it was shown in~\cite{gheorghiu2020estimating, kahanamokumeyer2021classicallyverifiable} that the (N)TCFs of the two proofs of quantumness considered here, can indeed be performed in classical logarithmic depth.
Thus, to show that our construction works, we prove two things:
\begin{enumerate}
\item The prover can evaluate $\hat{f}$ coherently in constant quantum depth (as well as perform its remaining operations in constant depth). This is the completeness condition of the protocol shown in Subsection~\ref{subsect:completeness}.
\item A randomized encoding of a (N)TCF is itself a (N)TCF. This means that the modified protocol is sound against classical polynomial-time provers. We show this in Subsection~\ref{subsect:soundness}.
\end{enumerate}

\subsection{Completeness} \label{subsect:completeness}
To show completeness, we give a strategy for an honest prover, that interleaves constant-depth quantum circuits and log-depth classical circuits, to succeed in the proofs of quantumness described in Section~\ref{sect:poq}. We assume that the (N)TCFs used in those protocols can be evaluated in constant classical depth and denote the corresponding function as $\hat{f}_k$. These circuits are allowed to contain gates of unbounded fan-out. We can always map such a circuit to one that uses only gates of bounded fan-out, provided multiple copies of the input bits are provided. The intuition for this was mentioned in the Introduction and in Figure~\ref{fig:fanoutb}. We will assume each input bit of the initial circuit has been copied $k$ times.

The first step is preparing the state corresponding to a coherent evaluation of the (N)TCF over a uniform superposition of inputs:
\begin{equation} \label{eqn:targetstate}
    \ket{\psi}=\sum_{b\in\{0,1\}}\sum_{\hat{x}\in\{0,1\}^ {poly(\lambda)}}\ket{b}_{\mathsf{B}} \ket{x}_{\mathsf{X}}\ket{\hat{f}_k(b,x)}_{\mathsf{Y}},
\end{equation}
where the $\mathsf{B}$ and $\mathsf{X}$ registers store the inputs of $\hat{f}$ and the $\mathsf{Y}$ register will store the computed value of $\hat{f}$. As a slight abuse of notation, we omit the normalization term and assume the state is an equal superposition. 

Instead of preparing the state in Equation~\ref{eqn:targetstate}, we will prepare a state that is essentially equivalent to it, namely:
\begin{equation} \label{eqn:actualstate}
    \ket{\psi}=\sum_{b\in\{0,1\}}\sum_{\hat{x}\in\{0,1\}^ {poly(\lambda)}}\ket{\bar{b}}_{\mathsf{B}} \ket{\bar{x}}_{\mathsf{X}}\ket{\hat{f}_k(b,x)}_{\mathsf{Y}},
\end{equation}
where $\ket{\bar{b}}=\ket{b}^{\otimes k}$ and $\ket{\bar{x}}=\ket{x}^{\otimes k}$. We view the $\mathsf{X}$ register as consisting of multiple sub-registers, one for each bit in $x$. In other words\footnote{Note that this is the only place where a subscript on $x$ is used to denote a bit of $x$. Throughout the rest of the section, $x_b$ will denote a specific $x$ \emph{string}, and \emph{does not} refer to the $b$'th bit of the string $x$.}, if $x = x_1 x_2 ... x_n$, with $n(\lambda) = poly(\lambda)$, and $\bar{x} = \bar{x}_1 \bar{x}_2 ... \bar{x}_n$, we assume $\mathsf{X} = \mathsf{X}_1 \otimes \mathsf{X}_2 \otimes ... \otimes \mathsf{X}_n$. Here, $\mathsf{X}_i$ holds the state $\sum_{x_i \in \{0, 1\}} \ket{\bar{x}_i}$.

The prover starts by preparing:
\begin{equation} \label{eqn:startstate}
    \ket{\psi_0}=\sum_{b,\hat{x}}\ket{\bar{b}}_{\mathsf{B}} \ket{\bar{x}}_{{\mathsf{X}}}\ket{0}_{{\mathsf{Y}}}.
\end{equation}
Note that the $\mathsf{B}$ and $\mathsf{X}$ registers contain cat states. These can be prepared in constant quantum depth, together with logarithmic classical depth. As outlined in the introduction, the idea is to first prepare a poor man's cat state in constant depth, as described in~\cite{watts2019exponential}. The prover then uses the parity information from the prepared poor man's cat state to perform a correction operation consisting of Pauli-$X$ gates. Determining where to perform the $X$ gates from the parity information requires logarithmic classical depth. 
The $X$ corrections will map the poor man's cat states to cat states.

Next, the function $\hat{f}$ needs to be evaluated and the outcome will be stored in $\mathsf{Y}$ register. With multiple copies of the input, the circuit evaluating $\hat{f}$ consists only of gates with bounded fan-out. It can therefore be mapped to an equivalent constant depth quantum circuit (having twice the depth, so as to perform the operations reversibly) consisting of Toffoli, Pauli-$X$ and $CNOT$ gates. Evaluating this circuit on the state from~\ref{eqn:startstate} will result in the state from~\ref{eqn:actualstate}, as intended.

The prover is then required to measure the $\mathsf{Y}$ register and report the outcome to the verifier. This adds one more layer to the circuit. The measured state will collapse to
\begin{equation*}
    \ket{\psi_{y}}=\sum_{b\in\{0,1\}}\ket{\bar{b}}_{\mathsf{B}} \ket{\bar{x}_b}_{\mathsf{X}}\ket{y}_{\mathsf{Y}}.
\end{equation*}
In the preimage test, the prover will also measure this state in the computational basis and report the outcome to the verifier.

The next steps will differ for the two protocols.
\begin{enumerate}
    \item{\textbf{For the BCMVV protocol:}} In the equation test, the prover applies a layer of Hadamard gates on the qubits in $\mathsf{B}$ and $\mathsf{X}$. It then measures them in the computational basis, denoting the results as $b' \in \{0, 1\}^k$ and $d \in \{0, 1\}^{n \cdot k}$. In the original protocol, $b'$ was one bit and $d$ was $n$ bits and they satisfy the relation $d \cdot (x_0 \oplus x_1) = b'$. To arrive at that result, the prover will xor all the bits in $b'$ and all bits in each $k$-bit block of $d$ and report those results to the verifier. Note that the distributions of these xor-ed outcomes is the same as the distribution over the outcomes of a Hadamard-basis measurement of:
\begin{equation*}
    \sum_{b\in\{0,1\}} \ket{b}_{\mathsf{B}} \ket{x_b}_{\mathsf{X}}.
\end{equation*}
    \item{\textbf{For the KMCVY protocol:}} In the computational Bell test, the prover receives the string $v$ from the verifier. The original protocol has the prover use an ancilla qubit to store the bitwise inner product $v \cdot x_b$. However, such a multiplication requires \emph{serial} $CNOT$ gates which cannot be performed in constant depth. 
We therefore use a multi-qubit ancila register initalized as a cat state $\ket{a}_{\mathsf{A}}=\frac{\ket{0}^{\otimes n}+\ket{1}^{\otimes n}}{\sqrt{2}}$. For every bit $v_i$, in $v$, if $v_i=1$, the prover applies a controlled-Z ($CZ$) gate with control qubit any of the qubits in $\mathsf{X}_i$ and target qubit $\ket{a}_i$. The resulting state will be 
\begin{equation*}    
\sum_{b\in\{0,1\}} \frac{\ket{0}^{\otimes n}_{\mathsf{A}}+(-1)^{v\cdot x_b}\ket{1}^{\otimes n}_{\mathsf{A}}}{\sqrt{2}} \ket{\bar{x}_b}_{\mathsf{X}}
    =
    \sum_{b\in\{0,1\}}\ket{(-1)^{v\cdot x_b}}_{\mathsf{A}} \ket{\bar{x}_b}_{\mathsf{X}}
\end{equation*}	    

    where we denote $\ket{(-1)^{v\cdot x_b}}=\sum_{b\in\{0,1\}} \frac{\ket{0}^{\otimes n}_{\mathsf{A}}+(-1)^{v\cdot x_b}\ket{1}^{\otimes n}_{\mathsf{A}}}{\sqrt{2}}$.
    Next, the prover is required to measure $\mathsf{X}$ in the Hadamard basis yielding the result $d \in \{0, 1\}^{n \cdot k}$. Once again, in the original protocol $d$ is an $n$-bit string. As in the BCMVV protocol, this is ``fixed'' by having the prover xor each $k$-bit block of $d$ and report those outcomes to the verifier. The verifier can then use this result to determine the state in the ancilla register.
    
    After the measurement, the ancilla register will be in the state $\ket{\gamma}_{\mathsf{A}}\in\{\ket{\bar{0}},\ket{\bar{1}},\ket{\bar{+}},\ket{\bar{-}}\}$ where $\ket{\bar{\pm}}=\frac{\ket{\bar{0}}\pm\ket{\bar{1}}}{\sqrt{2}}$\footnote{Note that here the bar notation, $\ket{\bar{a}}$, refers to an $n$-fold repetition, rather than a $k$-fold one as in the previous case. That is, here $\ket{\bar{a}} = \ket{a}^{\otimes n}$.}. As the last step, the prover receives $\phi \in \{-\pi/4,\pi/4\}$. The original protocol requires him to measure the ancilla register in the rotated basis 
    $$\left\{           
    \begin{array}{lr}
    \cos(\phi/2)\ket{\bar{0}}+\sin(\phi/2)\ket{\bar{1}}\\
    \cos(\phi/2)\ket{\bar{1}}-\sin(\phi/2)\ket{\bar{0}}
    \end{array}
    \right.$$
    and report the result, $b'$. But how does the prover perform this measurement in constant depth? We give an approach that requires one more round of interleaving constant-depth quantum circuits and a log-depth classical computation. The basic idea is to reduce the multi-qubit state in the ancilla to a single-qubit state, i.e. $\{\ket{\bar{0}},\ket{\bar{1}},\ket{\bar{+}},\ket{\bar{-}}\} \rightarrow \{\ket{0},\ket{1},\ket{+},\ket{-}\}$. This reduction needs to be done in such a way that $\{\ket{\bar{0}},\ket{\bar{1}}\} \to \{ \ket{0}, \ket{1} \}$ and $\{\ket{\bar{+}},\ket{\bar{-}}\} \to \{ \ket{+}, \ket{-} \}$. Once this is done, the resulting qubit can be measured in the rotated basis.

    To perform the reduction, the prover first measures all but one qubit of $\ket{\gamma}_{\mathsf{A}}$ in the Hadamard basis. Denote this $(n-1)$-bit outcome as $w$. If the initial state was $\ket{\bar{0}}$ or $\ket{\bar{1}}$, the unmeasured qubit will be $\ket{0}$ or $\ket{1}$ respectively. If the initial state was $\ket{\bar{\pm}}$, it can be re-expressed as
    \begin{equation*}
    \begin{split}
    \ket{\bar{\pm}}&\propto \ket{0}\ket{00...0}\pm\ket{1}\ket{11...1}\\
    &\propto\ket{0}(\ket{+}+\ket{-})^{\otimes n-1} \pm \ket{1}(\ket{+}-\ket{-})^{\otimes n-1}\\
    &\propto \sum_{w} \left( \ket{0}\pm(-1)^{|w|}\ket{1} \right) \ket{w}\\
    &\propto \sum_{w} Z^{|w| \; mod \; 2}\ket{\pm}\ket{w}
    \end{split}
    \end{equation*}
    
    Thus, the qubit after the measurement will be $Z^{|w| \; mod \; 2}\ket{\pm}$. The prover will apply the  $Z^{|w| \; mod \; 2}$ operation to this qubit. In this way, the state $\ket{\bar{\pm}}$ is reduced to $\ket{\pm}$.
    
    Finally, the prover has to measure the qubit in the rotated basis and report the outcome. This can be done in constant depth by rotating the qubit appropriately and measuring in the standard basis. As in the original protocol, this prover will pass the verifier's checks with probability $\cos(\pi/8)^2\approx 85\%$.    
\end{enumerate}

\subsection{Soundness} \label{subsect:soundness}
We do not need to prove soundness from scratch for our modified protocols. Instead, since our only change was to replace the (N)TCFs used in the protocols with randomized encodings, we will have the same soundness as the original constructions provided randomized encodings of (N)TCFs are still (N)TCFs. That is what we show here.

\begin{theorem}
A perfect randomized encoding of a (N)TCF, satisfying the randomness reconstruction property, is still a (N)TCF.
\end{theorem}
\begin{proof}
We show this result for NTCFs specifically, since the TCF case is subsumed. The idea of the proof is to show that every property of a NTCF is also satisfied by its randomized encoding.
\begin{enumerate}
    \item{\textbf{Efficient Function Generation.}} By definition, randomized encodings can be efficiently generated given a description of the function to be encoded. In this case, the description is given by the public key produced by the PPT algorithm $\textrm{GEN}_{\mathcal{F}}$. More precisely, $\textrm{GEN}_{\mathcal{F}}$ generates the key $k\in \mathcal{K}_{\mathcal{F}}$ together with a trapdoor $t_k$. The generating procedure for the encoding will run $\textrm{GEN}_{\mathcal{F}}$ and output $k$, the efficient circuit for generating a randomized encoding and the trapdoor $t_k$. Schematically,
    
    $$(\hat{f}_{k,b}, t_k) \xleftarrow{\textrm{randomized encoding}}(f_{k,b}, t_k) \equiv (k,t_k) \leftarrow \textrm{GEN}_{\mathcal{F}}(1^\lambda)\;.$$

    \item{\textbf{Trapdoor Injective Pair.}}
    \begin{enumerate}
        \item \textit{Trapdoor}: Due to perfect correctness, $\supp(\hat{f}_{k,b}(x_0,r_0)) \cap \supp(\hat{f}_{k,b}(x_1,r_1))=\emptyset$ is satisfied since if $\supp(\hat{f}_{k,b}(x_0,r_0)) \cap \supp(\hat{f}_{k,b}(x_1,r_1))\neq\emptyset$, then perfect correctness leads to $\supp(f_{k,b}(x_0)) \cap \supp(f_{k,b}(x_1))\neq\emptyset$ which violates the trapdoor injective pair property of the original function $f$. The efficient deterministic algorithm for inverting the randomized encoding also exists and is defined as $\Inv_{\hat{\mathcal{F}}}(t_k,b,\hat{y})=\Rrc\circ \Inv_{\mathcal{F}}\circ \Dec (t_k,b,\hat{y})$, i.e. the composition of the decoding operation for the encoding, the original $\Inv_{\mathcal{F}}$ procedure of the NTCF and the randomness reconstruction procedure (see Lemma~\ref{lem:rrc}).
        \item \textit{Injective pair}: 
        Let $\hat{R}_k$ be the set of all tuples of the form $((x_0, r_0), (x_1, r_1))$ such that 
        $\hat{f}_{k,0}(x_0, r_0) = \hat{f}_{k, 1}(x_1, r_1)$. Additionally, let $\hat{X}'_k \subseteq \hat{X}_k$ be the set of values $(x, r)$ which appear in the elements of $\hat{R}_k$.
        It is the case that every $(x, r) \in \hat{X}'_k$ appears in exactly one element of $\hat{R}_k$. This is because, using the collision-preservation property (Lemma~\ref{col-pre}), it must be that 
        $\hat{f}_{k,0}(x_0, r_0) = \hat{f}_{k, 1}(x_1, r_1)$ only if $f_{k,0}(x_0) = f_{k, 1}(x_1)$ and only for unique $r_1$ and $r_2$. We also know from the injective pair property of $f_{k, b}$, that every $x$ appears in exactly one tuple defining a collision for $f_{k, b}$.
        
        Also note that $|\hat{X}_k| = 2^m |X_k|$, where $|r| = m$. In other words, the set of possible inputs for $\hat{f}_{k, b}$ is $2^m$ times larger than that of $f_{k, b}$, as for every input, $x$, we also have the $m$-bit string $r$. The collision preservation property (Lemma~\ref{col-pre}) also ensures that $|\hat{X}'_k| = 2^m |X'_k|$. Since we know that $\lim_{\lambda \to \infty} |X_k '|/|X_k| = 1$ it also follows that $\lim_{\lambda \to \infty} |\hat{X}_k '|/|\hat{X}_k| = 1$.
    \end{enumerate}
    \item{\textbf{Efficient Range Superposition.}}
    The efficient range superposition property of the original function $f$ means there's an efficient quantum procedure to create a state approximating a superposition over the range of $f$. Assume we add an additional register, $\mathcal{R}$, to represent the randomness of the encoding, $\hat{f}$, and initialize it as a uniform superposition over computational basis states. We can now combine the efficient procedure for generating $\hat{f}$ with the procedure for generating the range superposition of $f$ and apply them coherently on $\mathcal{R}$. This will then yield the desired state
    
    \begin{equation*}
    \sum_{x, r, y}\sqrt{(\hat{f}'_{k,b}(x, r))(y)}\ket{x}\ket{r}\ket{y}\;,
\end{equation*}
suitably normalized.    
    
    \item{\textbf{Adaptive Hardcore Bit.}} We prove this property by contradiction. Assume there exists a QPT adversary $\hat{\mathcal{A}}$ that breaks the adaptive hardcore bit property for the randomized encoding. This means that there exists a non-negligible function $p(\lambda)$ that satisfies
    $$\Big|\Pr_{(k,t_k)\leftarrow \textrm{GEN}_{\mathcal{F}}(1^{\lambda})}[\hat{\mathcal{A}}(k) \in \hat{H}_k] - \Pr_{(k,t_k)\leftarrow \textrm{GEN}_{\mathcal{F}}(1^{\lambda})}[\hat{\mathcal{A}}(k) \in\overline{\hat{H}}_k]\Big| \,\geq \, p(\lambda)\;$$
    where
    $$\hat{H}_k = \big\{(b,\hat{x}_b,\hat{d},\hat{d}\cdot(\hat{x}_0\oplus \hat{x}_1))\,|\; b\in \{0,1\},\; (\hat{x}_0,\hat{x}_1)\in \hat{\mathcal{R}}_k,\; \hat{d}\in \hat{\dset}_{k,0,x_0}\cap \hat{\dset}_{k,1,x_1}\big\}\;$$ and
    $$\overline{\hat{H}}_k = \{(b,\hat{x}_b,\hat{d},c)\,|\; (b,\hat{x},\hat{d},c\oplus 1) \in \hat{H}_k\big\}\;.$$
    
    By definition $\hat{x}_b=(x_b,r_b)$, therefore $\hat{d}$ can be split into $(d_x,d_r)$ such that
    \begin{equation*}
        \hat{x}_b\cdot\hat{d}=(x_b\cdot d_x) \oplus (r_b\cdot d_r)
    \end{equation*}
    which implies that
    \begin{equation*}
        \hat{d}\cdot(\hat{x}_0\oplus \hat{x}_1) = (d_x\cdot(x_0\oplus x_1)) \oplus (d_r\cdot(r_0\oplus r_1)).
    \end{equation*}
    Note that the output of $\hat{\mathcal{A}}$ is a tuple $(b,\hat{x}_b,\hat{d},\hat{d}\cdot(\hat{x}_0\oplus \hat{x}_1))$.
    One can now define a new QPT adversary $\mathcal{A}$ which runs $\mathcal{\hat{A}}$ and then outputs $(b,x_b,d_x,\hat{d}\cdot(\hat{x}_0\oplus \hat{x}_1) \oplus (d_r\cdot(r_0\oplus r_1)) )$. This then implies that
    $$\Big|\Pr_{(k,t_k)\leftarrow \textrm{GEN}_{\mathcal{F}}(1^{\lambda})}[\mathcal{A}(k) \in H_k] - \Pr_{(k,t_k)\leftarrow \textrm{GEN}_{\mathcal{F}}(1^{\lambda})}[\mathcal{A}(k) \in\overline{H}_k]\Big| \,\geq\, p(\lambda)\;.$$
    Hence, the adaptive hardcore bit of the original NTCF family is violated. We conclude that the randomized encoding must also satisfy the adaptive hardcore bit property.
\end{enumerate}

\end{proof}

\subsection{Resource estimation} \label{subsect:resources}
In this section, we give some estimates of the resources required to run our modified protocols. We summarize this information in Table~\ref{tab:my_label} and proceed to explain the results. The functions listed in the table are the same as the ones from~\cite{kahanamokumeyer2021classicallyverifiable}, as these are the existing candidate TCFs used in proof of quantumness protocols.

\begin{table}[h!]
    \centering
    \begin{tabular}{ccccc}
    \hline 
    Function&Adaptive H.C.&\# of quantum-classical interleavings&Depth&Width\\
    \hline 
    LWE&\ding{51}&3&14&$O(\lambda l^{4})$\\
    Ring-LWE&\ding{55}&4&18&$O(\lambda l^4)$\\
    $x^2 \mod n$&\ding{55} &4&18&$O(\lambda l^4)$\\
    Diffie-Hellman&\ding{55}&4&18&$O(\lambda l^4)$\\
    \hline
    \end{tabular}
    \caption{The table of resource estimations for each type of (N)TCF function that may be used. Here H.C. means hardcore bit. The number of quantum-classical interleavings refers to the instances where the prover performs a constant-depth quantum circuit followed by a classical computation. This is done, for instance, in the preparation of cat states as well as when it responds to one of the verifier's challenges. Depth refers to the total number of layers of quantum gates that the prover has to perform. Width refers to the width of the quantum circuits the prover has to implement. Here, $\lambda$ denotes the security parameter and $l$ is the size of the branching program implementing the randomized encoding, as described in Appendix~\ref{app:bre}.}
    \label{tab:my_label}
\end{table}

\subsubsection{Quantum depth and quantum-classical interleavings}
In this subsection we explain the overall quantum depth that the prover has to perform in our modified proofs of quantumness. Depth here represents the number of layers of quantum gates or measurements (as described in Section~\ref{sect:prelim}) that the prover will perform throughout the protocol, in the worst case. As mentioned, the prover's operations consist of alternating between constant-depth quantum circuits and log-depth classical computation. This latter step we referred to as a quantum-classical interleaving.

For the NTCF-based protocol which uses LWE, the total quantum depth is 14 and 3 quantum-classical interleavings are performed, whereas for the TCF-based approaches the depth is 17 and the number of interleavings is 4. Let us explain where these numbers come from:
\begin{enumerate}
    \item{\textbf{Preparation of cat states.}} As mentioned, we prepare cat states by interleaving a constant depth quantum circuit with a log-depth classical computation, followed by another quantum circuit. The exact steps are outlined in~\cite{watts2019exponential}, while here we just summarize the gates performed in each step. The procedure starts with a layer of Hadamard gates followed by two layers of $CNOT$ gates. Some of the qubits are then measured in the computational basis. The remaining qubits will collapse to a poor man's cat state, while the measured qubits contain the parity information for that state. To ``correct'' the state to a cat state, the parity information is used to compute a Pauli-$X$ correction. This is one quantum-classical interleaving. The final quantum layer consists of Pauli-$X$ gates. Thus, the total depth will be 5 and we have 1 quantum-classical interleaving. This applies to all cat states, as they can be prepared in parallel.
    
    \item{\textbf{Evaluation of the randomized encoded function.}} As illustrated can see in Figure \ref{fig:classical-circuit}, the classical circuit for a randomized encoding has depth 3. In the quantum case, the AND gates are implemented by Toffoli gates and the XOR gate is a $CNOT$. As the quantum gates are reversible, one needs to \emph{uncompute} any auxiliary results and so the quantum depth will be double that of the classical circuit. Hence, for this step the quantum depth is 6 and there are no quantum-classical interleavings.
    
    \item{\textbf{Measurement of the $\mathsf{Y}$ register.}} Measuring the image register requires a layer of computational basis measurements and so the depth is 1. The results are read out and sent to the verifier, which we count as 1 quantum-classical interleaving.
    
    \item{\textbf{Preimage test or equation/Bell test.}} If a preimage test is performed, the prover only needs to measure the $\mathsf{X}$ register in the computational basis and report the result. This counts as depth 1 and 1 interleaving. In the NTCF protocol, if an equation test is performed, then the prover is expected to apply a layer of Hadamard gates to the $\mathsf{X}$ register and measure them. This counts as depth 2 and 1 interleaving. In the TCF protocol, when the computational Bell test is performed, the prover's operations (as outlined in Subsection~\ref{subsect:completeness}) will consist of a layer of $CZ$ gates, a layer of Hadamard gates together with a computational basis measurement, a classical computation and reporting the results to the verifier, a Pauli-$Z$ operation, a rotation gate and finally another measurement and reporting the results to the verifier. This counts as depth 6 and 2 interleavings.
\end{enumerate}

We can see that for the NTCF-based protocol the worst-case depth is $5+6+1+2=14$ and the number of interleavings is $1+0+1+1=3$. For the TCF-based one, the depth is $5+6+1+6=18$ and the number of interleavings is $1+0+1+2=4$.

\subsubsection{Circuit width}
The constant-depth versions of the proof of quantumness protocols require larger numbers of qubits than the original version. As explained, most of this is due to the use of cat states, which effectively copy the input and allow us to apply a constant depth circuit with bounded fan-out gates.
That circuit is a randomized encoding of the original TCF. Following the construction of randomized encodings from~\cite{AIK06} and described in Appendix~\ref{app:bre}, the width of the constant-depth circuit will depend on the size of the branching program used to evaluate the original function.
In Appendix~\ref{app:bre} we explain how, as a result of \emph{Barrington's theorem}, the size of this branching program is exponential in the depth of the original TCF. As all TCFs considered here can be evaluated in logarithmic depth, the resulting branching programs will have sizes polynomial in the security parameter $\lambda$.
Giving a precise account of the size of the branching program, as a function of $\lambda$, for each TCF, is beyond the scope of this paper. Instead, we find in Appendix~\ref{app:bre} that the overall circuit width for the prover's quantum circuit is $O(\lambda l^4)$, where $l$ is the size of the branching program used to evaluate the TCF. The $\lambda$ factor comes from having to repeat the branching program construction in parallel $O(\lambda)$ times. This is because one branching program computes a single output bit of the TCF and so one has to consider a different branching program (of the same size) for each output bit.

As a rough estimate, we can relate the width to the security parameter for the LWE-based NTCF of~\cite{mahadev2018classical, brakerski2018cryptographic}.
There we know from~\cite{gheorghiu2020estimating} that the functions can be evaluated in depth $\propto4\log{\lambda}$. From Barrington's theorem, the size $l$ of the corresponding branching program is on the order of $\lambda^8$. As the width is $O(\lambda l^4)$, we find that the prover requires $O(\lambda^{33})$ qubits. 
This is a discouraging result for the purposes of implementing these protocols on near-term devices. However, it should be noted that this was merely a rough calculation based on existing asymptotic estimates. We conjecture that these estimates are not optimal and can be improved with a tighter analysis, better circuit implementations and more compact branching programs.
Additionally, for a fixed-size implementation (say $\lambda=50$), it is likely that additional optimizations are possible that could further reduce the number of required qubits.

\section{Proofs of quantumness via phase encoding}
\label{sect:phase_enc}
The first construction based on randomized encoding is a generic method that works for all types of (N)TCFs. However, as mentioned, its naive implementation based on Barrington's theorem leads to circuits which are too wide to be implemented on near-term devices.

In this section, we propose another approach that can be implemented on much narrower circuits, thus bringing it closer to implementation on near-term devices. This construction relies on \emph{phase encodings} to evaluate a specific NTCF, based on the LWR problem that is defined in Subsection~\ref{sec:def_lwr}. As we will see, the resulting circuits also involve only constant quantum depth and logarithmic classical depth.

Before presenting the protocol, we first define the LWR-based NTCF, denoted as $f$, and introduce its phase encoded implementation.

\subsection{LWR-based NTCF}
\label{sec:lwr-ntcf}
The LWR-based NTCF was suggested in \cite{brakerski2018cryptographic} but not used. It is however used in~\cite{zhu2021demonstration}, but without the phase encoding. The specific NTCF we consider is the following:

\begin{definition}[LWR-based NTCF] \label{def:lwrntcf}
Let $\lambda > 0$ be a security parameter. We take $n(\lambda), m(\lambda), q(\lambda), p(\lambda)$ as functions of $\lambda$ subject to the following constraints: $n = O(\lambda)$, $q = 2^{O(n)}$ is prime, $m = \Omega(n \log q)$, and $p=O(\sqrt{mn\log q})$ is a power of 2. Additionally $\chi$ will denote a discrete Gaussian distribution over $\mathbb{Z}_q$ having width $O(q/p^5)$. Taking $\mathbf{A} \leftarrow_r \mathbb{Z}_q^{m \times n}$, $s \leftarrow_r \{0, 1\}^n$, $e \leftarrow_{\chi^m} \mathbb{Z}_q^{m}$ (so that $\| e \|_{\infty} = O(q/p^5)$), we define the function
$$f(b,x):\{0,1\}\times\mathbb{Z}_q^n \rightarrow \mathbb{Z}_p^m=\lfloor g(b,x)\rfloor_p$$ 
where 
$$g(b,x):\{0,1\}\times\mathbb{Z}_q^n \rightarrow \mathbb{Z}_q^m =\mathbf{A}x+b \cdot (\mathbf{A}s+e).$$
\end{definition}
\noindent For the specific constants in the parameters defined above, we use the same values as in~\cite{brakerski2018cryptographic}. It should be noted that the width of the error distribution is taken to be polynomially smaller than in~\cite{brakerski2018cryptographic} ($O(q/p^5)$ versus $O(q/p)$). But since the width is still superpolynomial (in $n$) we are still in the ``hardness regime'' where both LWE and LWR are intractable. For more details, we refer the reader to the Preliminaries of~\cite{brakerski2018cryptographic}. The reason for this choice will become apparent in Subsection~\ref{subsect:fidelity}.

Although we are referring to $f$ as an NTCF, it is not clear if this is indeed the case. Following the definition from Subsection \ref{sec:poqps}, we next show that all the properties are satisfied. As $f(b,x)$ uses the same LWE instance as the LWE-based NTCF of~\cite{brakerski2018cryptographic}, we will have the same $\Gen$, which immediately proves the efficient function generation property. Additionally, Lemma~\ref{lemma:trapdoorLWR} confirms that the $(k,t_k)$ pair sampled by $\Gen$ is also the key and trapdoor pair for the LWR-based function (for this reason we will sometimes write the function as $f_k$). We can also see that if $(0, x)$ is the preimage of $y=f(0,x)$, the other preimage is $(1,x-s)$. The trapdoor injective pair property then follows. The efficient evaluation property comes from the fact that mod-$q$ matrix multiplication and additions can be efficiently performed by polynomial-depth quantum circuits. In fact, the rest of this section is devoted to showing an efficient evaluation in constant quantum-depth using the phase encoding construction.

We are left with showing the adaptive hardcore bit property. As a first step, we show the following:
\begin{lemma}\label{lem:LWRclaws}
$x_0$ and $x_1$ form a claw of the LWR-based NTCF if and only if they are also a claw of the corresponding LWE-based NTCF (from~\cite{brakerski2018cryptographic}), with high probability.
\end{lemma}
\begin{proof}
Consider
\begin{align*}
    f(b,x) &= \biglfloor \mathbf{A}x+b\cdot (\mathbf{A}s +e) \bigrfloor_p \\
    h(b,x) &=  \mathbf{A}x+b\cdot (\mathbf{A}s +e) + e'
\end{align*}
where $h$ is the LWE-based NTCF using in~\cite{brakerski2018cryptographic} and both functions are based on \emph{the same LWE sample} $\mathbf{A}s+e$. The statement we would like to show is then re-expressed as
\begin{equation*}
    f(0,x_0)=f(1,x_1) \Leftrightarrow h(0,x_0)=h(1,x_1)
\end{equation*}
with high probability over the choices of $\mathbf{A}, s,$ and $e$. We can prove it by showing both implications.
\begin{itemize}
    \item {($\rightarrow$)} Consider its contrapositive: if $h(0,x_0)\neq h(1,x_1)$, then $f(0,x_0)\neq f(1,x_1)$, with high probability. In \cite{brakerski2018cryptographic}, it was shown that $h(0,x_0)\neq h(1,x_1)$ if and only if $x_1 \neq x_0-s$, with high probability. Now take $x_1 = x_0 - s + w$ for some non-zero $w \in \mathbb{Z}_q^n$. We know that $\mathbf{A} w$ is a uniformly random vector (over the random choice of $\mathbf{A}$) and therefore every bit of $f(1,x_1)$ has a probability of $\frac{1}{2}$ to be flipped with respect to $f(0,x_0)$. Thus, the probability of $f(0,x_0)=f(1,x_1)$ can be bounded by the additive Chernoff inequality 
    \begin{equation*}
        \Pr(d_H(f(0,x_0),f(1,x_1))=0)\leq \exp \left(-\frac{m \log_2 p}{4} \right)
    \end{equation*}
    which is negligible.
    \item {($\leftarrow$)} Suppose $h(0,x_0)=h(1,x_1)$, which immediately leads to $x_1 = x_0-s$, with high probability. We then have $f(0,x_0) = \lfloor \mathbf{A}x_0 \rfloor_p$ and $f(1,x_1) = \lfloor \mathbf{A}x_0 + e \rfloor$. As we have $\| e \|_{\infty} = O(q/p^5)$, the probability of $f(0,x_0)=f(1,x_1)$ is $1-negl(n)$ as shown in \cite{lwr_revisited}.
\end{itemize}

\end{proof}

\noindent Now we have all the ingredients for the proof of the adaptive hardcore bit property.
\label{sec:adaptive-hc}
\begin{theorem}
The LWR-based NTCFs ($f_k(b,x)$) have the adaptive hardcore bit property.
\end{theorem}
\begin{proof}
We present a proof by contradiction. Suppose $f_k(b,x)=\lfloor\mathbf{A}x+b(\mathbf{A}s+e)\rfloor_p$ is an LWR-based NTCF where $k$ is the key and $t_k$ is the trapdoor, both generated by $\Gen$. Assume there exists a QPT adversary $\hat{\mathcal{A}}$ that breaks the adaptive hardcore bit property of $f$. This means that there exists a non-negligible function $\kappa(m)$ that satisfies
\begin{equation*}
\Big|\Pr_{(k,t_k)\leftarrow \Gen_{\mathcal{F}}(1^{\lambda})}[\hat{\mathcal{A}}(k) \in \hat{H}_k] - \Pr_{(k,t_k)\leftarrow \Gen_{\mathcal{F}}(1^{\lambda})}[\hat{\mathcal{A}}(k) \in\overline{\hat{H}}_k]\Big| \,\geq\, \kappa(m)\;
\end{equation*}
where
\begin{eqnarray*}
\hat{H}_k &=& \big\{(b,x_b,d,d\cdot(x_0\oplus x_1))\,|\; b\in \{0,1\},\; (x_0,x_1)\in \hat{\mathcal{R}}_k \big\}\;,\\
\overline{\hat{H}}_k &=& \{(b,x_b,d,c)\,|\; (b,x,d,c\oplus 1) \in \hat{H}_k\big\}\;,
\end{eqnarray*}
and $\hat{\mathcal{R}}_k$ is the set of all tuples $x_0$, $x_1$ such that $f_k(0,x_0)=f_k(1,x_1)$. We can then consider the LWE-based NTCF $h_k(b,x):=\mathbf{A}x+b \cdot (\mathbf{A}s+e)+e'$, whose corresponding sets are denoted by $H_k$, $\overline{H}_k$, and $\mathcal{R}_k$. As is shown in Lemma~\ref{lem:LWRclaws}, we have $\mathcal{R}_k = \hat{\mathcal{R}}_k$, with overwhelming probability, hence $H_k = \hat{H}_k$ and $\overline{H}_k = \overline{\hat{H}}_k$ . Therefore, we can define the QPT adversary, $\mathcal{A} :=\hat{\mathcal{A}}$. It satisfies
\begin{equation*}
\Big|\Pr_{(k,t_k)\leftarrow \textrm{GEN}_{\mathcal{F}}(1^{\lambda})}[\mathcal{A}(k) \in H_k] - \Pr_{(k,t_k)\leftarrow \textrm{GEN}_{\mathcal{F}}(1^{\lambda})}[\mathcal{A}(k) \in\overline{H}_k]\Big| \,\geq\, \kappa(m)\;
\end{equation*}
which breaks the adaptive hardcore bit property of LWE-based NTCFs.
\end{proof}
\noindent This implies that $f(b,x)$ satisfies all requirements of an NTCF.

\subsubsection{Prime $q$}
\label{sec:q_is_prime}

As mentioned in Definition~\ref{def:lwrntcf}, we require $q$ to be a prime. This is, in fact, also a requirement in~\cite{brakerski2018cryptographic}. The reason for this is that some of the properties of these NTCF-based constructions hold only when $\mathbb{Z}_q$ is a finite field, rather than a finite ring. Normally, this would just be a minor technical point. However, in our case since we would like to perform the prover's operations in constant depth, we would need to provide a procedure that allows the prover to prepare equal superpositions over the field elements. In other words, the prover needs to create an equal superposition of a prime number of elements. While this can be done in constant quantum depth, using cat states and ideas from~\cite{quantumfanout}, we will find that this is not necessary, provided $q$ is sufficiently large and \emph{sufficiently close to a power of 2}. In this section, we show that these conditions can indeed be satisfied and it is possible to efficiently choose a prime $q$ that is close to a power of 2.

We start with a result from~\cite{dusart}:
\begin{lemma}[\cite{dusart}]
For $q'>3275$, there exists a prime $q$ in the interval
\begin{equation*}
    q'<q<\left(1+\frac{1}{2\ln^2 q'}\right)q'.
\end{equation*}
\end{lemma}
\noindent This implies that the ratio of $q$ and $q'=2^n$ is bounded by
\begin{equation*}
    1<\frac{q}{q'}<1+\frac{1}{2(\ln2)^2 n^2}=1+O(n^{-2}).
\end{equation*}

Moreover, a specific prime in between $q'=2^n$ and $\left(1+\frac{1}{2\ln^2 q'}\right)q'$ can be efficiently found. It suffices to sample random integers in the range and check if they are prime. The checking can be done by (for instance) the Miller-Rabin algorithm \cite{rabintest}, in polynomial time. We can show that the number of samples to check is $O(n)$ using the \emph{Prime number theorem}, which states that, if $\pi(N)$ is the prime counting function, for integers in the range $(0,N)$, then it is the case that
\begin{equation*}
    \pi(N)\sim \frac{N}{\log N}.
\end{equation*}

\noindent Thus, the number of primes in the desired range can be estimated by
\begin{align*}
    \pi\left( \left(1+\frac{1}{2\ln^2 q'} \right)q'\right) &\sim \frac{2^n\left(1+\frac{1}{2(\ln2)^2 n^2}\right)}{n+\log\left(1+\frac{1}{2(\ln2)^2 n^2}\right)}
     \sim 2^n\left( \frac{1}{n} + \frac{1}{2(\ln2)^2 n^3} + O(n^{-4}) \right)
\end{align*}
and
\begin{align*}
    \pi\left( \left(1+\frac{1}{2\ln^2 q'}\right)q'\right) - \pi(q') &= \frac{2^n}{2(\ln 2)^2 n^3 } + O(2^n n^{-4}).
\end{align*}
Therefore, the density of primes in the range is 
\begin{align*}
    \rho = \frac{\pi\left( \left(1+\frac{1}{2\ln^2 q'}\right)q'\right) - \pi(q')}{ q'\frac{1}{2\ln^2 q'} } \sim \frac{2^n \frac{1}{2(\ln2)^2 n^3}}{2^n \frac{1}{2(\ln2)^2 n^2}} = \frac{1}{n}+O(n^{-2}),
\end{align*}
which immediately implies that a prime can be found with an expected number of $O(n)$ random samples. All of this is incorporated in the $\Gen$ procedure as that is responsible for choosing a suitable $q$.
As will also be mentioned later, since $q$ is close to a power of 2, when the prover has to create an equal superposition over the elements of $\mathbb{Z}_q$ it will instead create the superposition over elements up to $q'$, the nearest power of 2, larger than $q$. The resulting state will be sufficiently close in trace distance that we only incur a $1/poly(n)$ penalty in completeness for making this replacement.

\subsection{Phase encoding}
\label{sec:phase_enc}
The concept of phase encoding was described in Section~\ref{sect:prelim}. 
In this section we will look at several properties of the phase encoding for the LWR-based NTCF (Definition~\ref{def:lwrntcf}). We aim to show how to evaluate $g(b,x) =\mathbf{A}x+b \cdot (\mathbf{A}s+e)$ in phase, and show that measuring the resulted state in Hadamard  basis will reveal the value of $f(b,x)=\lfloor  g(b,x)\rfloor_p$, with high probability.

It is natural to start by considering the phase encoding of $g(b,x)$ for a specific $(b,x)$. Note that $x\in \mathbb{Z}_q^n$ and $g(b,x)\in \mathbb{Z}_q^m$, both being vectors. The phase encoded state that we would like the prover to prepare (for each $b$ and $x$) should have the following form:
\begin{equation}
\label{eq:phase_enc_easiest}
    \ket{\phi(b,x)} = \bigotimes_{i=1}^m \ket{\phi_i(b,x)}
\end{equation}
with
\begin{equation}
   \ket{\phi_i(b,x)} = \frac{1}{\sqrt{2}}(\ket{\bar{0}}+e^{i\phi_i(b,x)}\ket{\bar{1}})
\end{equation} 
and
\begin{equation}
    \phi_i(b,x) = \frac{2\pi g_i(b,x)}{q} - \frac{\pi}{2}
\end{equation}
where $g_i$ represents the $i$'th component of $g(b,x)$.

For the majority of this section, we will focus on the case $p=2$. That is, we assume that $f(b, x)$ simply takes the most significant bit of each component of $g(b, x)$. This, of course, is not the NTCF we defined since there we had that $p=O(\sqrt{mn\log q})$. We will address the case of general $p$ in Subsection~\ref{subsect:generalp}.

For $p=2$, we denote the output of $f(b,x)=\lfloor g(b,x) \rfloor_2$ by $y$, a binary string of length $m$. We have $y_i = \lfloor g_i(b,x) \rfloor_2$ where $y_i$ is the $i$'th bit of $y$ and $g_i(b,x)$ is the $i$'th component of $g(b,x)$.
Before explaining how to prepare the phase encoded state in constant depth, let us first investigate how to decode $y=f(b,x)$ from $\ket{\phi(b,x)}$ with high probability.

\subsubsection{Decoding by measurements}
\label{sec:dec_by_measure}
The phase encoding can be \emph{probabilistically decoded} through Hadamard measurements. Denote the process of measuring the $XX...X$ observable on the state in Equation~\ref{eqn:hadamard_phase_encoding} by $M$ and the measurement outcomes of all $m$ phase encoded states by $z \in \{0, 1\}^m$. One can then write $z\leftarrow M(\ket{\phi(b,x)})$. It should be clear that $z=y$ indicates that the decoding was completely successful.

Let us consider the case of a single component in the encoding, namely $\ket{\phi_i}$. In order to investigate the possible values of $z_i=M(\ket{\phi_i})$, $\ket{\phi_i}$ can be rewritten as
\begin{align}
\label{eqn:hadamard_phase_encoding}
    \ket{\phi_{i}}&=\frac{1}{\sqrt{2}}(\ket{\bar{0}}+e^{i\phi_i}\ket{\bar{1}})\\
    &=\frac{1}{2}((1+e^{i\phi_i})\ket{\bar{+}}+(1-e^{i\phi_i})\ket{\bar{-}}).
\end{align}

\noindent If the qubit is measured in the Hadamard basis, we can express the outcome probabilities as
\begin{equation}
    \PrM(\pm|\ket{\phi_i}) = \frac{1}{4}[(1\pm\cos\phi_i)^2+\sin^2 \phi_i]=\frac{1}{2}(1\pm\cos\phi_i).
\end{equation}
with $\phi_i = \frac{2\pi g_i}{q} - \frac{\pi}{2}$. Note that $g_i<q/2$ is equivalent to $y_i=\lfloor g_i \rfloor_2$ = 0. Additionally, $g_i<q/2$ leads to $\cos\phi_i>0$. Therefore the probability of getting $+$ is larger than that of $-$. If we map $+$ to 0 and $-$ to 1, it is clear that the Hadamard measurement is essentially a probabilistic decoding of $y_i$ from $\phi_i$, with success probability always greater than $\frac{1}{2}$. More compactly, we can write the probability of measuring any $z_i$ from $\ket{\phi_i}$ by
\begin{equation}
    \PrM(z_i|\ket{\phi_i}) = \frac{1}{2}(1 + (-1)^{z_i} \cos\phi_i).
\end{equation}
Furthermore, the probability of \emph{successfully decoding} $\phi_i$ (i.e. $z_i = y_i$) is denoted by
\begin{equation}
    \pcor(\phi_i):=\Pr(z_i=y_i) =\PrM(y_i|\ket{\phi_i}).
\end{equation}
where $\Pr(z_i=y_i) = \PrM(+|\phi_i)=  \frac{1}{2}(1+\cos\phi_i)$ if $y_i=0$ and $\Pr(z_i=y_i) =\PrM(-|\phi_i) = \frac{1}{2}(1-\cos\phi_i)$ if $y_i=1$. Similarly, the probability of unsuccessful decoding is represented by 
\begin{equation}
    \pinc(\phi_i) :=\Pr(z_i\neq y_i) = \PrM(\neg y_i | \ket{\phi_i}) = 1-\pcor(\phi_i).
\end{equation}

We can now evaluate the expected values of these probabilities over the uniform choice of the matrix $\mathbf{A}$ and show the following:

\begin{lemma}
Over the choice of matrix $\mathbf{A}$, the average probability of successful decoding of any $\ket{\phi_i}$ is $\frac{1}{2}+\frac{1}{\pi}\approx 0.82$.
\end{lemma}

\begin{proof}
To clarify, there are two sources of randomness here. On the one hand we have the randomness of the measurement and on the other hand we have the random choice of the matrix $\mathbf{A}$. We're interested in seeing the expected probability of a successful (as well as an unsuccessful) decoding over the choice of $\mathbf{A}$. As $g(b, x) = \mathbf{A}x + b \cdot (\mathbf{A}s + e)$, we can see that if $\mathbf{A}$ is uniform (over a finite field), then $g(b, x)$ will also be uniform (for any non-zero $b$ and $x$).
Hence, $\Pr(\phi_i)=\Pr(g_i)=\frac{1}{q}$ for all $\phi_i\in \{-\frac{\pi}{2},\frac{2\pi}{q}-\frac{\pi}{2},\dots,\frac{3\pi}{2} \}$.
The expected probability of a correct decoding is then

\begin{equation}
\label{eq:pc1}
\begin{aligned}
    \bpcor:=\EA(\pcor(\phi_i)) &=\sum_{g_i=0}^{q/2-1} \Pr(\phi_i) \PrM(+|\phi_i) + \sum_{y_i=q/2}^{q-1} \Pr(\phi_i) \PrM(-|\phi_i)\\
    & = 2\sum_{g_i=0}^{q/2-1} \Pr(\phi_i) \PrM(+|\phi_i)\\
    &= 2 \sum_{g_i=0}^{q/2-1} \frac{1}{q}\frac{1}{2}(1+\cos\phi_i):=S\\
\end{aligned}    
\end{equation}
which we can view as a \emph{Riemann sum}. For large $q$, the summation converges to an integral
\begin{equation}
    \bpcor = S \rightarrow I := 2\int_0^{\frac{q}{2}-1} \frac{1}{2q} \left( 1+\cos \left( \frac{2\pi g_i}{q}-\frac{\pi}{2} \right) \right)dg_i.
\end{equation}    
By the change of variable  $\phi_i = \frac{2\pi g_i}{q} - \frac{\pi}{2}$, this becomes
\begin{align}
        \bpcor \rightarrow I&=\frac{1}{\pi} \int_{-\frac{\pi}{2}}^{\frac{\pi}{2}} \frac{1}{2} (1+\cos(\phi_i))d\phi_i\\
    &=\frac{1}{2}+\frac{1}{\pi} \sim 0.82.
\end{align}
We also have the expected probability of an incorrect decoding
\begin{equation}
    \bpinc:=\EA(\pinc(\phi_i)) \rightarrow 1 - \bpcor = \frac{1}{2}-\frac{1}{\pi} \sim 0.18.
\end{equation}
The approximation $S\rightarrow I$ comes with an error which we can bound. Such an error for an $(l+1)$-order differentiable integrand $\chi$ can be determined with the Euler-Maclaurin formula
\begin{equation}
    S - I = \sum_{k=1}^{l}\frac{B_k}{k!}\left(\chi^{(k-1)}\left(\frac{q}{2}-1 \right) - \chi^{(k-1)}(0) \right)+R_l
\end{equation}
where $B_k$ is the $k$-th Bernoulli number, $R_l=o(q^{-l})$ is the remainder term, and $\chi(y_i) = \frac{1}{q}(1+\cos(\frac{2\pi g_i}{q}-\frac{\pi}{2}))$ is the integrand. We can see that $\chi^{(k-1)}(\frac{q}{2}-1) - \chi^{(k-1)}(0) = 0$ for odd $k$. Therefore, the error can be written as
\begin{equation}
\begin{aligned}
    S - I &= \frac{B_2}{2}\frac{1}{q}\frac{2\pi}{q}\left[-\sin \left(\frac{\pi}{2}-\frac{2\pi}{q} \right) + \sin \left(-\frac{\pi}{2} \right)\right] + o(q^{-2})\\
    & = -\frac{1}{3q^2}+o \left(q^{-2} \right) =O(q^{-2}).
\end{aligned}    
\end{equation}
\end{proof}

As $g(b,x)$ is uniform (over the random choice of $\mathbf{A}$ and whenever $(x, b) \neq (0, 0)$), each of its components will be a uniform value in $\mathbb{Z}_q$. Thus, we can view the measurement of each component of $\ket{\phi(b, x)}$ to be an independent and identically distributed random variable. As the expected probability of a correct decoding is $0.82$, it follows from a Chernoff bound that $0.82m$ values will be decoded correctly, with overwhelming probability over the choice of $\mathbf{A}$.
While this means that most values are correctly decoded, we, in fact, need \emph{all} values to be decoded correctly with high probability.
To achieve this, we use a \emph{classical repetition code} and repeat each output component several times in order to take a majority vote.

\subsubsection{Decodability and repetition code ($p=2$)}
\label{sec:rep-and-majority}
Instead of the prover having to prepare $\ket{\phi(b, x)}$ (for each $b$ and $x$), we will instead ask it to prepare:
\begin{equation}
\label{eq:rep_code_version}
    \ket{\phi(b, x)} = \bigotimes_{i=1}^m \ket{\phi_i(b, x)}^{\otimes v} = \bigotimes_{i=1}^m \left(\frac{1}{\sqrt{2}}(\ket{0}+e^{i\phi_i}\ket{1})\right)^{\otimes v}
\end{equation}
where $v$ represents the number of repetitions. In this case, to decode the value of the $i$'th component, one measures all $v$ copies of that component and uses the majority outcome as the value $z_i$. 

We say that one component, for instance the $i$'th component, has been correctly decoded, if $z_i = y_i$, where recall that $y_i$ is the most-significant bit of $g_i(b, x)$. By analogy, we say that the whole state has been correctly decoded if all of its components were (i.e. $z = y$).
Our goal is to find the relation between $v$ and $m$ such that $z=y$ with sufficiently high probability (say, $99\%$) \emph{for most states} $\ket{\phi(b, x)}$ (say, $99\%$ of all such states). In doing so, we show the following

\begin{theorem} \label{thm:decodability}
At least $99\%$ of all $\ket{\phi(b,x)}$ states can be correctly decoded with probability $99\%$, whenever $v=\Omega(m^2\log m)$.
\end{theorem}
\begin{proof}

Without loss of generality, we focus on the case of $g_i < \frac{q}{2}$, that is $y_i = 0$. Recall that
\begin{equation}
    \pcor(\phi_i) = \PrM(+|\phi_i) = \frac{1}{2}(1+\cos(\phi_i)) = \frac{1}{2}\left(1+\sin\left(\frac{2\pi g_i}{q}\right)\right).
\end{equation}

It should be clear that for the very special case $g_i=0$, the probability of having the correct measurement outcome is $\frac{1}{2}$. In this case, it is impossible to tell if $z_i$ should be 0 or 1 even with repetition, because no matter how large $v$ is, there will always be an equal number of correctly and incorrectly decoded bits, on average. Therefore, any component $g_i$ that is extremely close to $0$ or $\frac{q}{2}$ so that $\pcor(\phi_i)$ is close to $\frac{1}{2}$ would make the whole $\ket{\phi(b, x)}$ state \emph{undecodable}\footnote{In fact, even if we ignore the cases where $\pcor(\phi_i)=\frac{1}{2}$, it is still required to have $v=O(q)$ to distinguish between $\phi_i=\frac{2\pi}{q}-\frac{\pi}{2}$ and $\phi_i=-\frac{2\pi}{q}+\frac{3\pi}{2}$ where $g_i = 1$ and $g_i=q-1$, respectively. This is clearly unacceptable since $q$ is exponential in $n$ and the resulting circuit would be exponentially wide.}. 

To be more explicit, we will consider $\ket{\phi_i}$ to be undecodable whenever we have that either $|g_i| < \frac{q}{cm}$ or $|g_i - q/2| <  \frac{q}{c m}$, for a constant $c > 0$ to be determined later. But as noted before, for a uniform $\mathbf{A}$, each $g_i$ (excluding the case $g(0, 0)$) is also uniform in $\mathbb{Z}_q$. It follows that the probability that $g_i$ leads to an undecodable $\ket{\phi_i}$ is at most $\frac{1}{q} \frac{4q}{cm} = \frac{4}{cm}$, over the choice of $\mathbf{A}$. From a union bound, we then also have that the probability of $\ket{\phi(b, x)}$ to be undecodable (i.e. at least one of its components is undecodable) is at most $m \frac{4}{cm} = \frac{4}{c}$.
This means that at least a fraction $1 - \frac{4}{c}$ of all $\ket{\phi(b, x)}$ states are, in fact, decodable. That is, all of their components are at least $\frac{q}{cm}$ away from the undecodability boundary. By taking $c = 400$, we have that $99\%$ of $\ket{\phi(b, x)}$ are decodable.

Without loss of generality, let's now consider a state that is barely decodable, with say $g_i = \frac{q}{c m}$. The probability of correctly decoding the corresponding $\ket{\phi_i}$ state will be
\begin{equation}
    \pcor(\phi_i) = \frac{1}{2}\left(1+\sin\left(\frac{2\pi g_i}{q}\right)\right) \approx \frac{1}{2} \left( 1 + \frac{1}{O(m)} \right).
\end{equation}
The state is biased away from $1/2$ by $1/O(m)$. From an application of the Chernoff-Hoeffding bound\footnote{Each measurement is viewed as an i.i.d.~random variable. The empirical mean of these variables is expected to be close to $1/2 + 1/O(m)$. Chernoff-Hoeffding tells us that a deviation of $\epsilon$ from this expected value occurs with probability $\exp(-v \epsilon^2)$. Thus, since the case of interest is $\epsilon = 1/O(m)$, we can see that to have a constant probability of incorrectly decoding, it must be that $v = \Omega(m^2)$.} it follows that repeating the measurement $\Omega(m^2)$ times and taking a majority vote is enough to ensure that the value is correctly decoded with \emph{constant probability} (say $99\%$). Of course, we want that \emph{all} $m$ values are correctly decoded which means that we should take the number of repetitions $v$ so that the probability of correctly decoding one value is at least $1 - 1/O(m)$. Once again, we can use Chernoff-Hoeffding and find that $v = \Omega(m^2 \log m)$. As the probability of incorrectly decoding one value is now $1/O(m)$, from a union bound the probability of incorrectly decoding \emph{any} of the $m$ values is $O(1)$. By suitably choosing the constant factors, we can set this probability to be, say $1\%$.
We therefore have that $v = \Omega(m^2 \log m) = \Omega(n^2 \log m \log^2 q) = \Omega(n^4 \log n)$.
\end{proof}

\subsubsection{Phase encoding for general $p$} \label{subsect:generalp}

The analysis from the previous subsections was concerned with the case $p=2$. We now adapt this to the general case of $p = O(\sqrt{m n \log q})$.

As we expect $p$ to be a power of 2, the rounding $\lfloor g_i \rfloor_p$ for any value of $g_i$ is exactly a $(\log_2 p)$-bit number. What we have been doing so far with the phase encoding is to encode the most significant bit of $f_i = \lfloor g_i \rfloor_p$ in phase. What about the other $\log_2 p - 1$ bits? The solution is simply to phase encode those bits as well.

\begin{lemma}
Applying the phase encoding to the $\log_2 p$ significant bits of every $g_i \in \mathbb{Z}_q$, leads to a repetition factor $v= \Omega(n^4 \log^2 n)$ in order to achieve the same guarantees as Theorem~\ref{thm:decodability}.
\end{lemma}
\begin{proof}
Specifically, the $k$'th significant bit of $g_i$ can be encoded as
\begin{equation}
    \ket{\phi_{i,k}} = \frac{1}{\sqrt{2}}(\ket{\bar{0}} + e^{i\phi_{i,k}}\ket{\bar{1}})
\end{equation}
with
\begin{equation}
    \phi_{i,k} = \frac{2^k \pi g_i}{q}.
\end{equation}

How does this affect the decodability results of the previous sections? The expected probability of decoding a single bit, without repetition, will still be negligibly close to $0.82$. This is because, as we saw in Subsection~\ref{sec:dec_by_measure}, the deviation from this expectation is inverse in the square of the field size, which is now $\sim \frac{q}{2^k}$. As $k \leq \log_2 p$, $p = O(\sqrt{m n \log q})$ so that $2^k = O(\sqrt{m n \log q})$ and $q = 2^{O(n)}$, the deviation from the expected value of $0.82$ remains negligible in $n$ (or $\lambda$).

The decodability boundary, from Subsection~\ref{sec:rep-and-majority}, also changes from $\frac{q}{cm}$ to $\frac{q}{2^k cm}$. As $2^k = O(\sqrt{m n \log q})$ and $m = \Omega( n \log q)$, the boundary becomes $\frac{q}{c' n^4}$, for some constant $c' > 0$.
Following the same steps as in Subsection~\ref{sec:rep-and-majority}, to ensure that most states can be correctly decoded, we see that the number of repetitions remains $\Omega(n^4)$. But this is just for the $m$-bit vector containing the $k$'th most significant bit of each component. As we have $\log_2 p$ such vectors, and we want all of them to be decoded correctly, we need to add an additional $\log_2 p$ factor so that overall we have $v = \Omega(n^4 \log n \log_2 p) = \Omega(n^4 \log^2 n)$.
\end{proof}

Thus, for each $b$ and $x$, the state the prover will prepare is
\begin{equation} \label{eqn:truephaseencoding}
    \ket{\phi(b, x)} = \bigotimes_{i=1}^m \bigotimes_{k=1}^{\log_2 p} \left( \ket{\bar{0}} + e^{i\phi_{i,k}} \ket{\bar{1}} \right)^{\otimes v}.
\end{equation}

\subsubsection{Constant-depth circuit implementation} \label{sec:const-depth-phase-enc}
Here we show that the phase encoding construction can be performed in constant quantum depth.

\begin{theorem}
It is possible to prepare the state in Equation~\ref{eqn:truephaseencoding} in constant quantum depth and with logarithmic depth classical computation.
\end{theorem}
\begin{proof}
We've already mentioned that cat states can be prepared in constant quantum depth with one quantum-classical interleaving. Let us then assume that we have sufficient cat states (of a size that will be determined later) and see how we can apply the required phases in constant quantum depth.

Recall that $g(b,x) =\mathbf{A}x+b \cdot (\mathbf{A}s+e)$, and determines the phase\footnote{We again focus only on the case of the most significant bit, as the $k$'th most significant bit can be obtained by simply mapping $q$ to $q/2^k$.} $\phi_i = \frac{2\pi g_i}{q}-\frac{\pi}{2}$.
The phase can then be expressed as
\begin{equation}
\begin{aligned}
    e^{i\phi_i} &= \exp \left(-\frac{\pi i}{2} \right) \exp \left(bi\frac{2\pi (\mathbf{A}s)_i+ 2\pi e_i}{q} \right) \exp \left(\frac{2\pi i}{q} \sum_{j=1}^n A_{ij} x_j \right)\\
    &= \exp(\phi'_i(b)) \prod_{j=1}^n \exp \left(\frac{2\pi i}{q} A_{ij} x_j \right)
\end{aligned}
\end{equation}
where
\begin{align}
    \exp(\phi'_i(b)) := \exp \left(-\frac{\pi i}{2} \right) \exp \left(bi\frac{2\pi (\mathbf{A}s)_i+ 2\pi e_i}{q} \right).
\end{align}

Note that $\phi'_i$ only depends on $b$ and not on $x$. Having multiple copies of $b$, we can easily apply a $\phi'_i$ rotation in parallel using $Z$-rotations ($R_z$) and  controlled-$Z$-rotations ($CR_z$):

\begin{equation}
    R_z \left(-\frac{\pi}{2} \right) CR_z \left(\frac{2\pi (\mathbf{A}s)_i+ 2\pi e_i}{q} \right) \frac{1}{\sqrt{2}}(\ket{\bar{b}}\ket{\bar{0}}+\ket{\bar{b}}\ket{\bar{1}}) = \ket{\bar{b}} \otimes \frac{1}{\sqrt{2}}(\ket{\bar{0}}+e^{i\phi'_i(b)}\ket{\bar{1}}).
\end{equation}

\noindent The corresponding circuit is shown in Figure~\ref{fig:circuit_phiprime}.
\begin{figure}[h]
    \centering
\begin{quantikz}
\lstick{$\mathsf{X}_0$} & \ctrl{1} & \qw   \\
\lstick{$\mathsf{Z}_{i,0}$} & \gate{R_z \left(\frac{2\pi (\mathbf{A}s)_i+ 2\pi e_i}{q} \right)} & \qw   \\
\lstick{$\mathsf{Z}_{i,1}$} & \gate{R_z(-\frac{\pi}{2})} & \qw 
\end{quantikz}
\caption{The quantum circuit for the vector addition operations in phase encoding. Here $\mathsf{X}_0$ is the first qubit of the $\mathsf{X}$ register that stores information of $b$. $\mathsf{Z}_{i,j}$ is the $j$'th qubit of the $i$'th cat state which stores information of $\phi_i$.}
\label{fig:circuit_phiprime}
\end{figure}
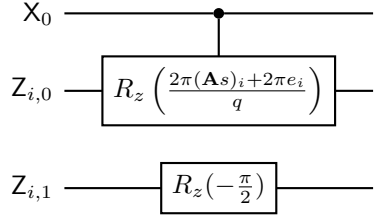

We now need to implement the phase-encoded matrix-vector multiplication in parallel on the cat state. Note that $x_j$ is a non-negative integer less than $q$ and it can be expanded as
\begin{equation}
    x_j = \sum_{k=0}^{\lceil \log_2(q) -1 \rceil} 2^{k}x_{j,k}
\end{equation}
denoting the $k$'th significant bit of $x_j$ by $x_{j,k}$. The phase can be further expanded:
\begin{align}
   \prod_{j=1}^n \exp \left(\frac{2\pi i}{q} A_{ij} x_j \right)= \prod_{j,k} \exp \left(\frac{2\pi i}{q} 2^k A_{ij} x_{j,k} \right).
\end{align}
Therefore, the desired phase can be applied to the cat state by parallel controlled-$Z$-rotation gates in constant-quantum depth. Specifically,
\begin{equation}
\begin{aligned}
   \left( \prod_{j=1}^n \prod_{k=0}^{{\lceil \log_2(q) -1 \rceil}} CR_z \left(\frac{2\pi}{q}2^k A_{i,j} \right)\right) \frac{1}{\sqrt{2}} (\ket{\overline{x_{j,k}}}\ket{\bar{0}}+e^{i\phi_i'(b)}\ket{\overline{x_{j,k}}}\ket{\bar{1}}) &= \\
\ket{\overline{x_{j,k}}} \otimes \frac{1}{\sqrt{2}} & (\ket{\bar{0}}+e^{i\phi_i(b, x)}\ket{\bar{1}})
\end{aligned}
\end{equation}
where the $CR_z$ gates can be performed in parallel if the size of cat is $\Omega(n\log q)=\Omega(n^2)$. The local quantum circuit for multiplying $A_{i,j}$ with the $k$'th significant bit of $x_j$ is shown in Figure~\ref{fig:circuit_CRzs}.

\begin{figure}[h]
    \centering
\begin{quantikz}
\lstick{$\mathsf{X}_{j,k}$} & \ctrl{1} & \qw \\
\lstick{$\mathsf{Z}_{i,l}$} & \gate{R_z \left(\frac{2\pi}{q}2^k A_{i,j} \right)}  & \qw 
\end{quantikz}
\caption{Part of the quantum circuit for matrix-vector multiplication in phase. Here $\mathsf{X}_{j,k}$ is the qubit that stores the $k$'th bit of $x_j$, and $\mathsf{Z}_{i,l}$ is the $l$'th qubit of the cat state storing the information of $\ket{\phi_i}$.}
\label{fig:circuit_CRzs}
\end{figure}
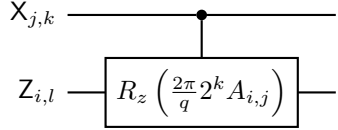
Thus, all operations can be performed in constant quantum depth.
\end{proof}

It is worth noting that in current physical realizations of quantum computers, these (controlled) rotations can be performed directly by tuning microwave frequencies for superconducting qubits~\cite{qipsc} or laser frequencies for trapped-ions~\cite{Bruzewicz2019TrappedionQC}. Alternatively, if one insists on having a fixed-size gate set, \cite{quantumfanout} provides a constant-depth implementation with $1/poly$ error which is also acceptable.

The Hadamard measurements discussed in the previous sections are performed by measuring $X$ on each qubit of a phase encoded cat state and then taking the parity of the outcomes.

\subsection{LWR-based protocol with phase encoding}
\label{sec:phase_enc_complete}
The protocol using the LWR-based NTCF and the phase encoding is outlined in Figure~\ref{fig:bcmvvprotocol2}. The verifier behaves essentially the same as in the BCMVV protocol. The major difference is in the prover's honest strategy, which requires it to perform the constant-depth evaluation of the phase encoding.

\protocol{Modified BCMVV protocol}
{Honest provers' strategy for the constant-depth version of the BCMVV protocol \cite{brakerski2018cryptographic} based on phase encoding.}{fig:bcmvvprotocol2}{
Let $\lambda = n$ be a security parameter and $N \geq 1$ a number of rounds. The parties taking part in the protocol are a PPT machine, known as the verifier and a QPT machine, known as the prover. They will repeat the following steps $N$ times:
\begin{enumerate}
\item The verifier generates $(k,t_k)\leftarrow \Gen(1^\lambda)$. It sends $k$ to the prover.
\item The prover uses $k$ to implement the phase encoding of the function $g_k(b,x)$, and prepare the following state:
\begin{equation*}
    \ket{\psi} = \frac{1}{\sqrt{2q^n}}\sum_{b \in \{0, 1\}, x\in \mathbb{Z}_q^n } \ket{\bar{b}}_{\mathsf{B}}\ket{\bar{x}}_{\mathsf{X}}\ket{\phi(b,x)}_{\mathsf{Z}}\;
\end{equation*}
with
\begin{equation*}
     \ket{\phi(b, x)} = \bigotimes_{i=1}^m \bigotimes_{k=1}^{\log_2 p} \left( \ket{\bar{0}} + e^{i\phi_{i,k}} \ket{\bar{1}} \right)^{\otimes v}.
\end{equation*}
where $\phi_{i,k}(b,x) = \frac{2\pi 2^k g_i(b,x)}{q}-\frac{\pi}{2}$.
The prover then measures the $\mathsf{Z}$ register in Hadamard basis. By conducting majority votes for the parities of the Hadamard measurement outcome of every block $\left(\ket{\bar{0}}+e^{i\phi_{i,k}}\ket{\bar{1}}\right)^{\otimes v}$, the prover obtains a new string $y\in\{0,1\}^{m \log_2 p}$ which it sends to the verifier. The remaining state is
\begin{equation}
    \ket{\psi_y} = \sum_{b\in\{0,1\}} \ket{\bar{b}}\ket{\bar{x}_b}\ket{y}.
\end{equation}
\item The verifier selects a uniformly random challenge $c\leftarrow_R \{0,1\}$ and sends $c$ to the prover.
\item \begin{enumerate}
    \item (Preimage test:) When $c=0$, the prover measures in the standard basis the $\mathsf{BX}$ registers of the state leftover in step 2. It obtains the outcomes $b \in \{0, 1\}$ and $x\in\{0,1\}^{poly(n)}$, which it sends to the verifier. If $f_k(b,x)=y$, the verifier sets $N_c\leftarrow N_c + 1$.
    \item (Equation test:) When $c=1$, the prover measures each qubit in the $\mathsf{BX}$ register in the Hadamarad basis. It obtains the outcomes $b' \in \{0, 1\}^{k'}$ and $d\in\{0,1\}^{poly(n)\cdot k'}$ which it sends to the verifier. Here $k'$ denotes the size of a cat state (used to encode $b$ and each bit in $x$). The verifier computes $(x_0,x_1) = \LWRInv(t_k,y)$ and sets $N_c\leftarrow N_c+1$ if $d \cdot (\bar{x}_0\oplus \bar{x}_1) = b''$ where $b''$ is the $\mathrm{xor}$ of all $k'$ bits of $b'$.
\end{enumerate}
\end{enumerate}
At the end of the $N$ rounds, if $\frac{N_c}{N}> 0.95$, the verifier accepts.
}

As we saw in the previous subsections, due to the randomness over the choice of $\mathbf{A}$ and the probabilistic nature of the measurements, the protocol is not perfectly complete. That is, the success probability for the honest prover is no longer $100\%$ as in the original BCMVV protocol. Before accounting for all sources of ``imperfections'' we first need to examine the post-measurement state in the preimage register after the prover performs step 2 in the protocol. Ideally, we would like this state to be as close as possible to an equal superposition over valid preimages. Thus, in the next subsection we compute a bound on the fidelity of the true state with respect to an ideal state.

\subsubsection{Fidelity of the post-measurement state and the success probability for an honest prover} \label{subsect:fidelity}

We wish to determine the success probability of an honest prover in the protocol. To do so, we need to characterize the prover's state after it measures the phase-encoded image register. We will show that the state in the preimage register (post-measurement of the phase-encoded image register) has high overlap with the ``ideal'' preimage state that would have be obtained if the prover performed the evaluation in the computational basis, rather than in phase. With this result, we can then compute the protocol's completeness in the next subsection.

To start the proof we will consider splitting the prover's measurement of the image register into two steps. First, the prover measures in the Hadamard basis all but one qubit from \emph{each} phase encoded state in the image register. Then, it measures the remaining unmeasured qubits as well. This separation is fictitious, as in the protocol the prover will measure all qubits of the image register in one step. But performing this separation and considering the prover's state after it measures all but one qubit of each phase encoded state will make the analysis simpler.
Let us begin with the honest prover's state after performing the coherent evaluation of the function in phase,

\begin{equation}
\label{eq:whole_cov_phase_enc}
    \ket{\psi} = \frac{1}{\sqrt{2q^n}} \sum_{b\in\{0,1\}} \sum_{x\in \mathbb{Z}_q^n} \ket{\bar{b},\bar{x}}_\mathsf{BX} \ket{\phi(b,x)}_\mathsf{Z}
\end{equation}
where, as before,
\begin{equation}
    \ket{\phi(b,x)} = \bigotimes_{i=1}^{m} \bigotimes_{k=1}^{\log_2 p} \ket{\phi_{i,k}(b,x)}^{\otimes v}.
\end{equation}
Also recall that each component $\ket{\phi_{i,k}}$ has the form of a rotated cat state
\begin{equation}
    \ket{\phi_{i,k}(b,x)} = \frac{1}{\sqrt{2}}(\ket{\bar{0}}+e^{i\phi_{i,k}} \ket{\bar{1}}).
\end{equation}
The prover will measure each qubit of such a state (or, more precisely, of the coherent superposition of such states) in the Hadamard basis. It should be clear that when measuring all but one qubit in the Hadamard basis, the state of that qubit becomes
\begin{equation}
    \ket{\tilde{\phi}_{i,k}(b,x)} = \frac{1}{\sqrt{2}}(\ket{0} \pm e^{i\phi_{i,k}}\ket{1}),
\end{equation}
where the $\pm$ relative phase is determined by the parity of the Hadamard basis measurement outcomes. Without loss of generality, let us fix\footnote{We can do this because, as we will see, this is equivalent to the prover having to flip the outcome of one of the measurements it performs. Alternatively, the prover can always perform a quantum-classical interleaving here in order to flip the phase, though this is not necessary.} this phase as $+$.

We now rewrite each component $\ket{\tilde{\phi}_{i,k}}$ as
\begin{equation}
\begin{aligned}
    \ket{\tilde{\phi}_{i,k}(b,x)} =& \frac{1}{\sqrt{2}}(\ket{0}+e^{i\phi_{i,k}}\ket{1}) \\
    &= \alpha(0|\phi_{i,k}) \sqrt{\PrM(+|\ket{\phi_{i,k}})}\ket{+} + \alpha(1|\phi_{i,k}) \sqrt{\PrM(-|\ket{\phi_{i,k}})}\ket{-}\\
    &\rightarrow^H \alpha(0|\phi_{i,k})\sqrt{\PrM(0|\ket{\phi_{i,k}})}\ket{0} + \alpha(1|\phi_{i,k})\sqrt{\PrM(1|\ket{\phi_{i,k}})}\ket{1}
\end{aligned}
\end{equation}
where in the last line we mapped from the Hadamard basis $\{\ket{+},\ket{-}\}$ to the computational basis $\{\ket{0},\ket{1}\}$, and $\alpha(0|\phi_{i,k})$ and $\alpha(1|\phi_{i,k})$ are pure phases (i.e. $|\alpha(0|\phi_{i,k})| = |\alpha(1|\phi_{i,k})| = 1$). 
Let us now consider what happens when all of these qubits are measured. Let $\tilde{z}\in\{0,1\}^{mv\log_2 p}$ denote the Hadamard measurement outcome of all $mv\log_2 p$ $\ket{\tilde{\phi}_{i,k}}$ states.
This string can be expressed as a concatenation of $m\log_2 p$ substrings $\tilde{z}_{i,k}\in \{0,1\}^v$ for $i\in\{1,\dots,m\}$ and $k\in\{1,\dots,\log_2 p\}$. The substring with index $i,k$ represents the measurement outcomes of $\ket{\tilde{\phi}_{i,k}}^{\otimes v}$. We can then write the state as

\begin{equation}
\begin{aligned}
    \ket{\tilde{\phi}_{i,k}(b,x)}^{\otimes v} &\rightarrow^H \left( \alpha(0|\phi_{i,k})\sqrt{\PrM(0|\ket{\phi_{i,k}})}\ket{0} + \alpha(1|\phi_{i,k})\sqrt{\PrM(1|\ket{\phi_{i,k}})}\ket{1} \right)^{\otimes v} \\
    &= \sum_{\tilde{z}_{i,k} \in \{0,1\}^v} \left( \prod_{j=1}^v \alpha(\tilde{z}_{i,k,j}|\phi_{i,k}) \sqrt{\PrM(\tilde{z}_{i,k,j}|\ket{\phi_{i,k}})} \right) \ket{\tilde{z}_i}\\
    &= \sum_{\tilde{z}_{i,k} \in \{0,1\}^v} \alpha(\tilde{z}_{i,k}|\phi_{i,k},v) \sqrt{\PrM\left(\tilde{z}_{i,k}|\ket{\phi_{i,k}}^{\otimes v}\right)}  \ket{\tilde{z}_i}
\end{aligned}
\end{equation}
where $\tilde{z}_{i,k,j}$ denotes the $j$'th bit of the substring $\tilde{z}_{i,k}$, and $\alpha(\tilde{z}_{i,k}|\phi_{i,k},v)$ is the product of the pure phases $\alpha(\tilde{z}_{i,k,j}|\phi_{i,k}) $ with $j$ ranging from 1 up to $v$. The entire phase encoded state $\ket{\tilde{\phi}(b, x)}$ can then be expressed as:
\begin{equation}
    \begin{aligned}
        \ket{\tilde{\phi}(b,x)} \rightarrow^H \sum_{\tilde{z}\in\{0,1\}^{mv\log_2 p}} \alpha(\tilde{z}|\phi) \sqrt{\PrM(\tilde{z}|\ket{\phi(b,x)})}\ket{\tilde{z}}.
    \end{aligned}
\end{equation}
Finally, the state of the coherent phase encoding evaluation in Equation~\ref{eq:whole_cov_phase_enc} (but after the prover has measured all but one qubit of each phase-encoded cat state) can be expressed as well:
\begin{equation}
    \begin{aligned}
    \ket{\tilde{\psi}} \rightarrow^H \frac{1}{\sqrt{2q^n}}\sum_{b,x}\ket{\bar{b},\bar{x}}_\mathsf{BX}\sum_{\tilde{z}} \alpha(\tilde{z}|\phi(b,x)) \sqrt{\PrM(\tilde{z}|\ket{\phi(b,x)})}\ket{\tilde{z}}_\mathsf{Z}.
    \end{aligned}
\end{equation}
Recall that we aim to estimate the success probability of an honest prover. To do so, we can first find an \emph{ideal} state such that, if the prover holds that state, it would very likely succeed in the protocol. The success probability can therefore be estimated by evaluating the fidelity between the real and the ideal states, then evaluating the success probability if the prover holds the ideal state. Denoting the ideal state by $\ket{\psi_\mathrm{ideal}}$ and the procedure of majority voting by $\mathrm{Maj}$\footnote{In other words, $Maj(\tilde{z})$ will be a string of $m \log_2 p$ bits containing the majority value of each substring of $v$ bits.}, we let
\begin{equation} \label{eqn:psiideal}
    \begin{aligned}
        \ket{\psi_\mathrm{ideal}}& = \frac{c}{\sqrt{2q^n}} \sum_{x_0 \in \mathbb{Z}_q^n} \sum_{\maj(\tilde{z})=f(0,x_0)} \\
        & \left(\alpha(\tilde{z}|\phi(0,x_0)) \sqrt{\PrM(\tilde{z}|\ket{\phi(\bar{0},\bar{x_0})})}\ket{0,x_0}+ \alpha(\tilde{z}|\phi(\bar{1},\bar{x_1})) \sqrt{\PrM(\tilde{z}|\ket{\phi(1,x_1)})}\ket{1,x_1}\right)_\mathsf{BX}  \ket{\tilde{z}}_\mathsf{Z}\\
    \end{aligned}
\end{equation}
where $c$ is a normalization constant, $x_0$ and $x_1 := x_0 - s$ form a claw of $f(b,x)$, hence $f(0,x_0)=f(1,x_1)$. It should be clear why $\ket{\psi_\mathrm{ideal}}$ is considered ideal, since the state in the $\mathsf{BX}$ register conditioned on having measured $\mathsf{Z}$, will be a superposition of the claw $((0, x_0), (1, x_1))$. This is due to the fact that $\maj(\tilde{z})=f(0,x_0)$ which ensures that the image $f(0,x_0)$ can be perfectly decoded. Hence, only the claw $((0, x_0), (1, x_1))$ will be consistent with this outcome of the image register.

We now show the following:

\begin{lemma}
\label{lemma:first_ideal}
$F(\ket{\tilde{\psi}},\ket{\psi_\mathrm{ideal}}) = |\braket{\tilde{\psi}|\psi_\mathrm{ideal}}|^2 > 0.98$.
\end{lemma}
\begin{proof}
Let us first give a lower bound of $c$, where recall that $c$ is the normalization constant in Equation~\ref{eqn:psiideal}. We showed in Theorem~\ref{thm:decodability} that at least $99\%$ of $\ket{\phi}$'s are decodable. In other words, we have
\begin{equation}
    \sum_{\maj(\tilde{z})=f(0,x_0)}\PrM(\tilde{z}|\ket{\phi(0,x_0)}) \geq 0.99
\end{equation}
and
\begin{equation}
    \sum_{\maj(\tilde{z})=f(1,x_1)}\PrM(\tilde{z}|\ket{\phi(1,x_1)}) \geq 0.99
\end{equation}
for at least $99\%$ possible $x_0$'s. Keeping in mind that $f(0, x_0) = f(1, x_1)$, the normalization condition leads to
\begin{equation}
    \frac{c^2}{2q^n} \sum_{x_0 \in \mathbb{Z}_q^n} \sum_{\maj(\tilde{z})=f(0,x_0)} \left( \PrM(\tilde{z}|\ket{\phi(0,x_0)}) + \PrM(\tilde{z}|\ket{\phi(1,x_1)}) \right) =1,
\end{equation}
which implies that
\begin{equation}
    1 < c^2 \leq \frac{2q^n}{0.99\cdot (0.99+0.99)q^n}=1.02.
\end{equation}

The fidelity can be computed as
\begin{equation}
\begin{aligned}
    F(\ket{\tilde{\psi}},\ket{\psi_\mathrm{ideal}}) &= |\braket{\tilde{\psi}|\psi_\mathrm{ideal}}|^2\\
    &\geq \left| \frac{1} {c}  \right|^2 > 0.98.
\end{aligned}
\end{equation}
\end{proof}

In the ideal state, every $\tilde{z}$ measurement outcome corresponds to exactly two $\ket{\phi(b,x)}$ states that form a claw of $f$. Supposing a specific $\tilde{z}$ is measured, the remaining post-measurement state in the $\mathsf{BX}$ register will be
\begin{equation}
\label{eq:post-state}
    \ket{\psi_{\tilde{z}}} \propto \sum_b  \alpha\left(\tilde{z}|\phi(b,x_b)\right) \sqrt{\PrM(\tilde{z}|\ket{\phi(b,x_b)})} \ket{b,x_b}.
\end{equation}
Recall that the honest prover would certainly succeed in the protocol with an equal superposition over the claw (without any relative phase between the components):
\begin{equation}
\label{eq:perfect_claw}
    \ket{\psi_{y}} \propto \sum_b \ket{b,x_b}.
\end{equation}
Unfortunately, the state in Equation~\ref{eq:post-state}, resulting from the measurement of $\ket{\psi_\mathrm{ideal}}$, is not of this form due to the presence of the phases $\alpha(\tilde{z}|\phi(\bar{b},\bar{x_b}))$ which could lead to a non-negligible relative phase. We now show that this relative phase is in fact close to zero. To do so, consider a ``more ideal state'' $\ket{\psi_\mathrm{ideal,2}}$:
\begin{equation}
    \begin{aligned}
        \ket{\psi_\mathrm{ideal,2}}& =  \frac{c'}{\sqrt{2q^n}} \sum_{x_0 \in \mathbb{Z}_q^n} \sum_{\maj(\tilde{z})=f(0,x_0)} \\
        &\left(\alpha(\tilde{z}|\phi(0,x_0)) \sqrt{\PrM(\tilde{z}|\ket{\phi(0,x_0)})}\ket{0,x_0}+ \alpha(\tilde{z}|\phi(0,x_0)) \sqrt{\PrM(\tilde{z}|\ket{\phi(0,x_0)})}\ket{1,x_1}\right)_\mathsf{BX} \ket{\tilde{z}}_\mathsf{Z},
    \end{aligned}
\end{equation}
where $c'\in\mathbb{R}$ is another normalization factor. Note that in this state the two components corresponding to the preimage register share the same phase, $\alpha(\tilde{z}|\phi(0,x_0))$, meaning that there is no relative phase.

We start by bounding the normalization constant $c'$ from the norm of the state:
\begin{equation}
\begin{aligned}
    1=\braket{\psi_\mathrm{ideal,2}|\psi_\mathrm{ideal,2}}&=\frac{c'^2}{2q^n} \sum_{x_0 \in \mathbb{Z}_q^n} \sum_{\maj(\tilde{z})=f(0,x_0)} \PrM(\tilde{z}|\ket{\phi(0,x_0)}) \left( \braket{0,x_0|0,x_0}+\braket{1,x_1|1,x_1} \right)  \braket{\tilde{z}|\tilde{z}}
\end{aligned}
\end{equation}
which implies that
\begin{equation}
    1<c'^2\leq \frac{2q^n}{0.99 q^n \cdot (0.99+0.99)}.
\end{equation}

It should be clear that if the prover holds $\ket{\psi_\mathrm{ideal,2}}$, it would succeed in the equation and preimage tests with $100\%$ probability. Thus, to calculate the success probability of the real prover in our protocol, we simply evaluate the fidelity between $\ket{\psi_\mathrm{ideal}}$ and $\ket{\psi_\mathrm{ideal,2}}$.

\begin{lemma}
\label{lemma:second_ideal}
$F(\ket{\psi_\mathrm{ideal}},\ket{\psi_\mathrm{ideal,2}}) = |\braket{\psi_\mathrm{ideal}|\psi_\mathrm{ideal,2}}|^2 > 0.97$.
\end{lemma}
\begin{proof}
\begin{equation}
\begin{aligned}
    |\braket{\psi_\mathrm{ideal} | \psi_\mathrm{ideal,2}}| & =  \frac{cc'}{2q^n} \left| \sum_{x_0 \in \mathbb{Z}_q^n} \sum_{\maj(\tilde{z})=f(0,x_0)} \right. \\
    &\left. \left[ \PrM(\tilde{z}|\ket{\phi(0,x_0)} + \alpha^* (\tilde{z}|\phi(1,x_1)) \alpha(\tilde{z}|\phi(0,x_0)) \sqrt{\PrM(\tilde{z}|\ket{\phi(1,x_1)}\PrM(\tilde{z}|\ket{\phi(0,x_0)}}  \right] \right|\\
    &\geq \frac{1}{2q^n}\left[ \left|0.99 q^n \cdot 0.99 + \sum_{x_0}\braket{\phi(1,x_1)|\phi(0,x_0)}\right| - 0.01q^n  \right],
\end{aligned}
\end{equation}
since
\begin{equation}
\braket{\phi(1,x_1)|\phi(0,x_0)}=\sum_{\tilde{z}\in\{0,1\}^{mv\log_2 p}} \alpha^* (\tilde{z}|\phi(1,x_1)) \alpha(\tilde{z}|\phi(0,x_0)) \sqrt{\PrM(\tilde{z}|\ket{\phi(1,x_1)}\PrM(\tilde{z}|\ket{\phi(0,x_0)}}.
\end{equation}

The inner product $\braket{\phi(1,x_1)|\phi(0,x_0)}$ can also be evaluated by considering their phase encoded form. We start with
\begin{equation*}
     \ket{\phi(b, x)} = \bigotimes_{i=1}^m \bigotimes_{k=1}^{\log_2 p} \left( \ket{\bar{0}} + e^{i\phi_{i,k}} \ket{\bar{1}} \right)^{\otimes v}
\end{equation*}
where $\phi_{i,k}(b,x) = \frac{2^k \pi g_i(b,x)}{q}-\frac{\pi}{2}$.
As both states are phase encodings, the inner product will be determined by the angle differences between the components. In other words, letting
\begin{equation}
\Delta \phi_{i,k} = \frac{2^k \pi (g_i(0,x) - g_i(1, x - s))}{q}
\end{equation}
and noting that $g_i(0, x) = (\mathbf{A}x)_i$ and $g_i(1, x - s) = (\mathbf{A}x)_i + e_i$, it is the case that
\begin{equation}
\Delta \phi_{i,k} = \frac{2^k \pi e_i}{q}.
\end{equation}

We can now express the inner product as
\begin{equation}
\begin{aligned}
\braket{\phi(0, x)| \phi(1, x - s)} &= \prod_{i, k} \left[ \exp\left(i\frac{\Delta \phi_{i,k}}{2}\right) \cos \left(\frac {\Delta \phi_{i,k}}{2} \right ) \right]^v \\
&=  \exp\left(i\sum_{i, k}\frac{\Delta \phi_{i,k}\cdot v}{2}\right) \prod_{i, k} \left( \cos \left( \frac {\Delta \phi_{i,k}}{2} \right ) \right)^v
\end{aligned}
\end{equation}
But now note that $\| e \|_{\infty}  \leq \frac{cq}{p^5}$, for some constant $c > 0$, as per Definition~\ref{def:lwrntcf}. If we substitute this into the formula for $\Delta \phi_{i, k}$, keeping in mind that $2^k \leq p$, we find that
\begin{equation}
\Delta \phi_{i,k} = \frac{2^k \pi e_i}{q} \leq \frac{\pi}{p^4}.
\end{equation}
Taking $n$ to be sufficiently large, so that $p$ is sufficiently large, leads to
\begin{equation}
\cos\left(\frac {\Delta \phi_{i,k}}{2}\right) \geq 1 - \frac{\pi^2}{8p^8} - O \left( p^{-16} \right)
\end{equation}
and
\begin{equation}
    \prod_{i, k} \left( \cos \left(\frac {\Delta \phi_{i,k}}{2} \right ) \right)^v \geq \left(1 - \frac{\pi^2}{8p^8} \right)^{m v  \log_2 p}.
\end{equation}
But now $p^8 = O \left( (mn \log q)^4 \right) = O(n^{16})$ and $mv \log_2 p = O(n^2 \cdot n^4 \log n \cdot \log n) = O(n^6 \log^2 n)$. It follows that
\begin{equation}
\prod_{i, k} \left( \cos \left( \frac {\Delta \phi_{i,k}}{2} \right ) \right)^v \geq 1 - \frac{1}{poly(n)}.
\end{equation}
For the phase part
\begin{equation}
    \sum_{i, k}\frac{\Delta \phi_{i,k}\cdot v}{2} = O\left(p^{-4} (m\log_2 p)^3 \log_2(m\log_2 p) \right) = O\left( \frac{(\log_2 n)^4}{n^2} \right)
\end{equation}
and similarly
\begin{equation}
\exp\left(i\sum_{i, k}\frac{\Delta \phi_{i,k}\cdot v}{2}\right) = 1-O\left( \frac{(\log_2 n)^8}{n^4} \right)+i O\left( \frac{(\log_2 n)^4}{n^2}\right)=1-\frac{1}{poly(n)}+i\cdot \frac{1}{poly(n)}.
\end{equation}
Finally,
\begin{equation}
\braket{\phi(0, x)| \phi(1, x - s)} = 1-\frac{1}{poly(n)}+i\cdot \frac{1}{poly(n)}
\end{equation}
and the fidelity can be lower-bounded as follows

\begin{equation}
\begin{aligned}
|\braket{\psi_\mathrm{ideal} | \psi_\mathrm{ideal,2}}|^2 &\geq \left[ \frac{1}{2q^n} \left( \left|0.99 q^n \cdot 0.99 + \sum_{x_0}\left(1-\frac{1}{poly(n)}+i\cdot \frac{1}{poly(n)} \right)\right| - 0.01q^n\right) \right]^2  \\
&=\left[ -\frac{0.01}{2} + \frac{1}{2} \left| 0.98+1-\frac{1}{poly(n)}+i\cdot \frac{1}{poly(n)} \right|  \right]^2\\
&=\left[ -\frac{0.01}{2}+\frac{1}{2}\sqrt{\left(1.98-\frac{1}{poly(n)}\right)^2 + \left(\frac{1}{poly(n)}\right)^2 } \right]^2\\
&=\left[ \frac{-0.01+1.98-\frac{1}{poly(n)}}{2} \right]^2 = \left[\frac{1.97-\frac{1}{poly(n)}}{2}\right]^2 > 0.97
\end{aligned}
\end{equation}

for large sufficiently $n$.
\end{proof}

Combining Lemmas~\ref{lemma:first_ideal} and~\ref{lemma:second_ideal}, we conclude that the success probability for an honest prover is lower bounded by $0.95$, using a union bound.

\subsubsection{Completeness}
We can now compute the probability for an honest prover, following the strategy outlined in Figure~\ref{fig:bcmvvprotocol2}, to pass the verifier's checks.
We start with the observation that $q$ is prime. As mentioned, this would require the prover to create a superposition in the preimage register of $q^n$ components. Instead, the prover creates a superposition of $q'^n$ components, where $q'$ is a power of 2 that is close to $q$. From the results in Subsection~\ref{sec:q_is_prime}, we incur a $O(n^{-1})$ penalty in the honest prover's success probability as a result of this. Next, we saw that when performing the measurement of the image register, there is a chance that the $\ket{\phi(b, x)}$ state contains components that are undecodable. 
We limited the probability of this happening to $1\%$, with the parameter choices mentioned in Subsection~\ref{sec:rep-and-majority}. Assuming the state is decodable, we saw that the probability of incorrectly decoding is also $1\%$. With these results, we showed in Subsection~\ref{subsect:fidelity} that the prover's state, upon measuring the image register (and successfully decoding the result, which is sent to the verifier), gives it at least a $95\%$ success probability in the equation and preimage tests. This also accounted for the failure probability of incorrectly decoding the image register.
Finally, as discussed in Subsection~\ref{sec:const-depth-phase-enc}, if we choose to use a fixed-size gate set, we will incur another $1/poly(n)$ error.

Putting everything together, we find that the overall completeness of the protocol is $95\% - O(n^{-1})$.

\subsubsection{Soundness}
Since we showed that the LWR-based function $f(b,x)$ is also an NTCF, in Subsection \ref{sec:lwr-ntcf}, our new constant quantum depth protocol inherits the soundness of the original BCMVV protocol.

\subsection{Resource estimation} \label{subsect:resources2}
As in Subsection~\ref{subsect:resources}, we summarize the resources required for an honest prover to succeed in the protocol.
\subsubsection{Quantum depth and quantum-classical interleavings}
\label{sec:forbigp}
\begin{enumerate}
    \item{\textbf{Preparation of cat states.}}
    Same as in the randomized encoding construction, the depth of this step is 5 and the prover interleaves constant-depth quantum computation and classical log-depth computation once.
    \item{\textbf{Evaluation of the LWR function by phase encoding.}}
    As is illustrated in Figure~\ref{fig:circuit_CRzs}, this step consists of only parallel $\mathrm{CR}_z$ gates or $\mathrm{R}_z(\frac{\pi}{2})$ gates. The depth added is only 1 for the example case in Figure~\ref{fig:circuit_CRzs}.
    \item{\textbf{Measurement of the $\mathsf{Z}$ register.}} As is explained in Subsection~\ref{sec:rep-and-majority}, the measurement of the $\mathsf{Z}$ register contains Hadamard measurements and a majority vote (performed classically on the measurement outcome), hence this step has quantum depth 2 and adds 1 step of quantum-classical interleaving.
    \item{\textbf{Preimage test/equation test.}} Exactly the same as in the BCMVV protocol, this step requires at most depth 2 and 1 interleaving for the equation test.
\end{enumerate}
In summary, the phase encoding construction requires even shorter quantum depth than the generic construction, as the overall quantum depth is $5+1+2+2=10$. The number of quantum-classical interleaving is 3, same as the generic construction.

\subsubsection{Circuit width}
The total width of the circuit is determined by the product of several multipliers in the protocol:
\begin{enumerate}
    \item{\textbf{Number of output components of $g(b,x)$}} is $O(m) = O(n^2)$, by definition.
     \item{\textbf{$\lfloor \cdot \rfloor_p$ rounding function.}}  The phase encoding needs to be prepared for all of the $\log p$ bits. This leads to another $O(\log (\sqrt{mn\log q})) = O(\log m) = O(\log n)$ multiplier.
    \item{\textbf{Cat state.}} As discussed in Subsection~\ref{sec:const-depth-phase-enc}, the size of the cat state for each component $\ket{\phi_i}$ needs to be $O(n\log q)=O(n^2)$.
    \item{\textbf{Repetition for majority votes.}} This is calculated in Subsection \ref{subsect:generalp} and each $\ket{\phi_i}$ needs to be repeated for $v=O(n^4 \log^2 n)$ times.
  
\end{enumerate}
In summary, the total circuit width required is $O(n^8 \log^3 n)$. Although this is still a high-order polynomial, it is a significant improvement over the randomized encoding construction (where we estimated $O(n^{33})$ width). Note that the normal, poly-depth, construction requires $O(m\log q)=O(n^3)$ width.

It is also worth mentioning that there can be a trade-off between the size of the cat states and the depth of the circuit, since the matrix multiplication does not need to be fully parallelized. In practice, one can double the number of $CR_z$ gates applied on each qubit to halve the width.

\subsection{Robustness against noise}
Another feature of our phase encoding construction is some amount of intrinsic robustness against noise, which makes it closer to practical use on near-term devices.

The key reasons for the noise-resistance are the use of cat states, the classical repetition code we applied in measuring the $\mathsf{Z}$ register, as is discussed in Subsection~\ref{sec:rep-and-majority}, the error-correcting properties of the LWR construction which we used implicitly in Subsection~\ref{subsect:fidelity} and the constant gap between the best quantum strategy and the best classical strategy (assuming intractability of LWE) as encapsulated by Inequality~\ref{ineq:winntcf}.

We can therefore see that errors on the image register, $\mathsf{Z}$, may lead to bit flips of the output string $z$ such that $z\neq y$ (where recall that $y$ is the ideal decoding). However, since any bit $z_i$ is determined by majority voting for all $v$ repetitions of the phase encoding of that bit, the probability that $z_i$ is flipped is much smaller than that of single bit flipping. Intuitively speaking, some correctly measured bits may be flipped due to noise that might appear in any stage of the protocol, but incorrect bits are equally likely to be flipped. Hence the majority vote will still very likely output $z_i=y_i$.

A repetition code is also used indirectly in the preimage register, as the preimages are encoded in cat states. While this makes the preimage test robust to noise, the equation test will not be, in general. 
This is because in the equation test, the prover needs to report a string $d$ and a bit $b$ such that
\begin{equation}
d \cdot (\bar{x}_0 \oplus \bar{x}_1) = b
\end{equation}
where $\bar{x}_0$ and $\bar{x}_1$ are the repetition code encodings of preimages $x_0$ and $x_1$ (that match the image the prover returned in the previous round of the protocol). In this case we can see that even a single bit flip in either the string $d$ or of the bit $b$ can make the equation invalid.
We therefore leave it as an open problem to find a fully noise-robust implementation of the protocol.

\bibliography{mainbib}

\appendix
\section{Randomized encoding construction from~\cite{AIK06}} \label{app:bre}
The construction of randomized encodings from~\cite{AIK06} is based on \emph{branching programs}. 
We are only interested in mod-2 branching programs, which we define here:

\begin{definition}[Branching programs \cite{AIK06}]
A branching program (BP) is defined by a tuple $BP=(G,\phi,s,t)$ where $G=(V,E)$ is a directed acyclic graph, $\phi$ is a labeling function assigning each edge either a positive literal $x_i$ or a negative literal $\neg{x_i}$. An input binary vector $\vec{w}$ determines a subgraph $G_w$ where an edge labeled as $x_i$ is preserved if and only if $w_i=1$. In a (counting) mod-2 BP, the BP computes the number of paths from $s$ to $t$ modulo 2. The size, $l$, of a BP is defined as the number of vertices, $|V|$.
\end{definition}

\noindent As an example, Figure \ref{fig:bp} shows a mod-2 branching program of size $l=4$ and having three inputs $x=(x_0,x_1,x_2)$.
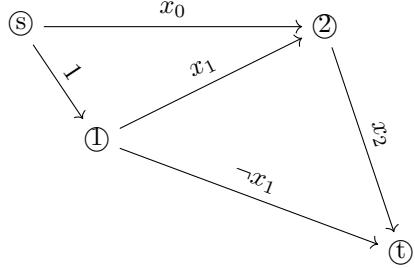
\begin{figure}[h!]
    \centering
    \begin{tikzpicture}
        \node (p_s) at (0,3) {$\textcircled{s}$};
        \node (p_t) at (5, 0) {$\textcircled{t}$};
        \node (p_1) at (1,1.5) {$\textcircled{1}$};
        \node (p_2) at (4,3) {$\textcircled{2}$};
        \draw [->] (p_s) -- (p_1) node [midway, above, sloped] (TextNode) {$1$};
        \draw [->] (p_s) -- (p_2) node [midway, above, sloped] (TextNode) {$x_0$};
        \draw [->] (p_1) -- (p_2) node [midway, above, sloped] (TextNode) {$x_1$};
        \draw [->] (p_1) -- (p_t) node [midway, above, sloped] (TextNode) {$\neg{x}_1$};
        \draw [->] (p_2) -- (p_t) node [midway, above, sloped] (TextNode) {$x_2$};
    \end{tikzpicture}
    \caption{This size-4 mod-2 branching program consists of 5 edges whose connectivity is decided by the value of the input bits. Note that $\neg{x}_1$ means that this edge is available if and only if $x_1=0$. As an example, when the input $x=(x_0,x_1,x_2)=(0,1,1)$, there is only one path from $s$ to $t$ which is $\textcircled{s}-\textcircled{1}-\textcircled{2}-\textcircled{t}$. Thus the output of this mod-2 BP will be 1.}
    \label{fig:bp}
\end{figure}

We now state one of the most important results concerning branching programs, due to Barrington:
\begin{theorem}[Barrington's theorem \cite{barrington1989bounded}]
\label{barrington}
If $f : \{0, 1\}^n \rightarrow \{0, 1\}$ can be computed by a circuit of depth $d$, then it can be computed by a branching program of width 5 and length $O(4^d)$.
\end{theorem}
\noindent The above theorem ensures that the log-depth (N)TCFs used in proof of quantumness protocols can be transformed into polynomial-size branching program. 
Given that branching programs output a single bit, this construction has to be performed for every output bit of a (N)TCF.

A size-$l$ mod-2 BP for a binary function $f$ can be represented by an adjacency matrix since BPs are directed acyclic graphs. Let $A(x)$ denote the $l \times l $ adjacency matrix of a BP with input $x$. We also denote as $L(x)$ the $(l-1)\times(l-1)$ submatrix of $A(x) - I$ obtained by deleting the first column and the last row. It turns out that the following fact holds:
\begin{lemma}[\cite{AIK06}]
$f(x)=\det(L(x)) \; mod \; 2.$
\end{lemma}
This lemma is the basis for constructing a randomized encoding for $f$. The goal will be to ``garble'' $L(x)$ through products with certain random matrices. The garbling should be done in such a way that the determinant of the resulting matrix matches that of $L(x)$, thus preserving the correctness of the construction.

To that end, let $r^{(1)} \leftarrow_R \{0, 1\}^{{l-1 \choose 2}}$ and $r^{(2)} \leftarrow_R \{0, 1\}^{l-2}$. Use these to construct matrices $R^{(1)}$ and $R^{(2)}$ of dimensions $(l-1)\times(l-1)$. Both matrices have all diagonal elements equal to $1$. The right upper-diagonal elements of $R^{(1)}$ (that is, the entries $R^{(1)}_{i,j}$ with $j > i$) are filled with the entries of $r^{(1)}$. The last column of $R^{(2)}$, except for the last element, (that is, the entries $R^{(2)}_{i, l-1}$, $1 \leq i \leq l - 2$) is filled with the elements of $r^{(2)}$. All other entries of $R^{(1)}$ and $R^{(2)}$ are $0$. The following can be shown:
\begin{lemma}[\cite{AIK06}]
$\det(L(x))=\det(R^{(1)} L(x) R^{(2)})$
\end{lemma}

This is not too difficult to see, as both $R^{(1)}$ and $R^{(2)}$ have determinant $1$. 
One now defines the randomized encoding $\tilde{f}(x, r^{(1)}, r^{(2)}) = R^{(1)} L(x) R^{(2)}$. It follows that:
\begin{lemma}[\cite{AIK06}]
$\tilde{f}$ is a perfect randomized encoding of $f$.
\end{lemma}

By construction, every entry of $\tilde{f}$ is a degree-$3$ polynomial in its input variables. However, computing this function (i.e. computing every matrix entry of $R^{(1)} L(x) R^{(2)}$) cannot be done in constant-depth. The reason is that some of the input variables are involved in a linear number of monomials of the output. To compute the function in constant depth, it must be that each input variable appears in only a constant number of monomials. The authors of~\cite{AIK06} remedy this by considering a randomized encoding for $\tilde{f}$. Before doing so, note that

\begin{lemma}[\cite{AIK06}]
The composition of perfect randomized encodings is still a perfect randomized encoding of the original function.
\end{lemma}

\noindent Thus, a randomized encoding for $\tilde{f}$ will also be a randomized encoding for $f$. Denote the $i, j$ entry of $\tilde{f}$ as $\tilde{f}_{i, j}$. We can see that
\begin{equation} \label{appeq:tildef}
\tilde{f}_{i,j}(x, r^{(1)},r^{(2)})=T_1(x, r^{(1)},r^{(2)}) \oplus T_2(x, r^{(1)},r^{(2)}) \oplus ... \oplus T_k(x, r^{(1)},r^{(2)})
\end{equation}
where each $T_m$ is a monomial in the input variables. Finally, define $\hat{f}$ as
\begin{equation} \label{appeq:hatf}
    \hat{f}_{i,j}(x, r^{(1)},r^{(2)},r,r^\prime) = (T_1 \oplus r_1,T_2 \oplus r_2,...,T_k \oplus r_k,r_1 \oplus r_1^\prime,r_1^\prime \oplus r_2 \oplus r_2^\prime,...,r_{k-1}^\prime  \oplus r_k)
\end{equation}
where $r$ and $r^\prime$ are newly introduced vectors of random bits. Note that adding all entries in~\ref{appeq:hatf} results in the summation from Equation~\ref{appeq:tildef}. Thus, $\hat{f}$ contains all of the information required to compute $\tilde{f}$ and moreover,
\begin{lemma}[\cite{AIK06}]
$\hat{f}$ is a perfect randomized encoding of $f$ with output locality $4$.
\end{lemma}
Here, output locality $4$ means that each output bit depends on at most $4$ input bits, which immediately implies that the function can be evaluated in constant depth.
The classical circuit computing an entry of $\hat{f}$ is shown in Figure \ref{fig:classical-circuit}. Detailed proofs of all these results can be found in \cite{AIK06}. 

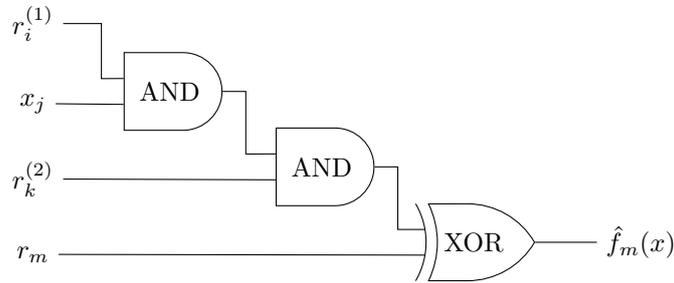
\begin{figure}[h!]
    \centering
\begin{tikzpicture}

    \node (r1) at (0, 3) {$r^{(1)}_i$};
    \node (xj) at (0, 1.94) {$x_j$};
    \node (r2) at (0, 0.94) {$r^{(2)}_k$};
    \node (rl) at (0, -0.06) {$r_m$};

    \node[and gate US, draw, rotate=0, logic gate inputs=nn] at (1.8,2.1) (and1) {AND};
    
    \draw (r1) -- (0.9,3)--(0.9,2.27)-- (and1.input 1);
    \draw (xj) -- (and1.input 2);
    
    \node[and gate US, draw, rotate=0, logic gate inputs=nn] at (3.8,1.1) (and2) {AND};
    
    \draw (and1.output) -- (2.8,2.1)--(2.8,1.27)-- (and2.input 1);
    \draw (r2) -- (and2.input 2);

    \node[xor gate US, draw, rotate=0, logic gate inputs=nn] at (5.8,0.1) (xor1) {XOR};
    
    \draw (and2.output) -- (4.8,1.1)--(4.8,0.27)-- (xor1.input 1);
    \draw (rl) -- (xor1.input 2);
    
    \node (fhat) at (8,0.1) {$\hat{f}_m(x)$};
    
    \draw (xor1.output) -- (fhat);
\end{tikzpicture}
    \caption{The circuit for evaluating each entry in the randomized encoding $\hat{f}$. The circuit shown here computes the $m$'th entry, with $m \leq k$, consisting of the monomial $r_i^{(1)} x_j r_k^{(2)}$ xored with $r_m$. For the entries with $m > k$, note that a single XOR gate is required.}
    \label{fig:classical-circuit}
\end{figure}

\section{Reconstruction of randomness}
\label{app:rrc}
In our first constant quantum-depth proof of quantumness, the prover is instructed to evaluate a randomized encoding of a TCF. The verifier must still be able to use the trapdoor in order to invert an output of the randomized encoding. As mentioned in Subsection~\ref{subsect:soundness}, this is true provided the encoding satisfies the randomness reconstruction property. Here we prove this fact for the construction of~\cite{AIK06}.

\begin{proof}[Proof of Lemma~\ref{lem:rrc}]
We would like to show that given an instance of $\hat{f}_{i,j}(x, r^{(1)},r^{(2)},r,r^\prime)$, as shown in Equation~\ref{appeq:hatf}, as well as $x$, it is possible to efficiently recover the randomness $r^{(1)},r^{(2)},r,r^\prime$. First note that if the terms $T_k$ were known as well as $r^{(1)},r^{(2)}$, it is straightforward to recover $r$ and $r'$. We will therefore focus on that case.
From Equation~\ref{appeq:hatf} it is possible to efficiently compute the result of Equation~\ref{appeq:tildef}, since $\hat{f}$ is a randomized encoding of $\bar{f}$: simply xor all the terms in Equation~\ref{appeq:hatf}. We will then focus on randomness reconstruction for $\bar{f}$ as that will then yield randomness reconstruction for $\hat{f}$.

Denote as $M = \tilde{f}(x, r^{(1)}, r^{(2)}) = R^{(1)} L(x) R^{(2)}$. Given $M$ and $x$ we wish to recover $r^{(1)}, r^{(2)}$. This boils down to solving a specific quadratic system of equations. To see why, take $l=4$ as an example,
\begin{equation*}
M = R^{(1)}L(x)R^{(2)}=
 \left[
 \begin{matrix}
   1 & r^{(1)}_1 & r^{(1)}_3 \\
   0 & 1 & r^{(1)}_2 \\
   0 & 0 & 1
  \end{matrix}
  \right]
   \left[
 \begin{matrix}
   x_1 & x_4 & x_6 \\
   -1 & x_2 & x_5 \\
   0 & -1 & x_3
  \end{matrix}
  \right]
   \left[
 \begin{matrix}
   1 & 0 & r^{(2)}_1 \\
   0 & 1 & r^{(2)}_2 \\
   0 & 0 & 1
  \end{matrix}
  \right]
  \end{equation*}
  
  \begin{equation*}
M =
\left[
 \begin{matrix}
 x_1 - r^{(1)}_1 & r^{(1)}_1 x_2 - r^{(1)}_3 + x_4 & r^{(2)}_1 (x_1 - r^{(1)}_1) + r^{(2)}_2 (r^{(1)}_1 x_2 - r^{(1)}_3 + x_4) + r^{(1)}_3 x_3 + r^{(1)}_1 x_5 + x_6\\
-1 & x_2 - r^{(1)}_2 & r^{(2)}_2 (x_2 - r^{(1)}_2) + r^{(1)}_2 x_3 - r^{(2)}_1 + x_5 \\
0 & -1 & x_3 - r^{(2)}_2
  \end{matrix}
  \right]
\end{equation*}
Note that the main diagonal of $M$ is just a linear system of $3$ equations with $3$ unknowns. It can therefore be solved, yielding $r_1^{(1)}$, $r_2^{(1)}$ and $r_2^{(2)}$. 
Plugging these values into the second diagonal (the one above the main diagonal), yields another system of linear equations with an equal number of unknowns. By repeating the process and solving all of these systems, all bits in $r^{(1)}$ and $r^{(2)}$ are recovered. 

We now show that this strategy works for arbitrary $l$.
Start by observing that:
\begin{equation*}
\left\{
             \begin{array}{lr}
             R^{(1)}_{i,j}=1, & i=j \\
             R^{(1)}_{i,j}=0, & i>j,
             \end{array}
\right.
\end{equation*}

\begin{equation*}
\left\{
             \begin{array}{lr}
             L_{i,j}=-1, & i=j+1 \\
             L_{i,j}=0, & i>j+1,
             \end{array}
\right.
\end{equation*}

\begin{equation*}
\left\{
             \begin{array}{lr}
             R^{(2)}_{i,j}=1, & i=j \\
             R^{(2)}_{i,j}=0, & (i>j) \vee (i<j<l-2).
             \end{array}
\right.
\end{equation*}

\noindent The entries of $M$ can then be expressed as:
\begin{equation*}
    M_{i,j}=\sum_{k_1,k_2}R^{(1)}_{i,k_1}L_{k_1,k_2}R^{(2)}_{k_2,j}.
\end{equation*}
Consider the entries on the main diagonal, excluding the last element:
\begin{equation*}
    M_{i,i}=\sum_{k_1}R^{(1)}_{i,k_1}L_{k_1,i}
    =R^{(1)}_{i,i}L_{i,i}+\sum_{k_1>i}R^{(1)}_{i,k_1}L_{k_1,i}
    =L_{i,i}-R^{(1)}_{i,i+1}
\end{equation*}
with $i < l - 2$ and where $R^{(1)}_{i,i+1}$ are the elements of the second diagonal of $R^{(1)}$ and the $L_{i,i}$'s are already known (as they only involve entries of $x$). This gives us a simple linear system which we can solve to recover the $R^{(1)}_{i,i+1}$ values.
Then, for $i=l-2$:

\begin{equation*}
    M_{l-2,l-2}=\sum_{k_2}L_{l-2,k_2}R^{(2)}_{k_2,l-2}
    =L_{l-2,l-3}R^{(2)}_{l-3,l-2}+L_{l-2,l-2}R^{(2)}_{l-2,l-2}
    =R^{(2)}_{l-3,l-2} + L_{l-2,l-2}.
\end{equation*}
From this we also recover $R^{(2)}_{l-3,l-2}$, i.e. the last entry in $r^{(2)}$. 
Note that the unknowns here consisted of the entries in the second diagonal of $R^{(1)}$ and the last element of $r^{(2)}$. This matches the number of equations and so all values could be recovered.

We now claim that the $k$'th diagonal of $M$ is a linear system which depends \emph{only} on the $k+1$ diagonal of $R^{(1)}$ and the $k$'th last element of $r^{(2)}$ given the solutions to the previous $k-1$ diagonals of $M$. 
Writing out the elements, we have:
\begin{equation*}
    M_{i,i+j} = \sum_{k_1,k_2}R^{(1)}_{i,k_1}L_{k_1,k_2}R^{(2)}_{k_2,i+j}.
\end{equation*}
with $j = k - 1$. For $i+j\neq l-2$:
\begin{equation*}
    M_{i,i+j} = \sum_{k_1}R^{(1)}_{i,k_1}L_{k_1,i+j}
    = -R^{(1)}_{i,i+j+1} + L_{i,i+j} + \sum_{i<k_1<i+j+1}R^{(1)}_{i,k_1}L_{k_1,i+j}
\end{equation*}
where the first term is from the $(k+1)$'th diagonal of $R^{(1)}$ and the remaining terms are known from solving the equations for the previous diagonals. Thus, we have a linear system, which we can solve, with unknowns comprising the elements of the $(k+1)$'th diagonal of $R^{(1)}$.

For $i+j = l-2$:
\begin{equation*}
    M_{l-2-j,l-2} = \sum_{k_2=l-2-j-1}^{l-2}L_{l-2-j,k_2}R^{(2)}_{k_2,l-2}
    +\sum_{k_1=l-2-j+1}^{l-2}R^{(1)}_{l-2-j,k_1} \sum_{k_2=k_1-1}^{l-2}L_{k_1,k_2}R^{(2)}_{k_2,l-2}.
\end{equation*}
The first term is a linear combination of the last $k+1$ entries of $R^{(2)}$, i.e. the last $k$ elements of $r^{(2)}$, and only the $k$'th element is unknown. The remaining terms are known from solving the systems corresponding to the previous diagonals.

We can therefore proceed in this fashion, starting from the first diagonal of $M$ and going upwards solving all systems of linear equations and thus recovering all values of $r^{(1)}$ and $r^{(2)}$. This procedure is clearly efficient and we have shown that it is also correct.
To conclude the proof, we also need to make sure that there is a unique solution to the system. This is guaranteed by the unique randomness property of the randomized encoding (Theorem~\ref{thm:uniquerand}).
\end{proof}

\end{document}